\documentclass[12pt]{llncs}
\usepackage[neveradjust]{paralist}

\usepackage{times}
\usepackage{geometry}

\usepackage{amsmath}
\usepackage{amssymb}
\usepackage{amsfonts}
\usepackage{graphicx}
\usepackage{subfigure}
\usepackage{xcolor}

\usepackage{url}
\usepackage{hyperref}
\definecolor{darkblue}{rgb}{0,0,.5}
\hypersetup{colorlinks=true, breaklinks=true, linkcolor=darkblue, menucolor=darkblue, urlcolor=darkblue, citecolor=darkblue}

\usepackage[ruled,vlined,linesnumberedhidden]{algorithm2e}

\renewcommand{\epsilon}{\varepsilon}

\newcommand{\diff}[1]{\frac{\partial}{\partial #1}}

\newcommand{\OPT}{\mathrm{OPT}}

\usepackage[neveradjust]{paralist}

\usepackage{times}
\usepackage{geometry}

\usepackage{amsmath}
\usepackage{amssymb}
\usepackage{amsfonts}
\usepackage{graphicx}
\usepackage{subfigure}
\usepackage{xcolor}

\renewcommand{\epsilon}{\varepsilon}

\newtheorem{theor}{Theorem}
 \newtheorem{lem}[theor]{Lemma}

 \newtheorem{cor}[theor]{Corollary}
 \newtheorem{fa}[theor]{Fact}
\begin{document}

\title{Satisfiability thresholds beyond $k-$XORSAT}
\author{Andreas Goerdt, Lutz Falke}
\institute{ Technische Universit\"at Chemnitz,
Fakult\"at f\"ur Informatik \\
Stra{\ss}e der Nationen 62, 09107 Chemnitz, Germany \\
\email{ \{goerdt, falu\}@informatik.tu-chemnitz.de ,  \\
http://www.tu-chemnitz.de/informatik/TI/}}

\maketitle

\begin{abstract} We consider random systems of equations 
$x_1+ \dots  + x_k=a,$ $0 \le a \le 2$ which are interpreted as 
equations modulo $3.$ We show for $k \ge 15$ that the satisfiability threshold 
of such systems occurs  where the $2-$core has density 
$1.$ We show a similar result for random uniquely extendible constraints
over $4$ elements. Our results extend previous results of 
Dubois/Mandler for equations $\mod 2$ and $k=3$ and Connamacher/Molloy for
uniquely extendible constraints over a domain of  $4$ elements with $k=3$ arguments. 

Our proof technique is based on variance calculations, using a 
technique introduced Dubois/Mandler.  
However, several additional observations (of independent interest) are necessary.   
\end{abstract}

\section{Introcuction} 

\subsection{Contribution}
Often constraints are equations of the type $f(x_1, \dots , x_k)=a$ 
where $a$ is an element 
of the domain considered and $f$ 
is a $k-$ary  function on this 
domain, for example addition of $k$ elements. 
Given a formula, which is a conjunction of $m$   constraints over $n$ variables
we want to find a solution. 
It is natural to assume that $f$ has the property:
Given $k-1$ arguments we can always  set the last argument such,  
that the constraint becomes true. 
In this case we can 
restrict attention to the $2-$core.
It is obtained by 
iteratively deleting all variables which occur
at most once. 
Thus it  is the maximal 
subformula in which each variable occurs at least twice.  

We consider the random instance $F(n, p):$  
Each equation over 
$n$ variables 
is picked 
independently with 
probability $p;$  
the domain 
size $d$ and 
the number of slots per 
equation $k$ is fixed. 
We consider the case $p=c/n^{k-1}$  
and the number of constraints is 
linear in $n$ whp. 
(with high probability, that is 
probability $1\,-\,o(1),  n$ large. )
The density of a formula is equal to the 
number of equations divided by the number of variables. 
The following is well known:
\begin{fa} [\cite{MOLLOY}] \label{Core}
1. Conditional on the number of variables $n'$
and equations $m'$ of the $2-$core the $2-$core is a 
uniform random member of all formulas where each variable 
occurs at least twice. \\
2. There exist $n'=n'(c)$ and $m'=m'(c)$ 
such that the number of variables of the 
$2-$core is  $n'(1+o(1))$ and the number of equations 
$m'(1+o(1))$ whp.    \\    
3. There exists a $T$ such that  whp. for $c \le T-\varepsilon$ 
the $2-$core has density $\le 1-\varepsilon$ and for $c \ge  T+\varepsilon$ the $2-$core has density $\ge 1+ \varepsilon.$ 
$T$ is determined as the solution of an analytical equation. 
\end{fa}

The expected number of solutions  of the $2-$core is $d^{n-m} ,n$ the number of variables, 
$m$ the number of equations. When the $2-$core has density $\ge 1+ \varepsilon$ 
whp. no solution exists. This holds in particular when 
the density of  $F(n, p)$ itself is  $\ge 1+\varepsilon.$ The formulas considered here always 
have density $<1.$ 
In seminal work Dubois and Mandler \cite{DUMA} consider   
equations  $\mod 2:$ $x_1 \, +\, \dots \, + \,x_k=a,$ $0 \le  a\le 1, k=3.$
They show satisfiability whp. when the $2$-core has density $\le 1-\varepsilon.$ 
For larger $k\ge 15$ a  full proof  for this result is given in 
\cite{MITZ}, Appendix C . Thus $T/n^{k-1}$ is the threshold for unsatisfiability
in this case.   

It is a natural conjecture that the same threshold applies to 
equations as discussed initially (and to some other types.) 
However, it seems difficult
to prove the conjecture in  some generality. 
One of the difficulties seems to be  that we have $2$  parameters 
$k$ and $d.$ 
We make some progress towards this conjecture.
We show it for equations $\mod 3. $ 
(The result is for $k>15,$ but we think it mainly technical to get it for all $k \ge 3.$)
\begin{theor} \label{mod3}
Let $F(n, p)$ be the  random  set of equations $\mod 3:$     $ x_1 + \dots + x_k \,=\, a , \, 0 \le a \le 2,$ 
$x_1+ \dots  +x_k$  an ordered  $k-tuple$ of variables.  
If $p < (T-\varepsilon)/n^{k-1}$ $F(n, p)$ is satisfiable whp. for $k>15.$ 
\end{theor}
The main task is to show that 
a $2-$core of density $\le 1-\varepsilon$ 
has a solution with probability $>\varepsilon>0$ . Our proof starts as 
Dubois/Mandler: Let $X$ be
the number of satisfying assignments of the $2-$core.
Its expectation is $ \ge d^{\varepsilon n}, d=3.$
We show that $E[X^2] \le O(E[X])^2.$ This implies
(by Cauchy-Schwartz (or Paley-Zygmund) inequality)
that the probability to have a solution is
$\ge \varepsilon >0.$  By Fact \ref{Core} $F(n, p)$ has a solution with the same probability. 
We apply Friedgut-Bourgain's Theorem  to $F(n, p)$ 
to show that unsatisfiability has a 
sharp threshold. By this the probability  becomes 
$1-o(1).$ In \cite{CRDA} Friedgut-Bourgain is applied to
the $\mod 2-$case. 
It seems that  our  proof for the $\mod 3-$case  
is somewhat simpler
(and applies  to the $\mod 2-$case and 
other cases.)

To determine  $E[X^2]$ Dubois/Mandler
apply  Laplace Method (one ingredient: bounding a sum through its maximum term.) 
The main difficulty is to bound a  real 
function of several arguments  from above. 
They show that their function has only one local maximum. 
We proceed by the same  method, but  substantial changes are necessary for $ k>3.$  

First, we observe (cf. \cite{MITZ}, Appendix C) 
that the function in question is $\le$  the {\it infimum } with respect to  certain other
parameters. This is based on generating functions: If $f(x)=\sum c_kx^k$ then $c^k \le f(a)/a^k,\, a>0, c_i\ge 0$ 
(a method rarely used in the area, a notable exception is \cite{PUY}.) 
Thus to bound the maximum from above 
we need to find  suitable
parameters and show that the value 
with respect to these parameters is less than the required upper bound .
(This leads to involved, but elementary calculus. )

To make this approach work we need appropriate generating functions:   
$X=X_{a_1}+ \dots X_{a_{3^n}},$  where $X_{a_i}$ is the indicator 
random variable  of the event 
that assignment $a_i$ makes the formula true. 
Then $X^2= \sum_a \sum_b \, X_aX_b.$ To get 
$E[X^2]$ we need to determine 
Prob$[X_aX_b=1].$
To this end we observe that the  equation $x_1+ \dots + x_k=c$  which is true  
under $a$ is true  exactly under those  assignments $b$ such that 
$0k_0 + 1k_1 +2k_2 = 0 \mod 3, $ and  $k_i$ is  the number of slots of 
$x_1+ \dots  +x_k$   filled with a variable $x$ with $b(x)=a(x)+i.$ 
Thus there are $\sum_{k_1=k_2 \mod 3} {k \choose k- k_1-k_2, k_1, k_2}$ different
ways in which an equation can become true under $a, b.$
The following generating function allows us to deal with these 
possibilities analytically. With  $\bf  w_1 \,= \, \exp(2\pi \i/3)$ the primitive third root of unity
and  $ \bf  w_2\,=w_1^2$ we define
$r(x_0, x_1, x_2) = \frac{1}{3} \left[ (x_0\,+\,x_1\,+\,x_2 )^k \,+  \, 
(x_0 \, + \,{ \bf w_1} x_1 + {\bf  w_2} x_2)^k \,+ \, (x_0 \,+ \,{ \bf w_2} x_1 \,+ \, 
{\bf w_1}  x_2)^k \right]$ \\
then  Coeff$[x_1^{k_1}x_2^{k_2}, r(1,\, x_1, \,  x_2)]= {k \choose k-k_1-k_2, k_1, k_2}$
if $k_1=k_2\mod 3$ and $0$ otherwise (easy from properties  $ \bf  w_j.)$  In the $\mod 2-$case we use 
$1/2\left[(1+x)^k \,\, + \,\, (1-x)^k\right]$ instead  \cite{MITZ}, Appendix C.\\

With the motivation to get an exact threshold of unsatisfiability for
a type of constraint whose worst-case complexity is NP-complete,
Connamacher/Molloy  \cite{COMO} see also the very recent \cite{CO} introduce uniquely extendible constraints.
A $k-$ary  uniquely extendible constraint  is a function 
from $D^k$ to true, false with the property: Given values from $D$ for 
any $k-1$ argument slots there is exactly one 
value for the remaining slot which makes the constraint true.  
(The $k> 8$ in the following result can be eliminated at the price of some additional technical effort.) 
\begin{theor} \label{uni}
Let $F(n, p)$ be the random formula of uniquely extendible constraints:
Each constraint is a  random $k-$tuple of variables and
a $k-$ary uniquely extendible constraint over $D$ and we pick with probability $p.$   
For $|D|=4 $ and $p < (T- \varepsilon)/n^{k-1}$ $F(n, p)$ is satisfiable whp. for $k>8.$ 
\end{theor}
The threshold $T/n^{k-1}$ is proved   for $k=3$ and $|D|=4$, cf.  \cite{CO} remark following Theorem 8. 
Our proof uses the technique  as in the 
$\mod 3-$case, however the details are different.  
One of the  contributions making is the generating polynomial \\   
$ p(x)=\frac{1}{d} \left[(1+x)^k+ (d-1)(1-\frac{x}{d-1})^k \right],$
as $r(x_0, x_1, x_2)$ above, not used before.

\subsection{Motivation}
Many computational problems can be naturally formulated as conjunctions of constraints. 
And we are interested to find a solution of this conjunction. 
Algorithmic properties of these conjunctions are considered in 
theoretical research (with remarkable results e. g.  in the realm of approximation\cite{HA})
and applied research, e. g. \cite{MEI}. An additional aspect is the investigation of 
conjunctions of randomly picked constraints; \cite{MO} is a fundamental study here.  
Propositional formulas in $k$-conjunctive normalform 
provide  an example  which has lead to a rich literature e. g. \cite{DIAZ}.
One of the characteristic properties of this research is that
its findings can often be related to experimental work by running algorithms
on randomly generated instances. 

One of the aspects of random formulas is a threshold phenomenon:
If the number of constraints of a conjunction picked is less than a threshold value
the conjunction is typically satisfiable, if it is more
we get unsatisfiability whp. Moreover instances picked close to
the threshold seem to be algorithmically hard, thus 
being candidate test cases for  algorithms.  
The threshold phenomenon and the possibility to investigate it by 
experiments causes  physics to become interested in the area e. g.  \cite{MEZE}. 
On the other hand, physical approaches lead to 
new algorithms and  classical theoretical computer science research, e. g. \cite{COPRI}. 

One of the major topics is to determine
the value of the threshold in natural cases. 
A full solution even in the natural $k-$CNF SAT case 
has not been obtained, but many partial results, \cite{DIAZ} for $k=3.$ 
Note that $k-$CNF does not have the unique extendability property as 
possessed by the constraints considered here. 
And it seems to be a major open problem to get the precise threshold for 
constraints without unique extendibility and not similar to $2-$CNF. 
A mere existence result is the 
Friedgut-Bourgain theorem \cite{FRI}. 
Based on this theorem thresholds for formulas of constraints over 
domains with more than $2$ elements are considered in 
\cite{MO}. Ordering constraints are considered in \cite{GO},
only partial results towards a threshold  can be proven. 
In order to get definite threshold results further techniques are required. Therefore
it is a useful effort to further develop the techniques with
which thresholds can be proven. This is the  general contribution  of this paper. 

A notable early exception, in that the precise threshold can be proven
is the $\mod 2-$case considered above. Historically \cite{DUMA}  is 
the first paper which uses variance calculation based on 
Laplace method in this area. Subsequently, for $k-$CNF SAT this method has lead to substantial progress 
in \cite{ACMO}.  The contribution 
here is that $\mod 2-$proof can be 
refined and extended to cover other 
cases based on  observations of 
independent interest. 
Note that random sparse linear systems over finite fields 
are used to construct error correcting codes, 
e. g. \cite{LMSS} or \cite {RU}, motivating the $\mod 3-$case.  A very recent study of the $\mod 2-$ case is 
\cite{ACHLIO}. More literature can be found in  \cite{KO}, 
but precise threshold results have not been obtained.

\subsection{Contents}
\begin{tabular}{l l}
I. Equations modulo $3$ &             \\
\quad 1. Notation and basics &   \\
\quad 2. Outline of the proof of Theorem \ref{EX2EI} & \\
\quad 3. Proof of Theorem \ref{OPT}  &  \\
\quad \quad \quad 3.1 Proof of Lemma \ref{lem1} & \\
\quad \quad \quad 3.2 Proof of Lemma \ref{lem2} & \\
\quad \quad \quad 3.3 Proof of Lemma \ref{lem3} & \\
\quad \quad \quad 3.4 Proof of Lemma \ref{lem4} & \\
\quad 4. Proof of Theorem \ref{LAPLA} & \quad\\
\quad 5. Remaining proofs & \quad \\
\quad \quad \quad 5.1 Local limit consideration& \\
\quad \quad \quad 5.2 The sharp threshold &   \\
 & \\
II. Uniquely extendible constraints & \\
\quad 1. Outline  & \\
\quad 2. Proof of Theorem \ref{UNOPT} for $d=4,\, s\ge 7 , \,  \lambda \le 1-1/d$  &\\
\quad 3. Proof of Theorem \ref{UNOPT}   for $d=4,\, \lambda \ge 1-1/d, s \ge 5.$   &\\
\end{tabular}

\quad \quad \\
\\

\noindent
{\bf \Large I.  Equations  modulo $3$ }

\setcounter{section}{0}

\section{Notation and basics}


We use the abbreviation 
\begin{eqnarray}
M(m, n) \,:= \, \sum_{ v_1, \dots v_n \,\ge 2}\, {m \choose v_1,\dots v_n } \mbox{ and } 
N_0 \,:= \, M(km, n). \mbox{ Then } N_0\cdot 3^m  \label{DEFN0} 
 \end{eqnarray}
is the number of all formulas with $ k $   variables  per equation and $ m $ equations.
We consider the uniform distribution on the set of formulas. 
Note that the formulas we consider are   $2-$cores. (Here the same equation to occur several times.
This happens with probability $o(1)$ as $m$ is linear in $n$ and can be ignored.  )
Let $X$ be the number of solutions of a formula.    We have $X=\sum_a\, X_a$ where $a$ stands for an assignment of the variables with $0,1, 2$  
and $X_a(F)=1$  if $F$ is true under $a$ and $0$ otherwise. 
The  expectation of $X$ is  $3^{n-m}$ because given an assignment each equation is true independently  with probability $1/3.$ We assume that 
$m=\gamma n , \gamma $ bounded above by a constant $<1.$  As $k$ is also constant, the asymptotics is only with respect  to $n.$
We need to show the following theorem 
\begin{theor} \label{EX2}
$\mbox{E}[X^2] \le  C \cdot 3^{2(n-m)} $
\end{theor}
We have  E$[X^2] \,= \, \sum_{(a, b)}\, \mbox{E}[X_a \cdot X_b] $ where  $(a, b)$ refers to all ordered pairs of assignments. 

Let $\overline{W}=(W_0, W_1, W_2)$ be a partition of the set of variables into $3$ sets.We always use the notation
$w_i = \sharp W_i,\, \,  \bar{w}\, = \, (w_0, w_1, w_2).$  For two assignments we write 
$b \, = \, D(a, \overline{W}) $  iff  $W_i= \{x\, | \,b(x)=a(x) +i \mod 3 \}.$ We have that
$a(x_1 +\dots +x_k)=b(x_1 +\dots +x_k)$ (Here $ \,  a(x_1 +\dots +x_k)$  is the value of $ x_1 +\dots +x_k$ under $a$
(analogously for $b$).) iff $\sum_{i=0, 1, 2} \,i \cdot \sharp \{ j \,| \,x_j \in W_i\} \,= \,0 \mod 3.$ 
This is equivalent to $ \sharp \{ j \,| \,x_j \in W_1\}\,= \, \sharp \{ j \,| \,x_j \in W_2\} \mod 3.$ 
Given $\bar{l}= (l_0, l_1, l_2) $ with $\sum l_i \, =\, km $ we let  ${ \cal K}(\bar{l})$ be the set of all
$3\times m-$matrices $(k_{i, j})_{0\le i\le 2, 1\le j \le m}$ with  $k_{1, j}= k_{2, j} \mod 3 $   and  each column sums to $k$, that is  $\sum_i k_{i, j}\,= k$   
for each $j.$ Moreover,  $ \sum_j k_{i,j} \,=\, l_i $ for $ i=0, 1, 2 $  ( the $i'$th row sums to $l_i.)$

We denote
\begin{eqnarray}
 K(\bar{l}) :=
 \sum_{(k_{i, j})\in {\cal K}(\bar{l})}\,\prod_{j=1}^m{ k \choose k_{0, j}\,, k_{1, j}\,, k_{2, j}} .\, \, 
\mbox{ Then } \, \, \hat{N}(\bar{w}, \bar{l})\,:= \,  K(\bar{l}) \cdot \prod_{i=0}^2 M(l_i, w_i)  \label{NUFO}
\end{eqnarray}
is the number of formulas  $F$ true under two assignments $a, b$ with
$b\,= \, D(a, \overline{W}) $  (with $w_i =\sharp W_i)$ and  the variables from
$W_i$ occupy exactly $l_i$ slots of $F.$   The factor   
$K(\bar{l})$ of  $\hat{N}(\bar{w}, \bar{l})$ counts how the $l_i$ slots available for $W_i$ are distributed over
the left-hand-sides of the equations. The second factor counts how to place the variables into their slots.
Note that the right-hand-side of an equation  cannot be chosen, it is determined by the value of the left-hand-side
under $a, b.$ 

We abbreviate ${n \choose \bar{w}}\,= \, { n \choose w_0, w_1, w_2}. $  Given an assignment    $a, $  $\bar{w}, $ and $\bar{l}, $  the number of assignment formula pairs  $(b, F)$ with : There exist  $\overline{W}$ with $\sharp W_i= w_i,$ such that
$b \in D(a, \overline{W}),$ 
$F$ is true under $a$ and $b$,  and the variables from $W_i$ occupy exactly   $l_i$ slots of $F$ is 
\begin{eqnarray}
N( \bar{w}, \, \bar{l})\,
:= {n \choose \bar{w}} \cdot \hat{N}( \bar{w}, \, \bar{l}). \, \, \mbox{ This implies  } \, \,  \mbox{E}[X^2]\,= \, 3^n \cdot  \sum_{\bar{w},  \bar{l}}\, N(\bar{w}, \,\bar{l}) \cdot  \frac{1}{3^m \cdot N_0}.  \label{EWQU} 
\end{eqnarray}

Theorem \ref{EX2} follows directly from the next theorem:
\begin{theor}  \label{EX2EI}
$\sum_{\bar{w}, \bar{l}}  N(\bar{w}, \bar{l})/N_0 \, \le \, C\cdot 3^{(1- \gamma )n} .$ 
\end{theor}
One more piece of notation: 
$ \omega_i= w_i/n$ usually is the fraction of variables belonging to  $W_i.$ And 
$\lambda_i\,= l_i\,/ (km)=l_i/(k\gamma n)$ is the fraction of slots  filled with a variable from $W_i.$ 
We  use  $\bar{\omega} = (\omega_0, \omega_1, \omega_2), $ and   $\bar{\lambda}=(\lambda_0, \lambda_1, \lambda_2).$ 
Sometimes $\omega_i, \lambda_i$ stand for arbitrary  reals, this should be clear form the context.

\section{ Outline of the proof of Theorem  \ref{EX2EI} } 


First,  bounds for $M(m, n)$ and $K(\bar{l}).$ 
We consider  $q(x):= \exp(x)-x-1=  \sum_{j \ge 2} \frac{x^j}{j!} $ for $x\ge 0.$ Then for $a>0$ and all $m, n$  
\begin{eqnarray}
M(m, n) \,= \, \mbox{Coeff}[x^m, q(x)^n]\cdot m! \, < \, q(a)^n \cdot \frac{1}{a^m } \cdot m!\,\le \,  q(a)^n \left(\frac{m}{a\cdot e}\right)^m \cdot O(\sqrt{m})  \label{M1}
\end{eqnarray}
using Stirling  in the form $ m! < (m/e)^m \cdot  O(\sqrt{m}).$

To get rid of the $\sqrt{m}-$factor we let $Q(x):= xq'(x)/q(x)$ with $q'(x)$ the derivative of $q(x), \,q'(x)=\exp(x)-1$ for $x>0.$ 
Then $Q'(x)>0$  for $x>0,$ $Q(x)>x,$ and $Q(x)\longrightarrow 2$ for $x \longrightarrow 0.$
Thus, for $y>2$ the inverse function $Q^{-1}(y)>0$ is defined and differentiable. Lemma \ref{MLOC} is proved in 
Section \ref{REPRO}. 

\begin{lem}  \label{MLOC}
Let $Cn \ge m \ge (2+\varepsilon)n, \, \, C, \varepsilon>0$ constants. Then 
\begin{eqnarray}
M(m, n) \,= \, \Theta(1)\cdot \left( \frac{ m}{ae}\right)^m\cdot q(a)^n \mbox{ with } a \mbox{ defined by }  Q(a)=\frac{m}{n} \nonumber
\end{eqnarray}
\end{lem}

Throughout we use $s=s(k, \gamma)$ uniquely defined  by $Q(s) = k \gamma= k \gamma n/n=km/n.$ Note that for $k \ge 3$ we can  
assume that $k \gamma >2$ and $s$ always  exists. We have $Q(s) \ge s.$ We often write $Q$ instead of $Q(s).$  Recall $N_0=M(km, n)$  and we get a tight bound on the number of formulas (cf. (\ref{DEFN0}).) 
\begin{cor} \label{N0}
 $ N_0 \,= \, \Theta(1) \left(k \gamma n/(s e) \right)^{k \gamma n}\cdot q(s)^n. $
\end{cor}

We treat the sum $ K(\bar{l})$ similarly to  
$M(m, n).$   
Instead of $q(x)$ we use the function,  
\begin{eqnarray}
 r(\bar{x}):= \sum_{k_1 = k_2 \mod 3}  { k \choose k_0, k_1, k_2} x_0^{k_0}x_1^{k_1} x_2^{k_2}, \, \,  \bar{x}=(x_0, x_1, x_2).\, \mbox{ Then} \, \,\nonumber \\
{K(\bar{l})} \, =\, 
 \sum_{(k_{i, j})\in {\cal K}(\bar{l})}\,\prod_{j=1}^m{ k \choose k_{0, j}\,, k_{1, j}\,, k_{2, j}} \, \, = \, \mbox{Coeff}[\bar{x}^{\bar{l}},\, r(\bar{x})^m\,  ] \, 
< \,\frac{ r(\bar{c})^m}{\bar{c}^{\bar{l}}}   \label{K1}
\end{eqnarray}
with the notation $ \bar{x}^{\bar{l}} = \prod_i x_i^{l_i}$ and $\bar{c}=(c_0, c_1, c_2)>0,$ meaning $c_i >0$ for all $i.$ 

For calculations it is useful to have  a different representation of $r(\bar{x}).$ Let 
$\i$ be the imaginary unit, and
 ${\bf w_1}\,:= -1/2+ (\sqrt{3}/2) \i$ is the primitive third root of unity, 
${ \bf w_2} := -1/2- (\sqrt{3}/2 )\i\,= \, {\bf w_1}^2.$ We have 
\begin{eqnarray}
r(\bar{x}) = 
\frac{1}{3} \left[ (x_0\,+\,x_1\,+\,x_2 )^k \,+  \, 
(x_0 \, + \,{ \bf w_1} x_1 + {\bf  w_2} x_2)^k \,+ \, (x_0 \,+ \,{ \bf w_2} x_1 \,+ \, 
{\bf w_1}  x_2)^k \right] \label{RCOM}
\end{eqnarray}
The preceding equation is well known and easy to prove from basic properties of 
roots of unity.  Note that  in  derivatives  $\frac{d}{dx_i} r(\bar{x})$ the roots of unity
are treated as constants.

For $ x_i, y_i >0$ we define (convention  $\alpha^\alpha=1$ for $\alpha=\omega_i$ or $\alpha=\lambda_i$ and
$\alpha =0)$   
\begin{eqnarray}
\Psi( \bar{\omega}\, ,\, \bar{\lambda}\, ,\, \bar{x} \,,\, \bar{y}  \, )\,= \, 
\prod_{i=0,1,2}\left( \frac{q(x_i )}{\omega_i q(s)} \right)^{\omega_i} \cdot
\left[ \prod_{i=0, 1, 2} \left(\frac{\lambda_i s }{x_i y_i}\right)^{\lambda_i } \right]^{k \gamma}
 r(y_0, y_1, y_2)^\gamma  \nonumber 
\end{eqnarray}
With $\omega_i=\lambda_i=1/3, a_i=s(k, \gamma)=s, $ and $c_i=1,$  we have  
$\Psi(\overline{\omega}, \overline{\lambda}, \bar{a}, \bar{c})=
3\cdot (1/3)^{k\gamma}\cdot ((1/3)3^k)^\gamma= 3^{1-\gamma} $ (use (\ref{RCOM}).)

\begin{lem} \label{LEMBA}
$ N(\bar{w}, \bar{l})/N_0 \, \, < \, \, \Psi( \bar{\omega}\, ,\, \bar{\lambda}\, ,\, \bar{a} \,,\, \bar{c}  \, )^n\cdot O(n)^{3/2}$
for any $\bar{a}, \bar{c},  a_i, c_i>0.$  
\end{lem}
\begin{proof}
${n \choose \bar{w}} \le \prod_i (1/\omega_i)^{\omega_i n}$ for all $\bar{w}, $  (\cite{MITZEN}, page 228 )
$\prod_{i=0, 1, 2} M(l_i, w_i)/N_0 \, \le \, \\
 \prod_i \left( (l_i/(a_ie))^{l_i} q(a_i)^{w_i} O(\sqrt{l_i})\right)\cdot (es/(k\gamma n))^{k\gamma n}\cdot 1/q(s)^n \cdot O(1)$ 
with (\ref{M1} ) and Corollary \ref{N0}. 
Observe that $l_i= \lambda_i k \gamma n, \sum_i \lambda_i=1, \sum \omega_i=1.$ Concerning $K(\bar{l})$ apply (\ref{K1}).
\end{proof} 

For  reals  $a, b$ we let ${\cal U}_\varepsilon(a, b)=\{(c, d)| \, \, |c-a|, |d-b|<\varepsilon \}$ be the open square 
neighborhood  of $(a, b).$  The notation $\bar{\lambda}, \bar{\omega}\in {\cal U}_\varepsilon(a, b)$ is used to mean  $(\lambda_1, \lambda_2), (\omega_1, \omega_2)\in {\cal U}_\varepsilon(a, b).$
Theorem \ref{OPT} is proved in Section \ref{PROOPT}. 
\begin{theor}   \label{OPT}
For any  $\bar{\lambda}>0 $ there exist $\bar{a}, \bar{c}\,>\,0$ such that: \\
(1) $\Psi( \bar{\omega}\, ,\, \bar{\lambda}\, ,\, \bar{a} \,,\, \bar{c}  \, )\, \le \, 3^{1-\gamma}. $ \\
(2) For any $\varepsilon>0, $ if  $\bar{\lambda} \notin  {\cal U}_\varepsilon(1/3, 1/3)$  then 
$\Psi( \bar{\omega}\, , \bar{\lambda}\, , \bar{a} \,, \bar{c}   )\, \le \, 3^{1-\gamma} -  \delta  $ for a  $\delta>0.$ \\
\end{theor} 

\begin{cor} \label{SUMOPT}
Let $U= {\cal U}_\varepsilon(1/3, 1/3)$  then 
$\sum_{\bar{\lambda} \notin U, \lambda_i>0 , \bar{\omega} } N(\bar{w}, \bar{l})/N_0  \,  <  \, C \cdot 3^{(1-\gamma)n}. $
\end{cor}
\begin{proof}
The sum has only $O(n^4)$ terms.  
With Lemma \ref{LEMBA} and Theorem \ref{OPT} (2) 
we  see that each term is bounded above by 
$ (3^{1-\gamma}- \delta)^n O(n)^{3/2}.$ 
\end{proof}

To treat $(\lambda_1, \lambda_2) $ close to $(1/3, 1/3)$
we need a lemma analogous to Lemma \ref{MLOC} for $K(\bar{l}).$ 
Let the function $R(x_1, x_2) \,= \, (R_1(x_1, x_2),\, R_2( x_1, x_2))$ be defined by 
$R_i(x_1, x_2) = \\=  x_i r_{x_i}(1, x_1, x_2)/r(1, x_1, x_2)$ for $i=1, 2,$ 
$r_{x_i}(1, x_1, x_2)$ is the partial derivative of $r(1, x_1, x_2)$ wrt. $x_i.$ 
The Jacobi Determinant of $R(x_1, x_2) $ is $>0$ at $ x_1=x_2=1$ (proof Subsection \ref{LOLICO}.)
Thus there is a neighborhood of
$(1, 1)$ 
in which  $R( x_1, x_2)$ is invertible and the inverse function is differentiable. 
We have that $R(1, 1)= (k/3, k/3).$ 
Thus for a suitable $\varepsilon $ and  $(\lambda_1, \lambda_2) \in  {\cal U}_\varepsilon(1/3, 1/3)$ 
we can define $(c_1, c_2)$ by $R(c_1, c_2)= (k \lambda_1, k\lambda_2).$ 
Moreover, $c_i=c_i(\lambda_1, \lambda_2)$ is differentiable. Lemma \ref{KLOC} is proved in Subsection \ref{LOLICO}.

\begin{lem}  \label{KLOC}
There is an   $\varepsilon>0$ such that for $(\lambda_1, \lambda_2) \in  {\cal U}_\varepsilon(1/3, 1/3)$  
\begin{eqnarray}
K(\bar{l}) \,= \, O\left( \frac{1}{n}\right) \cdot \frac{r(1, c_1, c_2)}{c_1^{l_1}c_2^{l_2}} \mbox{ with } R(c_1, c_2)=(k\lambda_1, k\lambda_2) \mbox{ defining } c_1, c_2.  \nonumber 
\end{eqnarray}
\end{lem}

\begin{cor} \label{PSIINU}
There is  $\varepsilon >0$ such that for $(\omega_1, \omega_2), ( \lambda_1, \lambda_2 ) \in   {\cal U}_\varepsilon(1/3, 1/3)$ 
\begin{eqnarray}
\frac{N(\bar{w}, \bar{l})}{N_0} \, \, \le  \, 
\,O \left(\frac{1}{n^2}\right) \Psi( \bar{\omega}\, ,\, \bar{\lambda}\, ,\, \bar{a} \,,\, \bar{c}  \, )^n \nonumber 
\end{eqnarray}
where $  Q(a_i)= l_i / w_i = \lambda_ik \gamma /\omega_i$ and $c_0=1$ and $ R(c_1, c_2)= (\lambda_1k, \lambda_2k). $ 
\end{cor}
\noindent
{\it Comment.} Observe that $\lambda_ik \gamma /\omega_i\,\approx \, k\gamma, \, a_i\approx s, c_i \approx 1.$ 
\begin{proof}
Our restriction on $\bar{\omega}$ implies that ${n \choose \bar{w}}\le O(1/n)\prod_i (1/\omega_i)^{\omega_i n}$ 
(Stirling), giving us one $O(1/n).$ We get $\prod_i M(l_i, w_i)/N_0 \le \\
\prod_i \left( (l_i/(a_ie))^{l_i} q(a_i)^{w_i}  \right)\cdot (es/(k\gamma n))^{k\gamma n}\cdot 1/q(s)^n \cdot O(1)$
applying Corollary  \ref{N0} and Lemma \ref{MLOC} for the $M(l_i, w_i).$ Concerning $K(\bar{l})$ apply Lemma \ref{KLOC}
which gives us a second factor $O(1/n).$  Otherwise the proof is as the proof of Lemma \ref{LEMBA}. 
\end{proof}

Lemma \ref{LOCMAX} is proved in Section \ref{PROLAPLA}. 

\begin{lem} \label{LOCMAX}
The function  $ \Psi( \bar{\omega}\, ,\, \bar{\lambda}\, ,\, \bar{a} \,,\, \bar{c}  \, )$
with  $a_i, c_i$ given by 
$  Q(a_i) = \lambda_ik\gamma /\omega_i$ and $c_0=1$ and $ R(c_1, c_2)= (\lambda_1k, \lambda_2k) $ 
has a local maximum with value $3^{1-\gamma}$  for $\lambda_i= \omega_i=1/3.$ In this case we get $a_i=s$ and $c_i=1.$ 
\end{lem} 

\begin{cor}   \label{SUMOMEGA}
Let $U= {\cal U}_\varepsilon(1/3, 1/3), \varepsilon$ small enough. Then 
$\sum_{\bar{\omega} \notin U,  \bar{\lambda} \in U, \lambda_i>0 } N(\bar{w}, \bar{l})/N_0   < \, C \cdot 3^{(1-\gamma)n}. $
\end{cor}
\begin{proof}
Let  $\varepsilon>0$ be  such that
$\Psi(\bar{\omega}, \bar{\lambda}, \bar{a}, \bar{c}) \le 3^{1-\gamma} $ for $\bar{a}, \bar{c} $
as specified in Lemma \ref{LOCMAX} and $\bar{\omega}, \bar{\lambda}\in U.$ 
Let $\bar{\omega } \notin U, \bar{\lambda} \in U.$ We show 
$\Psi(\bar{\omega}, \bar{\lambda}, \bar{a}, \bar{c}) \le 3^{1-\gamma} - \delta$ for some $\bar{a}, \bar{c}.$ 
This implies the claim   as in the proof of Corollary \ref{SUMOPT}. 

Let  $\varepsilon'<\varepsilon/3$  and  $U'= {\cal U}_{\varepsilon'}(1/3, 1/3).$   
For $ \bar{\lambda} \notin U'$  the claim follows with Theorem \ref{OPT} (2). 
For  $ \bar{\lambda} \in U'$ we show that 
$\Psi(\bar{\omega}, \bar{\lambda}, \bar{a}, \bar{c}) \le 3^{1- \gamma}- \delta$ for 
$a_i=s$ and $ c_0=1, R(c_1, c_2)=(k \lambda_1, k \lambda_2).$  (Recall $\bar{\lambda} \in U.$)
For $\Psi:= \Psi(\bar{\lambda}, \bar{\lambda}, \bar{a}, \bar{c})$ with 
$\bar{a}, \bar{c}$ as required by Lemma \ref{LOCMAX} we have $\Psi \le 3^{1-\gamma}.$ 
Note,  $Q(a_i)=\lambda_ik\gamma/\lambda_i= k \gamma$ 
which implies $a_i=s$ and $c_0=1, R(c_1, c_2)=  (k \lambda_1, k \lambda_2).$
Therefore all $a_i-$terms cancel and $\Psi= \prod (1/\lambda_i)^{\lambda_i}\cdot \left(\prod (\lambda_i/c_i)^{\lambda_ik\gamma}\right)p(\bar{c})^\gamma \le 3^{1-\gamma}. $ 

As $\bar{\omega} \notin U$ whereas  $\bar{\lambda} \in U'$ and $\varepsilon' \le \varepsilon/3$
we have that  
$\prod (1/\omega_i)^{ \omega_i } \, \le \prod (1/\lambda_i)^{\lambda_i}- \delta'$ for a $\delta' >0$ 
(proof omitted.) Then  
$\Psi(\bar{\omega}, \bar{\lambda}, \bar{a}, \bar{c}) \le \Psi - \delta'\left(\prod 
(\lambda_i/c_i)^{\lambda_i k \gamma }\right)p(\bar{c})^\gamma .$
If $\Psi \le 3/2$, we are done. Otherwise we have that $\left(\prod_i(\lambda_i/c_i)^{\lambda_ik\gamma}\right)p(\bar{c})^\gamma $
is bounded below by $1/2$  (as $\prod (1/\lambda_i)^{\lambda_i} \le 3)$ and the claim  follows, with
with $\delta = (1/2)\delta'.$  
\end{proof}

Theorem \ref{LAPLA} is proved in Section \ref{PROLAPLA} by Laplace method.

\begin{theor} \label{LAPLA}
Let $U={\cal U}_\varepsilon(1/3, 1/3).$ There is an $\varepsilon>0$ such that \\
\hspace*{4cm} $\sum_{\bar{\lambda}  , \bar{\omega} \in U} N(\bar{w}, \bar{l})/N_0 \,< \, C \cdot 3^{(1-\gamma)n}.$ \\  
\end{theor}

\noindent
{\it Proof of Theorem \ref{EX2EI}.}
Pick $\varepsilon $ such that Theorem \ref{LAPLA} applies.  
Use Corollary \ref{SUMOPT}, Corollary \ref{SUMOMEGA},  and Theorem \ref{LAPLA}  and the sum of all terms $N(\bar{w}, \bar{l})/N_0$ 
with $l_i>0$  is $\le C \cdot 3^{(1- \gamma)n}.$  Terms with an  $l_i=0$ do not add substantially to the sum (proof omittted.) 
 \qed

\section{Proof  of Theorem \ref{OPT}} \label{PROOPT}

We use the notation  $\bar{x}= (x_0,x_1, x_2), \bar{y}=(y_0, y_1, y_2)$ 
and define
\begin{eqnarray}
\mbox{OPT}_1(\bar{x} ,  s)\,= \,  \frac{ q(sx_0)}{q(s)}  \,+ \, \frac{q(sx_1)}{q(s)} + \frac{ q(sx_2)}{q(s)} \nonumber \\
\mbox{OPT}_2(\bar{x}, \bar{y} ,  s)\,= \, \left( \frac{1}{x_0y_0+x_1y_1+ x_2y_2} \right)^{Q}  , \, x_0y_0+x_1y_1+ x_2 y_2 >0\nonumber \\
 \mbox{OPT}_3(\bar{y}, s)\,= \, \left( y_0+y_1+y_2 \right )^Q \,+ \, 
2\cdot \left( y_0^2 \,+\,y_1^2\,+\, y_2^2\,- y_0 y_1\,- \, y_0 y_2\, - \, y_1 y_2 \right)^{1/2\cdot Q} \nonumber \\
\mbox{OPT}(\bar{x}, \bar{y},  s) \,=\, 
\mbox{OPT}_1(\bar{x} , s) \cdot 
\mbox{OPT}_2 ( \bar{x}, \bar{y} ,  s)\cdot 
\mbox{OPT}_3(  \bar{y}, s) .  \nonumber 
\end{eqnarray}
Observe that  OPT$(1, 1, 1, 1, 1,1 , s)= 3(1/3)^Q 3^Q=3=$ OPT$_1(1, 1, 1, s),$  OPT$(1,0,0, 1,0,0, s)=1\cdot (1/1)^Q \cdot 3=3=$OPT$_3(1, 0, 0,s).$  The following lemma shows the idea of OPT.  
\begin{lem} \label{LEMOPT} 
Given $\bar{\lambda}>0$ and let  $\lambda$ be the maximum of the $\lambda_i.$  
Let  $ \,  a_i, c_i \,>0$ be such that   
$P_i:= a_ic_i= \lambda_i/\lambda.$ Then  
\begin{eqnarray}
\Psi:= \Psi(\bar{\omega}, \bar{\lambda}, \bar{a}\cdot s, \bar{c}) \le \frac{1}{3^\gamma} 
\mbox{OPT}(\bar{a}, \bar{c}, s).       \nonumber 
\end{eqnarray}
\end{lem}
\begin{proof}
The factors of $\Psi$ one by one: The first factor: 
The AGM-inequality  gives \\
$ \prod_{i=0, 1, 2} \left(\frac{q(a_i s)}{\omega_i q(s)} \right)^{\omega_i} \,\le \, \mbox{OPT}_1(\bar{a}, s). $
(Applies for $\omega_i=0$, too.)

The second factor: We have $P_0+P_1+P_2= a_0c_0+ a_1c_1+ a_2c_2 = 1/\lambda$ 
and $\lambda_i/a_ic_i= \lambda$ for $i=0, 1, 2.$ Recall $Q=k\gamma, $ and the second factor of $\Psi \,=$ 
\begin{eqnarray}
\prod_{i=0, 1, 2} \left( \frac{\lambda_is}{a_isc_i}\right)^{\lambda_ik\gamma} \,= \, \lambda^{k \gamma} = 
\left(\frac{1}{a_0c_0+a_1c_1+a_2c_2}\right)^Q\,= \mbox{OPT}_2(\bar{a}, \bar{c}, s). \nonumber 
\end{eqnarray}

The third factor:  
We let $C_1=\sum_i c_i$ and $C_2= \sum_i c_i^2 -c_0c_1-c_0c_2-c_1c_2.$ 
Then $r(\bar{c}) =|r(\bar{c})| \le (1/3)(C_1^k\,+ \,2 C_2^{k/2})$ by the triangle inequality and as
$|c_0 +{\bf w_1}c_1 + {\bf w_2}c_2|= [(c_0-1/2\cdot (c_1+c_2))^2+ (\sqrt{3}/2( c_1-c_2))^2]^{1/2}=C_2^{1/2}.$  
Then 
$|r(\bar{c})|^\gamma  \le 1/3^\gamma (C_1^k+2C_2^{k/2})^\gamma \le 1/3^\gamma (C_1^{k\gamma}+2^\gamma C_2^{\gamma k/2})\le1/3^\gamma \mbox{OPT}_3(\bar{c}, s)$
as $Q=k\gamma,$ and as $x^\gamma$ is concave (by $\gamma <1)$ we have $(y+z)^\gamma \le y^\gamma + z^\gamma.$  
\end{proof} 

The following picture shows OPT$(1, a, a, 1,  c, c, s), \, \, 0 \le a, c \le 1.$  The $\le 3-$area  is dark.
We have a path from $a=c=0$ to $a=c=1$ through this area. Therefore, for  all $P$ with $0 \le P \le 1$ we have
$0 \le a,  c\le 1$ with $P=ac$ such that  OPT$(1, a, a, 1,  c, c, s) \le 3.$  In the notation of Lemma \ref{LEMOPT}
this corresponds to $\lambda_0 \ge \lambda_1=\lambda_2$ (and visualizes   Theorem \ref{OPT}  for this case.)     
The following four lemmas are the technical core of our proof. 

\begin{figure}[th]
\centering
\includegraphics[width=0.49\textwidth]{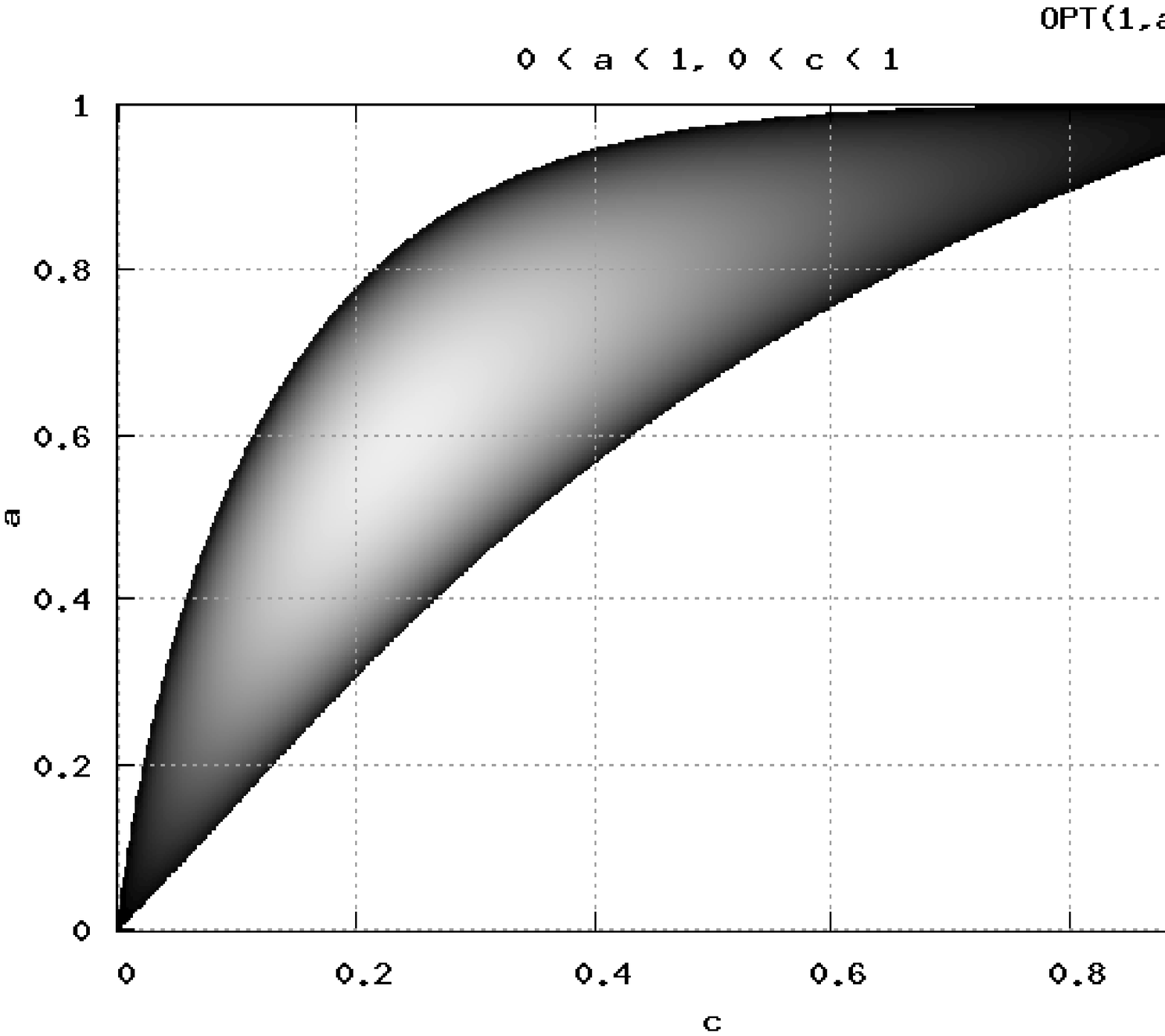} \hfill
\includegraphics[width=0.49\textwidth]{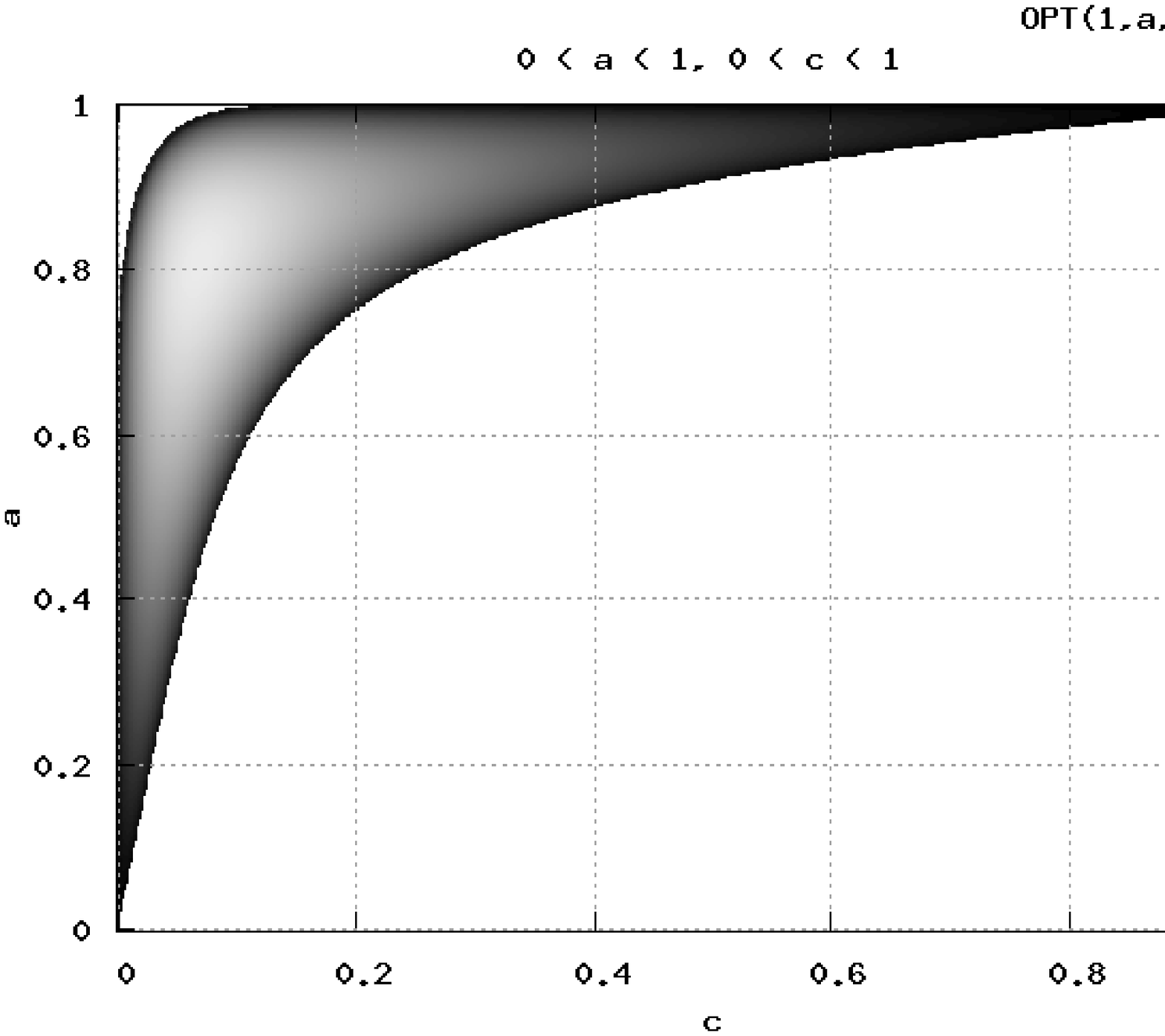}
\caption{OPT$(1,a,a,1,c,c,s)$ over the rectangle $0\le a \le 1, 0 \le c \le 1$ for $s=3$ and $s=14$.}
\end{figure}

\begin{lem} \label{lem1}
Let $s\ge 8, A(x)\,=\,A(x, s)  := \,(7/10)Q \cdot x.$ \\
(a) OPT$(y)\,\,:=$  OPT $\left( 1, A(y), \, A(y), \, 1, y, \, y, \, s \right)$ is strictly decreasing for  $   0 \le y \le 1/(2Q).$ The start value is
OPT$(0)=3.$ \\
(b) Given $0\le y \le 1/(2Q),$ OPT$(z):= $ OPT$(1, A(y+z), \, A(y-z), \,1,  y+z, \, y-z, \, s)$ is decreasing in $0\le z \le y.$ 
\end{lem}

\begin{lem} \label{lem2}
Let  $ s\, \ge 7 \, \,$, and  $  \frac{7}{20} \le A \le 1-\frac{1}{Q}.$ Then \\
OPT$(z):= $ OPT$\left(1, A,\, A, \, 1, 1/(2Q)+z, \, 1/(2Q)-z \,  , \, s \right) \le 3 \,- \, \delta$ for $0 \le z \le 1/(2Q).$ 
\end{lem}

\begin{lem} \label{lem3}
Let $\, s\, \ge \, 7,$ and $  1/(2Q) \le C \le 1/2.$ Then \\
OPT$(z):= $ OPT$\left(1, 1-1/Q,\, 1-1/Q, \,1,  C+z, \, C-z\,, s \right)\, \le  3 \,- \, \delta$   for $0 \le z \le C.$ 
\end{lem}

\begin{lem} \label{lem4}
Let $s\ge 15$ and $  A(x)\,=\, A(x, s) \, := \, 1+ 7/(10Q)\cdot x \,- \, 7/(10Q).$ \\
(a) OPT$(y):= $OPT$\left(1, A(y), A(y), 1, y, y, s \right) $ is strictly increasing in $4/10 \le y \le 1.$ The final value is 
OPT$(1)=3.$ \\
(b) Given $4/10 \le y \le 1,$ OPT$(z):=$OPT$\left(1, A(y+z), A(y-z),1,  y+z \, , y-z \, , \, s \right)   $ is decreasing in $0 \le z \le \min\{y, 1-y\}.$
\end{lem}

\noindent
{\it Proof of Theorem \ref{OPT} from the preceding  lemmas.}
We prove Theorem \ref{OPT} for $\lambda_0 \ge \lambda_1 \ge \lambda_2>0 $ first. 
We denote $P_i:=\lambda_i/ \lambda_0, \, \mbox{ then }  1\ge P_1\ge P_2>0.$ 

\noindent
{\it Case 1: $P_1\, +\, P_2 \le \frac{7}{20 Q}$.} With $A(x)$ from Lemma \ref{lem1} 
we have $A(x)\cdot x \,= \, (7/10)Q \cdot x^2.$ Thus there exist $y_1 \ge  y_2$ with 
$P_i= A(y_i)\cdot y_i.$ We represent $y_i$ such that Lemma \ref{lem1} is applicable.  
\begin{eqnarray}
y\,:= \, \frac{y_1 + y_2}{2} ,\,\,  z \,:= \, \frac{y_1 - y_2}{2}.
\mbox{ Then } y_1 = y+z , y_2 = y-z , 0 \le z \le y. \nonumber 
\end{eqnarray}
We show $y\le \frac{1}{2Q}$ and Lemma \ref{lem1} applies to $y, z.$  
\begin{eqnarray}
\frac{7}{20Q}\, \ge \, P_1+P_2 \,= \, \frac{7}{10}Q (y_1^2+ y_2^2) \Longrightarrow y_1^2+ y_2^2 \,\le \, \frac{1}{2Q^2}. \nonumber \\
 (y_1+ y_2)^2 \le  2 y_1^2+2 y_2^2 \le  \frac{1}{Q^2} \mbox{ and } y= \frac{y_1+y_2}{2}\le \frac{1}{2Q}. \nonumber 
\end{eqnarray}
With $a_0=c_0=1, a_1=A(y+z), a_2=A(y-z), c_1=y+z, c_2=y-z$ we have $a_ic_i=P_i.$ By Lemma \ref{LEMOPT} 
$\Psi:= \Psi(\bar{\omega},\bar{\lambda}, \bar{a}\cdot s, \bar{c})\le 1/3^\gamma \mbox{OPT}, \,\,
\mbox{OPT}:= \mbox{OPT}(\bar{a}, \bar{c},s).$ 
If   $P_1 \ge \varepsilon $ for an  $\varepsilon>0$ we have  OPT$< 3-\delta'$ by Lemma \ref{lem1} and Theorem \ref{OPT} holds.   

For smaller $P_1$ we have OPT $ \le 3,$ approaching $3$. Only (1) of Theorem \ref{OPT} holds. 
To get (2) for small $P_1$ we argue as follows:
For  $P_1$ approaching $0$ we see that 
$c_1$ and $c_2$ approach $0 .$ We consider the treatment of the factor $r(\bar{c})$ in the  proof Lemma \ref{LEMOPT}.
Both $C_1$ and $ C_2$ from this proof  approach $1$ in this case.  Therefore we have a $\delta'>0$ such that 
$(C_1^k+2C_2^{k/2})^\gamma \le C_1^{k\gamma}+2^\gamma C_2^{k\gamma/2}- \delta'.$ As $a_0=c_0=1$ the first two factors
of  OPT
do not approach $0.$  And we have  $\Psi(\bar{\omega},\bar{\lambda}, \bar{a}\cdot s, \bar{c})\le (1/3^\gamma)(\mbox{OPT}- \delta'') \le 3^{1-\gamma}-\delta$ 
and Theorem \ref{OPT} (2) holds. 

\noindent
{\it Case 2:  $\frac{7}{20Q} \le P_1\, +\, P_2 \le \left(1\,- \, \frac{1}{Q}\right) \frac{1}{Q}.$} To 
use Lemma \ref{lem2} we define $A$ by 
$A\cdot \frac{1}{Q}=P_1+P_2.$
and $A$ is as required by Lemma \ref{lem2}. We need to find 
an appropriate $z.$ 
 As $   P_1 \ge P_2 $  there is a $  y\ge \frac{1}{2}$  such that $ P_1=A\frac{1}{Q}y $  and $ P_2=A\frac{1}{Q}(1-y).$ 
 With $  y=\frac{1}{2}+ z' $   and $ 1-y = \frac{1}{2}-z' , z' \le \frac{1}{2},$ and $ P_1= A\left(\frac{1}{2Q}+ \frac{z'}{Q}\right) \, , \, P_2\,= \, A\left(\frac{1}{2Q}-\frac{z'}{Q} \right) $
 Lemma \ref{lem2} applies with $z=z'/Q.$
Again we set $a_0=c_0=1$ and $a_1=a_2=A, c_1= \frac{1}{2Q}+ z, c_2= \frac{1}{2Q}- z.$
By Lemma \ref{LEMOPT}   $\Psi(\bar{\omega},\bar{\lambda}, \bar{a}\cdot s, \bar{c})\le 3^{1-\gamma} -\delta.$ 

\noindent
{\it Case 3:  $ \left( 1 \,- \, \frac{1}{Q}\right) \frac{1}{Q} \le P_1+P_2 \le
  1\,- \, \frac{1}{Q} . $}
Let $C$ be given by  $\left( 1 \,- \, \frac{1}{Q}\right)\cdot C \,= \, \frac{P_1 + P_2}{2}.$ 
Then  $C$ is as required by Lemma \ref{lem3}. We have a $0 \le z' \le \frac{1}{2}$ such that 
\begin{eqnarray}
P_1 \,= \, \left( 1 \,- \, \frac{1}{Q}\right)\cdot C \cdot 2\left(\frac{1}{2}\, +\, z'\right)\, \, = 
\left( 1 \,- \, \frac{1}{Q}\right)\cdot (C \, +\, 2Cz'), \nonumber \\
P_2 \,= \, \left( 1 \,- \, \frac{1}{Q}\right)\cdot (C - 2Cz').\nonumber 
\end{eqnarray}
With $z=2Cz'\le C $  Lemma \ref{lem3} applies. We set $a_0=c_0=1$ and $a_1=a_2=1-1/Q$ and $c_1=C+z, c_2=C-z$ and finish the argument
as in Case 2.

\noindent
{\it Case 4:  $ P_1 + P_2 \ge 1- \frac{1}{Q}.$ }
With $A(x)$ as from Lemma \ref{lem4} we have 
$A(x)\cdot x\,= \, \left( 1- \frac{7}{10Q} \right) x \,+ \, \frac{7}{10Q}x^2$ and
 $A(x)x $ increases  from $0$ to $1$ for $0\le x \le 1.$ 
Let $y_i$ be such that $P_i\,=\,A(y_i)\cdot y_i.$ Then $y_2 \le y_1 \le 1$ and 
we can represent $y_i$ such that Lemma \ref{lem4} 
is applicable. 
\begin{eqnarray} 
 y :=  \frac{y_1+ y_2}{2} , z:=  \frac{y_1- y_2}{2},\mbox{ and }  y_1=y+z, y_2=y-z, z \le y, 1-y. \nonumber 
\end{eqnarray}
We show that $1\ge y \ge 4/10$ and Lemma \ref{lem4} applies to $y, z.$ 
We have  $ P_1+P_2=A(y_1)y_1 + A(y_2)y_2 = y_1+y_2 + 7/(10Q)(y_1^2+y_2^2-y_1 - y_2) \le y_1+y_2.$
Therefore $y = (y_1+y_2)/2 \ge 1/2(1-1/Q)\ge 4/10 $ as $Q \ge s \ge 15.$
Setting $a_0=c_0=1, a_1=A(y+z), a_2=A(y-z), c_1=y+z, c_2=y-z$ implies the claim. 

Now, assume the $\lambda_i$ are ordered in a different way. We apply the  
permutation leading from $ \lambda_0 \ge \lambda_1 \ge \lambda_2$   to the ordering considered
to the $P_i , a_i , c_i$  above.  The first two factors of $\Psi$ do not change, only $r(\bar{c})$ may change.
But, Lemma \ref{LEMOPT} still applies. 
The three factors, OPT$_1$, OPT$_2$, OPT$_3$  of OPT$(\bar{a}, \bar{c}, s)$ do not  change.
This refers to $C_1, $ and $C_2, $  too, and the argument above for  $P_1$ small  applies, too. \qed  

In the proofs to come in the following four subsections we use the notation
\begin{eqnarray}
L(a, s)= \frac{q(as)}{q(s)}= \frac{\exp(as)-as-1}{\exp(s)-s-1}, \, K(a, s)= \frac{q'(as)}{q'(s)}\,= \,\frac{\exp(as)-1}{\exp(s)-1}, \nonumber \\
 M(a, s)= \frac{\exp(as)}{\exp(s)}. \mbox{Then } aK(a, s) \le  L(a, s) \le K(a, s) \le M(a, s) \,,\,  0 \le a \le 1 . \, \label{BAKL} 
\end{eqnarray}
{\it Proof of (\ref{BAKL}.) }  $p(x):=q'(x), \,K:= K(a, s), L:= L(a, s).$ For $a=0$ or $a=1$ we have $aK=L.$ For $a>0,$ 
$aK \le L \Longleftrightarrow  ap(as)/q(as) \le p(s)/q(s) \Longleftrightarrow  asp(as)/q(as) \le sp(s)/q(s).$ 
The preceding inequality holds trivially for $a=1.$  We show that   $as p(as)/q(as)$ is  strictly increasing in $a>0.$ 
We observe that $q(x)/(xp(x)) = 1/x- 1/p(x).$  The derivative is of the last expression is $<0$ iff $x^2+2< \exp(x)+1/\exp(x).$ 
For $x=0$ we have equality and several differentiations show the inequality. 

For $a=0, L\le K$ is true. For $a>0$  $L \le K  \Longleftrightarrow 1- sa/p(sa) \le 1-s/p(s).$
The last inequality follows from $a \ge p(sa)/p(s)$ for $0 \le a \le 1.$ This follows from convexity. 
$K(a, s) \le M(a, s)$ is very easy to show.  \qed 

\begin{eqnarray}
\mbox{We also have } aK(a, s) \le \frac{7}{10} L(a, s), \,\mbox{ for } 0 \le a \le \frac{1}{2}, s\ge 4 \, \, \, \mbox{ (proof omitted.) }   \label{BAKL7/10} \\
\mbox{ We recall }   Q(x)= \frac{x q'(x)}{q(x)}=\frac{x (\exp(x)-1)}{\exp(x)-x-1}, Q=Q(s)=k\gamma , \, Q(s)> s.  \nonumber 
\end{eqnarray}

\subsection{ Proof of Lemma \ref{lem1}}
\noindent
{\bf Lemma \ref{lem1} (repeated)} 
Let $s\ge 8, A(x)\,=\,A(x, s)  := \,(7/10)Q \cdot x.$ \\
(a) OPT$(y)\,\,:=$  OPT $\left(1, A(y), \, A(y),1,  \, y, \, y, s \right)$ is strictly decreasing for  $   0 < y \le 1/(2Q).$ The start value is
OPT$(0)=3.$ \\
(b) Given $0\le y \le 1/(2Q),$ OPT$(z):= $ OPT$(1, A(y+z), \, A(y-z), \,1, \, y+z, \, y-z, \, s)$ is decreasing in $0\le z \le y.$ 

\noindent
{\it Proof of (a).} We have  
\begin{eqnarray} 
\mbox{OPT}(y) = 
 \left( 1+2L(A(y), s) \right) 
\left(\frac{1}{1+2A(y)\cdot y}\right)^Q  \left(\left(1\,+ \, 2y\right)^Q \,+ \,2\left(1\,- \, y \right)^Q \right). \nonumber \\
\mbox{ We write   OPT}_1(y)\,= \, 1+2L(A(y), s).
\mbox{ Clearly  OPT}(0)\,= \, 3 .\nonumber \\ 
\mbox{ We have } A' := \frac{d}{dy} \,A(y)\,= \,\frac{7}{10} Q.  \, \, \, \, \mbox{ And } 
\frac{d}{dy} \ln \mbox{ OPT }(y) \,>=< \, 0 \Longleftrightarrow \nonumber \\
\frac{A' \cdot 2 \cdot K(A(y), s)\,} {\mbox{ OPT}_1(y)}\,- \, 
\frac{2A(y)\,+ \, 2A'\cdot y}{1+2A(y) \cdot y}\,+ \,\frac{2(1+2y)^{Q-1}- 2(1-y)^{Q-1}}{\left(1\,+ \, 2y\right)^Q \,+ 
\,2\left(1\,- \, y \right)^Q } \nonumber \\ \,>=<\,0  \label{EQKEY}
\end{eqnarray} 
The  relationship (\ref{EQKEY}) is obtained by taking the derivative and dividing by $Q.$ 
To get the first summand we look into  the definition of $Q$ (the formula after \ref{BAKL7/10}.)
\begin{eqnarray}
\frac{d}{dy} \ln  \mbox{ OPT }_1(y)=  \frac{ \frac{A' s\cdot 2  (\exp( A(y)s )-1)}{q(s)}} {\mbox{ OPT}_1(y)}\, ,\, 
\frac{1}{Q(s)}  \frac{A's\cdot 2 (\exp(A(y)s)-1)}{q(s)}=  A'\cdot 2 K(A(y), s)         \nonumber 
\end{eqnarray}

Observe that the first and third term of (\ref{EQKEY})  is $\ge 0$ 
for $0 \le y \le 1$ whereas the second term is $\le 0.$ Moreover, $ A'\cdot y = A(y). $
We have that  $\frac{d}{dy} \ln \mbox{ OPT }(y)\,<\,0$  if the following two inequalities  both hold:
\begin{eqnarray}
\frac{A' \cdot 2 \cdot K(A(y), s)}{\mbox{ OPT}_1(y)}\,< \, \frac{\, \frac{7}{10}A' \cdot y}{1+2A(y)\cdot y} \label{klsymugl1} \\
 \frac{2(1+2y)^{Q-1}- 2(1-y)^{Q-1}}{\left(1\,+ \, 2y\right)^Q \,+ \,2\left(1\,- \, y \right)^Q }\,   
 < \frac{\frac{33}{10}A'\cdot y}{1+2A(y) \cdot y} \, \, \label{klsymugl2}
\end{eqnarray}
Note that for $y=0$ both sides of the first inequality are equal to  $0$ and of the second 
inequality, too. The derivative of OPT$(y)$ is $=0$ for $y=0.$ 

\noindent
{\it Comment:}  It is important to split up the left-hand-side of inequality (\ref{EQKEY}), otherwise the
calculations get very complicated.  Equally important is the  step leading to (\ref{EQKEY}). 
Analogous steps will occur several times.

\noindent
{\it Proof of (\ref{klsymugl1}) for $ 0 < y \le 1/(2Q) \, \,    ,\,\,  s \ge 7 $ .} 
We abbreviate $K:=K(A(y), s)\, \, , L:=L(A(y),s).$ Note OPT$_1(y)\, = \, 1+2L.$ 
 As $ A'>0$    we show 
\begin{eqnarray}
 \frac{2 \cdot K}{1+2L} \,< \, \frac{\frac{7}{10}y}{1+2A(y)\cdot y} 
\Longleftrightarrow 2K+ 4K\cdot A(y)\cdot y-2\frac{7}{10} L\cdot y\,< \, \frac{7}{10}y. \nonumber \\
\mbox{ By (\ref{BAKL7/10}) we know  } K \cdot A(y) \le \frac{7}{10} L  \mbox{ for } s \ge 4 \mbox{ as} A(y) \le \frac{1}{2}
\mbox{ ( by  } y \le \frac{1}{2Q}. )  
\end{eqnarray}
Thus (\ref{klsymugl1}) follows from  $ 2K+ 2K \cdot A(y) \cdot y < \frac{7}{10}y.$    
As $ 2K \cdot A(y) \cdot y \, \le \, 2K$ and $2K$ is convex and $2K=0$ for $y=0$ we show that 
$2K<(7/20) y$ for $y=1/(2Q).$  For $y=1/(2Q)$ we have $   A(y)=7/20$ and  $ 2K= 2(\exp((7/20)s)-1)/(\exp(s)-1).$ 
As $1/(2Q)=(\exp(s)-s-1)/(2s(\exp(s)-1))$ we have   for $y=1/(2Q)$ 
\begin{eqnarray}
2K\,< \,\frac{7}{20}y \Longleftrightarrow 2 \left(\exp \left(\frac{7}{20}s\right)-1 \right)\,<\,\frac{7}{40} \frac{\exp(s)-s-1}{s}  \nonumber 
\end{eqnarray}
This last inequality holds for $s \ge 7 $ (but not for $s\le 4.$ )\\

\noindent
{\it Proof of  (\ref{klsymugl2})  for $y \le 1/Q $ and $s \ge 2.$ } 
Inequality (\ref{klsymugl2}) is equivalent to 
\begin{eqnarray}
2(1+2y)^{Q-1}- 2(1-y)^{Q-1} \, < \,  \nonumber \\ 
 <  A(y)\left[\frac{33}{10}\left[\left(1\,+ \, 2y\right)^Q \,+ \,2\left(1\,- \, y\right)^Q\right] \,- \, 2y\cdot \left[ 2(1+2y)^{Q-1}- 2(1-y)^{Q-1} \right]\right] \,\, \label{klsymugl210} \\
\mbox{ The right-hand-side of (\ref{klsymugl210}) is  } \,\ge  \, \nonumber \\
A(y)\left[\frac{33}{10}\left[\left(1\,+ \, 2y\right)^Q \,+ \,2\left(1\,- \, y\right)^Q\right] \,- \, \frac{33}{10}y\cdot \left[ 2(1+2y)^{Q-1}- 2(1-y)^{Q-1} \right]\right] \nonumber \\
= \, \frac{33}{10} A(y)  \left[\left(1\,+ \, 2y\right)^{Q-1}(1+2y-2y) \,+ \,2\left(1\,- \, y \right)^{Q-1}(1-y+y)\right]\nonumber \\
= \, \frac{33}{10} A(y)  \left[\left(1\,+ \, 2y\right)^{Q-1} \,+ \,2\left(1\,- \, y \right)^{Q-1}\right].\nonumber \\
\mbox{ And (\ref{klsymugl210}) follows from  } 
\frac{ 2(1+2y)^{Q-1}- 2(1-y)^{Q-1}}{(1+2y)^{Q-1}  + 2(1-y)^{Q-1}} \, < \,  \frac{33}{10} A(y)  \label{klsymugl21} 
\end{eqnarray}
For $y=0$ both sides of (\ref{klsymugl21}) are equal to $0.$ 
We show that $33/10\cdot A'\,> \, $ the derivative with respect to $y$ of the left-hand-side of (\ref{klsymugl21}.)
By elementary calculation
\begin{eqnarray}
\frac{d}{dy} \frac{ 2(1+2y)^{Q-1}- 2(1-y)^{Q-1}}{(1+2y)^{Q-1} + 2(1-y)^{Q-1}}\,
=\,\frac{18 \cdot (Q-1)(1+y-2y^2)^{Q-2}}{\left[(1+2y)^{Q-1} + 2(1-y)^{Q-1}\right]^2}. \nonumber \\
\mbox{ We need to show }
\frac{33}{10}\frac{7}{10}Q\left[ (1+2y)^{Q-1} + 2(1-y)^{Q-1}\right]^2\,> \, 18 (Q-1)(1+y-2y^2)^{Q-2}. \nonumber \\
\mbox{ Enlarging the right-hand-side, } 1\le 1+y-2y^2 \mbox{ for } y \le 1/Q  \le 1/s \le 1/2  \mbox{(by  } Q(s)\ge s) \nonumber \\
\mbox{ we show } 33\cdot 7\left[ (1+2y)^{Q-1} + 2(1-y)^{Q-1}\right]^2\,> \,1800 (1+y-2y^2)^{Q-1} \,\nonumber \\
= \, 1800\left((1+2y)(1-y)\right)^{Q-1} \nonumber \\
\Longleftrightarrow 
231\left[(1+2y)^{2(Q-1)}\, + \, 4\left(( 1+2y)(1-y) \right)^{Q-1} \, + \, 4(1-y)^{2(Q-1)} \right]\,> \nonumber \\
> \, 1800\left((1+2y)(1-y)\right)^{Q-1} 
\Longleftrightarrow  \mbox{ (Division by } \left((1+2y)(1-y)\right)^{Q-1} \mbox{)}\nonumber \\
\Longleftrightarrow \, \, \left(\frac{1+ 2y}{1-y}\right)^{Q-1}\, + \, 4\,+ \, \,4 \left(\frac{1-y}{1+2y}\right)^{Q-1} \,> \, 1800/231. \nonumber \\
\mbox{ Rescaling the fraction to } x \mbox{ the preceding inequality  follows from } \nonumber \\
x+4\frac{1}{x} > 1800/231-4\,=  3.79 \dots\,  \mbox{ true for } \,  x > 0 .\nonumber 
\end{eqnarray}

\noindent
{\it Proof of (b).} We assume $0 \le y \le 1/(2Q)$ and $ 0 < z \le y.$ 
\begin{eqnarray}
A(y+z)= \frac{7}{10}Q\cdot (y+z)\, \,\,\,, \, A(y+z)\cdot (y+z) \,= \, \frac{7}{10}Q \cdot (y+z)^2 \nonumber \\
A(y+z)\cdot (y+z)\,+ \, A(y-z)\cdot (y-z)\,= \, \frac{7}{10}Q\cdot 2(y^2+z^2) \nonumber \\
\mbox{OPT}(z)\,= 
\, \left(1 \,+ \, L(A(y+z), s)\,+ \, L(A(y-z), s) \right) \cdot \nonumber \\ 
\cdot \left(\frac{1}{1 +  \frac{7}{10}Q\cdot 2 (y^2+z^2) }\right)^Q\cdot            
\left( (1+2y )^Q \,+ \, 2\cdot \left((1-y)^2 \,+ \, 3z^2 \right)^{Q/2}\right). \nonumber 
\end{eqnarray}

\begin{eqnarray}
\frac{d}{dz} \ln \mbox{OPT}(z) \,>=< \, 0 \Longleftrightarrow \nonumber \\
\frac{ \frac{7}{10}Q \cdot K(A(y+z), s)\,- \,\frac{7}{10}Q \cdot K(A(y-z), s) }
{1\,+ \, L(A(y+z), s)\,+ \, L(A(y-z), s) \,} \, \, - \, \, \frac{\frac{7}{10}Q 4 z}{1 +  \frac{7}{10}Q\cdot 2 (y^2+z^2) }
\,+ \, \nonumber \\
\frac{6z  \cdot ((1-y)^2+3z^2)^{Q/2-1}}{(1+2y )^Q \,+ \, 2\cdot ((1-y)^2+3z^2)^{Q/2} } \, \, >=<\,\,0. \nonumber 
\end{eqnarray}
The first term of the sum is obtained as the first term of  (\ref{EQKEY}.) 
The first and third term of the left-hand-side of the preceding inequality are $\ge 0$ for $0\le z \le y$  
whereas the second term is $\le 0.$ 

Analogously to  (\ref{klsymugl1}) and (\ref{klsymugl2})  $\frac{d}{dz} \ln \mbox{OPT}(z)\,<\,0$ is implied by
\begin{eqnarray}
\frac{ \frac{7}{10} Q  \left[ K(A(y+z), s)\,- \, K(A(y-z), s) \right]}                                           
{1\,+ \,  L(A(y+z), s)\,+ \, L(A(y-z), s) }\, \, <\, \, \frac{\frac{9}{10}\frac{28}{10}Q  z}{1 +  \frac{14}{10}Q (y^2+z^2) }
\label{klasymugl1} \\
\frac{6z  \cdot ((1-y)^2+3z^2)^{Q/2-1}}{(1+2y )^Q \,+ \, 2\cdot ((1-y)^2+3z^2)^{Q/2} }
\,< \, 
\frac{\frac{1}{10}\frac{28}{10}Q  z}{1 +  \frac{14}{10}Q (y^2+z^2) } \label{klasymugl2}
\end{eqnarray}

\noindent
{\it Proof of (\ref{klasymugl1}) for $y \le 1/(2Q)$ and $ s \ge 3.5.$}
The denominator of the right-hand-side fraction is maximal for 
$y=z=1/(2Q).$ In this case it is $1+ 7/(10Q)< 1+1/Q.$ We lower the denominator of the 
left-hand-side simply to  $1.$  The  claim follows from
\begin{eqnarray}
 K(A(y+z), s)\,- \, K(A(y-z), s)\,< \, \frac{\frac{18}{5}z}{1+\frac{1}{Q}} \nonumber 
\end{eqnarray}
The left-hand-side of the preceding 
inequality is convex   in $z$ for all $y<1/(2Q)$ (based on the convexity of $\exp(x)-\exp(-x).$)
For $z=0$ both sides are $=0.$ 
Therefore it is sufficient to show that the inequality holds for $z=y$ where $y \le 1/(2Q).$
Setting  $z=y$ yields    $K(A(y-z), s)=0$  and we show 
\begin{eqnarray}
 K(A(2y), s)\,\,< \, \frac{\frac{18}{5}y}{1+\frac{1}{Q}} \nonumber 
\end{eqnarray}
Again by convexity of the left-hand-side it is sufficient to show
the inequality for $y=1/(2Q).$ 
In this case we need to show 
\begin{eqnarray}
K ( A(1/Q),s )\,= \, \frac{\exp\left(\frac{7}{10}s\right)-1}{\exp(s)-1}\,<\,\frac{18}{10}\frac{1}{Q+1} \nonumber \\
\mbox{ By (\ref{BAKL}) we know  }\frac{\exp\left(\frac{7}{10}s\right)-1}{\exp(s)-1}\,\le \,  \exp\left( -\frac{3}{10}s\right). \nonumber \\
\mbox{ And } \exp\left( -\frac{3}{10}s\right) \,< \, \frac{18}{10}\frac{1}{Q+1}    \mbox{ holds (proof omitted) for  } s\ge 3.5. \nonumber 
\end{eqnarray}

\noindent
{\it Proof of (\ref{klasymugl2}) for $ s\ge 8.$  }
We show 
\begin{eqnarray}
\frac{(1+2y )^Q \,+ \, 2\cdot ((1-y)^2+3z^2)^{Q/2} }{6 z \cdot ((1-y)^2+3z^2)^{Q/2-1}}
\,> \, 
\frac{1 +  \frac{14}{10}Q (y^2+z^2) }{\frac{1}{10}\frac{28}{10}Q z }. \nonumber \\
\mbox{ Canceling } z \mbox{ in the denominator ,  setting } z=y \mbox{ on the right-hand-side, this follows from }\nonumber \\
\frac{(1+2y )^Q} {6  \cdot ((1-y)^2+3z^2)^{Q/2-1}}\,+ \, \frac{1}{3}((1-y)^2+3z^2)\,> \, \frac{1+\frac{14}{5}Qy^2}{\frac{28}{100}Q} \,= \, 
  \frac{100}{28Q} +   10y^2               \nonumber \\
\mbox{As } (1-y)^2+3z^2\,\le \, 1-2y+4y^2\,< 1 \mbox{ by }y\le 1/(2Q),  Q\ge s \ge 8 
\nonumber \\ \mbox{ this follows from } 
\frac{1}{6}(1+2y)^Q \,+ \, \frac{1}{3}(1-y)^2\,>\,  \frac{100}{28Q} +   10y^2   \nonumber 
\end{eqnarray}
The last inequality holds for $Q \ge 8, y \ge 0$ and then the  claim holds  as $Q \ge s.$

\subsection{Proof of Lemma \ref{lem2}}

\noindent
{\bf Lemma \ref{lem2} (repeated)} 
Let  $ s\, \ge 7 \, \,$ and  $  \frac{7}{20} \le A \le 1-\frac{1}{Q}.$ Then \\
OPT$(z):= $ OPT$\left(1, A,\, A, \, 1, 1/(2Q)+z, \, 1/(2Q)-z\, , s\right) \le 3 \,- \, \delta$  for $0 \le z \le 1/(2Q).$ 
\begin{proof}
 OPT$(z)\,= \,$ 
\begin{eqnarray}
\left( 1\,+ \ 2\cdot L(A, s) \right)
\left(\frac{1}{1+\frac{A}{Q}}\right)^Q 
\cdot \left[ \left( 1+\frac{1}{Q} \right)^Q \,+ \, 2\left(\left(1-\frac{1}{2Q}\right)^2+ 3z^2\right)^{1/2\cdot Q} \right] \nonumber 
\end{eqnarray}
is  increasing in $z.$ We  show the claim for $z=1/(2Q).$ Let from now on 
OPT$(A)\,= \, $ OPT$(1, A, A, 1, 1/Q, 0, s)\,= \, $ 
\begin{eqnarray}
\left( 1\,+ \ 2\cdot L(A, s) \right)
\left(\frac{1}{1+\frac{A}{Q}}\right)^Q 
\cdot \left[ \left( 1+\frac{1}{Q} \right)^Q \,+ \, 2\left( 1-\frac{1}{Q} +  \frac{1}{Q^2}\right)^{1/2\cdot Q} \right] \nonumber 
\end{eqnarray}

First, we show that OPT$(A)$  has
exactly one extremum in $0 \le A\le 1$  which is a minimum. 
\begin{eqnarray}
\frac{d}{dA} \ln \mbox{OPT}(A) \, > =< \,0 \Longleftrightarrow
\frac{2K(A, s)}{1+2L(A, s)} \,- \, \frac{\frac{1}{Q}}{1+ \frac{A}{Q}}\,>=< \,0 .\, \nonumber 
\end{eqnarray}
Concerning the first term of the preceding sum we refer to the explanation  following (\ref{EQKEY}.) 
For $A=0$  the first term  is $=0$ and the derivative is  $<0.$ 
For $A=1$  the first term is $=2/3,$ whereas the second term is
$1/(Q+1)\,<\,2/3$ for $Q>s>2,$ and the derivative is $>0.$  We show that 
the derivative is $=0$ for exactly one $ 0 < A < 1 $  which must be a minimum. 

The second fraction of the derivative is decreasing in $A.$  
We check that the first
fraction is increasing. Abbreviating  $L=L(A, s), \, L' \, = \, \frac{d}{dA}L(A, s) $  and analogously for $K,$ we get
\begin{eqnarray}
\frac{d}{dA} \frac{2K(A, s)}{1+2L(A, s)} \,> \,0  \Longleftrightarrow 
2K' (1+2L) \,> \, 2K2L' \Longleftrightarrow \nonumber \\
\mbox{ (Multiplication with }( \exp(s)-s-1)(\exp(s)-1), \mbox{division by } 2 \mbox{ and } s \mbox{.)} \nonumber \\
(\exp(s)-s-1)  \exp(sA) \,+ \,\exp(sA) 2( \exp(sA)-sA-1 ) \,> \, 2(\exp(sA)-1)^2 \nonumber \\ \Longleftrightarrow 
(\exp(s)-s-2sA)\exp(sA) \,>\,-\exp(sA)\,+ \, 2 \nonumber \\
\Longleftrightarrow \exp(s)-s-2sA \,> \,-1 + 2/\exp(sA)  \nonumber \\
 \mbox{ which is true for }  s >2, 0< A <1 \mbox{ by convexity of } 2/\exp(sA). \nonumber 
\end{eqnarray}

We need to show the claim for the boundary values $A=7/20 \, $ and $ A=1-1/Q.$ 
First,  $A=7/20:$  
\begin{eqnarray}
1 \,+ \, 2 L(A, s) \, \le 1 + 2 M(A, s) \, = \, 1\,+ \, 2\exp(-13/20 \cdot s) \, \, \, \mbox{ (by (\ref{BAKL}).)}\nonumber \\
\mbox{ With  the derivative of the logarithm and the Mean Value Theorem we can  show that } \nonumber \\
\left(\frac{1+\frac{1}{Q}}{1+\frac{A}{Q}} \right)^Q \mbox{ is increasing  in } Q  \mbox{ towards its limit } \exp(13/20) .
 \,  \frac{ 2\left( 1-\frac{1}{Q} \,+ \, \frac{1}{Q^2} \right)^{1/2\cdot Q}}{\left(1+ \frac{A}{Q}\right)^Q} \nonumber \\
\mbox{  is decreasing in }Q =Q(s)\ge 2 \mbox{ (proof by standard calculus methods)  and therefore }\nonumber \\
\mbox{ also in } s \mbox{ towards its limit  }   2\exp(-17/20).\mbox{ For } Q=7 \mbox{ we get a value } \le 0.9 \nonumber \\
\mbox{ Therefore,  for all } \ge s \, \ge 7 \mbox{( as } Q(s)\ge s)  \nonumber \\  \mbox{OPT}(A)\,
< \, (1+ \, 2\exp(-13/20 \cdot 7))( \exp(13/20)+0.9 ) \,= \, 2.87\dots  \, \,. \nonumber 
\end{eqnarray}
 
Now,  $A=1-1/Q:$  
\begin{eqnarray}
1\,+ \, 2L(A, s)\,\le \, 1 \,+ \, 2 M(A, s) \,= \,1\,+\, 2 \exp\left(-\frac{s}{Q}\right) \nonumber \\
 \,=\, 1+2\exp\left( -\frac{\exp(s)-s-1}{\exp(s)-1}\right) 
\mbox{decreasing in } s \mbox{ to } 1+2\exp(-1).  \nonumber \\
\mbox{ For } s=7 \mbox{ we get } \, 1+2L(A, s) \le   1.7404  \dots \nonumber \\
\left( \frac{1+\frac{1}{Q}}{1+\frac{A}{Q}} \right)^Q\,= \, \left( \frac{1+\frac{1}{Q}}{1\,+\,\frac{1}{Q}\,- \, \frac{1}{Q^2}}\right)^Q 
\mbox{  is decreasing in } Q =Q(s) \nonumber \\ \mbox{ (elementary proof omitted) and therefore  in}\,\, s  \mbox{ to } 1.\nonumber \\
\mbox{ For } Q=7 \mbox{ we get } 1.1344  \dots .\mbox{ As } Q(s)\ge s \, \mbox{this  bound applies to  } s=7 \mbox{, too.} \nonumber \\
\frac{ 2\left ( 1-\frac{1}{Q}+\frac{1}{Q^2} \right )^{1/2\cdot Q}}{\left(1+\frac{1}{Q}- \frac{1}{Q^2}\right)^Q} 
\mbox{is again decreasing (proof omitted) in } \nonumber \\ Q 
\mbox{ and } s \mbox{ to } 2\exp(-3/2). 
\mbox{ For } Q=7 \mbox{ we get }  0.564 \dots \nonumber \\
\mbox{Altogether for } Q(s)\ge s \ge 7 \nonumber \\ \mbox{OPT}(A) \le 1.741\cdot(1.135+ 0.565)\,= 2.9597 . \nonumber 
\end{eqnarray}

\end{proof}

\subsection{ Proof of Lemma \ref{lem3}}

\noindent
{\bf Lemma \ref{lem3} (repeated)}
Let $\, s\, \ge \, 7$ and $  1/(2Q) \le C \le 1/2.$  
Then  \\
OPT$(z):= $ OPT$\left(1, 1-1/Q,\, 1-1/Q, \,1,  C+z, \, C-z\,, s\right)\, \le  3 \,- \, \delta$  for $0 \le z \le C.$ 
\begin{proof} 
We abbreviate  $A=1-1/Q.$ First, analogously  to the proof of Lemma \ref{lem2} we can restrict attention to $z=C.$
OPT$(z)\,= \, $
\begin{eqnarray}
=\,\left(1+2L(A, s)\right) \left( \frac{1}{1+2AC}\right)^Q \left[(1+2C)^Q\,+ \, 2((1-C)^2+3z^2)^{1/2\cdot Q}\right] \nonumber \\
\mbox{ Let from now on OPT}(C)\,= \,  \mbox{OPT}(1, A,A, 1, 2C, 0, s) \,= \,\nonumber \\
=\,\left(1+2L(A, s)\right) \left( \frac{1}{1+2AC}\right)^Q  \cdot  \left[(1+2C)^Q\,+ \, 2(1+ 4C^2-2C)^{1/2\cdot Q}\right] \nonumber 
\end{eqnarray}

 OPT$(C)$ has exactly one extremum, which is a minimum for $0 \le C \le 1.$ 
\begin{eqnarray}
\frac{d}{dc} \ln \mbox{OPT}(C) \,>=<\, 0 \Longleftrightarrow \nonumber \\
 - \frac{2A}{1+2AC} \,+ \, \frac{2(1+2C)^{Q-1} \,+ \, (8C-2)(1+ 4C^2-2C)^{1/2 \cdot Q-1}}{(1+2C)^Q \,+ \, 2(1+ 4C^2-2C)^{1/2\cdot Q}}\, >=< \,0
\Longleftrightarrow \nonumber \\
2A \left((1+2C)^Q\,+ \, 2(1+ 4C^2-2C)^{1/2Q}\right) \,- \nonumber \\ 
-\,2A C \left(2(1+2C)^{Q-1}\,+ \, (8C-2)(1+ 4C^2-2C)^{1/2 \cdot Q-1}\right)\, = \nonumber \\
\,= \, 2A\left[(1+2C)^{Q-1}\,+ \, (2-2C)(1+ 4C^2-2C)^{1/2 \cdot Q-1}\right]\, <\,\, = \, \, > \nonumber \\
\,<\,\,=\,\,  >\,2(1+2C)^{Q-1}\,+ \, (8C-2)(1+ 4C^2-2C)^{1/2 \cdot Q-1}  \Longleftrightarrow \nonumber \\
2A\,<\, \, =\, \, > \, \frac{2(1+2C)^{Q-1}\,+ \, (8C-2)(1+ 4C^2-2C)^{1/2 \cdot Q-1}}{(1+2C)^{Q-1}\,+(2-2C)(1+ 4C^2-2C)^{1/2 \cdot Q-1}} \, \Longleftrightarrow \, \nonumber \\
A\,<\, \, =\, \, > \, \frac{(1+2C)^{Q-1}\,+ \, (4C-1)(1+ 4C^2-2C)^{1/2 \cdot Q-1}}{(1+2C)^{Q-1}\,+(2-2C)(1+ 4C^2-2C)^{1/2 \cdot Q-1}} \nonumber 
\end{eqnarray}
For $C=0$ the right-hand-side fraction is equal to $0<A$ and OPT$(C)$ is decreasing.
For $C=1$ the right-hand-side fraction  is greater than  $1>A$ and OPT$(C)$ is increasing. 

Next we show that the preceding fraction is increasing in $0<C<1,$ 
and equality is attained for only one $C$ which must be a minimum.  
\begin{eqnarray}
\mbox{ Rewriting } 4C-1 \,= \, (2-2C)+6C-3  \mbox{ the fraction is rewritten as } \nonumber \\
1 \,+ \, \frac{(6C-3)(1+ 4C^2-2C)^{1/2 \cdot Q-1}}{(1+2C)^{Q-1}\,+(2-2C)(1+ 4C^2-2C)^{1/2 \cdot Q-1}} \nonumber \\
\mbox{ Rescaling }1/2\cdot Q-1 \mbox{ to } Q   \mbox{ ( then } Q-1 \mbox{ scales to } 2Q+1 ) \mbox{ and } 2C \mbox{ to }C
\mbox{ we get } \nonumber \\
1 \,+ \, \frac{(3C-3)(1+ C^2-C)^Q}{(1+C)^{2Q+1}\,+(2-C)(1+ C^2-C)^Q} \, \,  \nonumber \\
\mbox{ Dividing through } 3(C-1)(1+ C^2-C)^Q \mbox{ the  preceding fraction is certainly increasing  if } \nonumber \\
\frac{(1+C)^{2Q+1}}{3(C-1)(1+ C^2-C)^Q} \mbox{ and } \frac{2-C}{3(C-1)} \mbox{ are both decreasing for } 0 <C <2, C \neq 1. \nonumber 
\end{eqnarray}
The second fraction is easily seen to be decreasing. We show that the inverse of the first fraction is
increasing. The numerator of its derivative is
\begin{eqnarray}
\left[ (1+C^2-C)^Q \,+ \,( C-1)(2C-1)Q(1+C^2-C)^{Q-1}\right] \cdot (1+C)^{2Q+1}\,  \,- \, \nonumber \\
- \, (C-1)(1+C^2-C)^Q\cdot  (2Q+1)(1+C)^{2Q} \,\, 
\,= \, (1+C)^{2Q}(1+C^2-C)^{Q-1}\cdot \nonumber \\
\left[(1+C)(1+C^2-C) \,+ \, (1+C)(C-1)(2C-1)Q \,- \, (2Q+1)(C-1)(1+C^2-C)\right] \nonumber \\
\mbox{ The expression in square brackets can be rewritten as} \nonumber \\
(1+C)(C-1)(2C-1)Q\,- \, 2Q(C-1)(1+C^2-C)\,+\, (-C+1+C+1)(1+C^2-C) \nonumber \\ 
=\,\, Q(1-C)^2\,+ \, 2(1+C^2-C)\,>\, 0   \nonumber 
\end{eqnarray}

Now it is sufficient to show the claim for the boundary values, $C=1/(2Q)$ and $C=1/2.$ 
The first case is contained in Lemma \ref{lem2}. Let $C=1/2.$ We proceed as in the proof of Lemma \ref{lem2},
case $A=1-1/Q.$ 
\begin{eqnarray}
1+2L(A, s) \le 1.7404 \mbox{ for } s\ge 7 \nonumber \\
\left( \frac{1+2C}{1\,+\,2CA} \right)^Q\,= \, \left( \frac{2}{2\,- \, \frac{1}{Q}}\right)^Q 
\mbox{  is decreasing in } Q =Q(s) \nonumber \\ \mbox{ (elementary proof omitted) and therefore  in}\,\, s  \mbox{ to } \exp(-1/2).
\nonumber \\
\mbox{ For } Q=7 \mbox{ we get } 1.67993\dots  .\mbox{ As } Q(s)\ge s \, \mbox{this  bound applies to  } s=7 \mbox{, too.} \nonumber \\
\frac{ 2(1+ 4C^2-2C)^{1/2\cdot Q}}{(1+2AC)^Q } \,= \,\frac{2}{(2-\frac{1}{Q})^Q}  \mbox{ decreasing to } 0 \nonumber \\
\mbox{ For } Q=7 \mbox{ we get } 0.02624\dots \nonumber \\
\mbox{ Altogether OPT}(C)  \le 1.75\cdot(1.68+0.027)=2.98 \mbox{ for } s\ge 7.\nonumber 
\end{eqnarray}

\end{proof}

\subsection{Proof of Lemma \ref{lem4}}
\noindent
{ \bf Lemma \ref{lem4} (repeated)} 
Let $s\ge 15$ and $  A(x)\,=\, A(x, s) \, := \, 1+ 7/(10Q)\cdot x \,- \, 7/(10Q).$ \\
(a) OPT$(y):= $OPT$\left(1, A(y), A(y), 1, y, y, s \right) $ is strictly increasing in $4/10 \le y < 1.$ The final value is 
OPT$(1)=3.$ \\
(b) Given $4/10 \le y \le 1,$ OPT$(z):=$OPT$\left(1, A(y+z), A(y-z),1,  y+z, y-z, s \right)   $ 
is decreasing in $0 \le z \le \min\{y, 1-y\}.$ \\

\noindent
{\it Proof of (a).}
We have OPT$(y) \,= \,$
\begin{eqnarray} 
 \left( 1+2L( A(y), s) \right) 
\left(\frac{1}{1+2A(y)\cdot y}\right)^Q  \left(\left(1\,+ \, 2y\right)^Q \,+ \,2\left(1\,- \, y \right)^Q \right) \nonumber \\
\mbox{ We write OPT}_1(y)=  1+2L(A(y), s).  \mbox{  Clearly  OPT}(1)\,= \, 3 \nonumber \\
\mbox{ We have } A' := \frac{d}{dy} \,A(y)\,= \,\frac{7}{10}\frac{1}{Q}.  \nonumber \\
\frac{d}{dy} \ln \mbox{ OPT }(y) \,>=< \, 0 \Longleftrightarrow    \mbox{(See  comment to  (\ref{EQKEY}.))} \nonumber \\
\frac{A' \cdot 2 \cdot K(A(y), s) }{\mbox{ OPT}_1(y)}\,- \, 
\frac{2A(y)\,+ \, 2A'\cdot y}{1+2A(y) \cdot y}\,+ \,\frac{2(1+2y)^{Q-1}- 2(1-y)^{Q-1}}{\left(1\,+ \, 2y\right)^Q \,+ \,2\left(1\,- \, y \right)^Q }\,>=<\,0   \nonumber 
\end{eqnarray} 
Observe that the first and third term of the preceding sum are $\ge 0$ for $0 \le y \le 1$ whereas the second term is $\le 0.$ 

We have that  $\frac{d}{dy} \ln \mbox{ OPT }(y)>0$  if the following two inequalities  both hold:
\begin{eqnarray}
\frac{A' \cdot 2 \cdot K(A(y), s) }{\mbox{ OPT}_1(y)}\,> \, \frac{\, 2A'\cdot y}{1+2A(y)\cdot y} \label{grosymugl1} \\
 \frac{2(1+2y)^{Q-1}- 2(1-y)^{Q-1}}{\left(1\,+ \, 2y\right)^Q \,+ \,2\left(1\,- \, y \right)^Q }\,>                                   
\frac{2A(y)}{1+2A(y) \cdot y} \, \, \label{grosymugl2}
\end{eqnarray}
Note that for $y=1$ both sides of the first inequality are equal to  $7/(10Q) \cdot 2/3$ and of the second 
inequality $2/3.$ Therefore the derivative of OPT$(y)$ is $=0$ for $y=1.$ \\

\noindent
{\it Proof of (\ref{grosymugl1}) for $1>y \ge \ 0\, \, ,\, \, \,  s \ge 4$ .} 
Let $ K= K(A(y), s)$  and $  L=L(A(y), s) .$
\begin{eqnarray}
\mbox{ As } A'>0 \mbox{ we need to show  } \frac{K}{1+2L} \,> \, \frac{y}{1+2A(y)\cdot y}.\nonumber \\
\Longleftrightarrow K\,+ \, 2K\cdot A(y)y \,- \, 2L\cdot y\, > \, y \nonumber \\
\mbox{ As }  K \ge L  \, \, \mbox { by (\ref{BAKL}) this follows from }  \nonumber \\ 
K\left( 1\, +\, 2A(y)\cdot y\, - 2y \right) \, \,
=\, K \left( 1+ 2\frac{7}{10Q}y^2 -   2\frac{7}{10Q}y)\right) \, > \, y   \label{grosymugl11}
\end{eqnarray}
For $y=1$ both sides of (\ref{grosymugl11}) are $=1.$ For $y=0$ (\ref{grosymugl11}) holds as $K>0$ in this case. 

$K$ considered as a function in $y$ is  convex, increasing and $> 0.$ 
The second term on the left-hand-side of (\ref{grosymugl11}), $ 1- 2\frac{7}{10Q}y^2 + 2\frac{7}{10Q} y, $ is convex, $>0,$ and increasing for $y>1/2.$
Therefore the left-hand-side of (\ref{grosymugl11}) is 
convex for $1/2<y<1.$ We next show that the derivative of the left-hand-side at $y=1$ is $<1.$
This implies that (\ref{grosymugl11})  holds for $1/2 \le y <1.$ 

\begin{eqnarray}
\frac{d}{dy} K \left( 1 \,+\,2\frac{7}{10Q}y^2\,- \,2\frac{7}{10Q}y)\right)\,=  \, 
\frac{s\cdot \frac{7}{10Q}\cdot \exp(sA(y))}{\exp(s)-1} \nonumber \\
\cdot 
\left(1\, \,\,+ 2\frac{7}{10Q}y^2\,- \,2\frac{7}{10Q}y\right) 
\, + \, \frac{\exp(sA(y))-1}{\exp(s)-1}\cdot\left( 4\frac{7}{10Q}y\,- \,2\frac{7}{10Q}\right). \nonumber \\
\mbox{ Plugging in }\,\, y=1 \, \,  \mbox{ yields } 
\frac{7}{10Q} \left(\frac{s \exp(s)}{\exp(s)-1} + 2\right)  \label{grosymugl12}  
\end{eqnarray}
For $s=4$ (\ref{grosymugl12}) is $0.9837 \dots< 1. $ As (\ref{grosymugl12}) is in decreasing in $s $
(proof omitted) 
(\ref{grosymugl11}) holds for all $s \ge 4$ and $1/2 \le y <1.$ 

$ 1- 2\frac{7}{10Q}y^2 + 2\frac{7}{10Q} y$    is decreasing for $y<1/2.$ 
Therefore, for $0 \le y \le 1/2,$ we  can bound the left-hand-side of 
(\ref{grosymugl11}) from  below by  
\begin{eqnarray}
K \left[ \left(  1 + 2\frac{7}{10Q}y^2  - 2\frac{7}{10Q}y \right)_{y=1/2} \right] \nonumber
\end{eqnarray}
This  function (the argument $y$ occurs only in $K$) is convex in $y$ . For $y=1/2$ it is $>y$ by the previous argument. 
For $y=1$ it is $<y.$ Therefore it is $>y$ for $0 \le y \le 1/2.$ The claim is shown. \\

\noindent
{\it Proof of (\ref{grosymugl2}) for $y \ge 4/10 $ and $ s \ge  3.5 .$ } 
Inequality (\ref{grosymugl2}) is equivalent to 
\begin{eqnarray}
2(1+2y)^{Q-1}- 2(1-y)^{Q-1} \, > \,  \nonumber \\  > 2 A(y)\left[ \left( 1\,+ \, 2y\right)^Q \,+ \,2\left(1\,- \, y \right)^Q \,- \, y\cdot \left[ 2(1+2y)^{Q-1}- 2(1-y)^{Q-1} \right]\right] \,= \, \nonumber \\
= \, 2A(y)\left[ (1+2y)^{Q-1}(1+2y- 2y) \,+\, 2(1-y)^{Q-1} (1-y +y)\right) \nonumber \\ \,= \,  2A(y)\left[ (1+2y)^{Q-1} + 2(1-y)^{Q-1}\right]
\Longleftrightarrow   \frac{ (1+2y)^{Q-1}- (1-y)^{Q-1}}{(1+2y)^{Q-1} + 2(1-y)^{Q-1}} \,> \,  A(y) \label{grosymugl21} 
\end{eqnarray}
For $y=1$ both sides of (\ref{grosymugl21}) are equal to $1.$ For $y<1$ (\ref{grosymugl21}) can be rewritten as
\begin{eqnarray}
\left(\frac{1+2y}{1-y}\right)^{Q-1}\,> \, \frac{2A(y)+1}{1-A(y)}.
\mbox{ With } y=\frac{4}{10} 
\mbox{ this becomes } 3^{Q-1}\, > \, \frac{30}{7}Q \,- \, 2 . \nonumber 
\end{eqnarray}
The preceding inequality holds for $Q>s \ge 3.5.$  and we have the claim for $y=4/10.$ 

To show the claim for $4/10<y<1$ we show that the left-hand-side of (\ref{grosymugl21}) 
is concave in $y.$ The derivative of the left-hand-side is
\begin{eqnarray}
9(Q-1) \frac{(1+y-2y^2)^{Q-2}}{\left[ (1+2y)^{Q-1} \,+ \, 2(1-y)^{Q-1} \right]^2 } \nonumber 
\end{eqnarray}
This is a decreasing function in $y\ge 4/10$ because the numerator is 
decreasing in this case whereas the denominator is increasing and $>0.$ \\

\noindent
{\it Proof of (b).}
Some preparatory calculations:
\begin{eqnarray}
A(y+z)\,= \, A(y)\,+ \, \frac{7}{10Q} z \, \,, \,\, \,  \, A(y-z)\,= \, A(y)\,- \, \frac{7}{10Q} z \, \, \,\nonumber \\
A(y+z) \cdot (y+z) \,= \, A(y)y +A(y)z +\frac{7}{10Q}zy + \frac{7}{10Q}z^2 \nonumber \\
A(y-z) \cdot (y-z) \,= \, A(y)y -A(y)z - \frac{7}{10Q}zy + \frac{7}{10Q}z^2 \nonumber \\
A(y+z) \cdot (y+z) \,+\, A(y-z) \cdot (y-z) \, \,= \, 2A(y) \cdot y  +\frac{14}{10Q}z^2 \nonumber 
\end{eqnarray}
We denote 
\begin{eqnarray}
\mbox{OPT}_1(z)\,= \,  1 \,+ \,L(A(y+z), s) \,+ \, L(A(y-z), s) \nonumber \\
\mbox{Then OPT}(z)\,= \,  \mbox{OPT}_1(z)\cdot  \hfill \nonumber \\
 \cdot \left(\frac{1}{1 +  2A(y) \cdot y  +\frac{14}{10Q}z^2 }\right)^Q\cdot            
 \left( (1+y )^Q \,+ \, 2\cdot \left((1-y)^2 \,+ \, 3z^2 \right)^{Q/2}\right). \nonumber 
\end{eqnarray}

We proceed to show that $\frac{d}{dz} \ln$ OPT$(z)\,< \,0$  for $z>0.$ Some derivatives first.
\begin{eqnarray}
\frac{d}{dz} \,A(y+z)\,= \,\frac{7}{10}\frac{1}{Q}, \, \mbox{\quad \quad }
\frac{d}{dz} \,A(y-z)\,= \,- \,\frac{7}{10}\frac{1}{Q}, \, \nonumber \\
\frac{d}{dz} \left(1 +  2A(y) \cdot y+\frac{14}{10Q}z^2 \right)\,   \,= \, \frac{28}{10Q}z \,\nonumber \\
\frac{d}{dz} \left( ( 1 + y )^Q \,+ \, 2\cdot ((1-y)^2+3z^2)^{Q/2} \right) \,= \, 6z \cdot Q \cdot ((1-y)^2+3z^2)^{Q/2-1}. \nonumber \\
\frac{d}{dz} \ln \mbox{OPT}(z) \,>=< \, 0 \Longleftrightarrow  \mbox{ (Recall comment to (\ref{EQKEY}).)} \nonumber \\
\frac{ \frac{7}{10Q} K(A(y+z), s)\,-\frac{7}{10Q} K(A(y-z), s) }
{\mbox{OPT}_1(z)  } \, \, - \, \, \frac{\frac{28}{10Q}z}{1+2yA(y)+\frac{14}{10Q}z^2}
\,+ \, \nonumber \\
\frac{6z  \cdot ((1-y)^2+3z^2)^{Q/2-1}}{(1+y )^Q \,+ \, 2\cdot ((1-y)^2+3z^2)^{Q/2} } \, \, >=<\,\,0. \nonumber 
\end{eqnarray}
Observe that the first and third term of the 
preceding inequality are $\ge 0$ for $0 \le z \le \min \{y, 1-y \},$ whereas the second term is $\le 0.$ 

We have that  $\frac{d}{dz} \ln \mbox{ OPT }(z)<0$  if the following two inequalities  both hold:
\begin{eqnarray}
\frac{ \frac{7}{10Q}  \left( K(A(y+z), s)\,-\, K(A(y-z), s) \right) }                                        
{ \mbox{OPT}_1(z)   }                            \, \, <\, \, 
\frac{\frac{11}{10Q}z}{1+2yA(y)+
\frac{14}{10Q}z^2}    \label{groasymugl1} \\
\frac{6z  \cdot ((1-y)^2+3z^2)^{Q/2-1}}{(1+y )^Q \,+ \, 2\cdot ((1-y)^2+3z^2)^{Q/2} }
\,< \, 
\frac{\frac{17}{10Q}z}{1+2yA(y)+\frac{14}{10Q}z^2} \label{groasymugl2}
\end{eqnarray}
Note that for $z=0$ both sides of the preceding inequalities are equal to  $0$ and the derivative
of $\ln \mbox{OPT}(z)$  is $=0.$ 
Moreover, we have  $1+2yA(y)+\frac{14}{10Q}z^2\,\le \,1+2y $ and the inequalities follow when they are
shown with the  denominator  $1+2y$ in the   right-hand-side fraction. To get  this,  observe that 
\begin{eqnarray}
2yA(y)+ \frac{14}{10Q}z^2\,= \, 2y+ \frac{14}{10Q}\left(y^2-y+z^2\right) \le 2y , \nonumber \\
\mbox{as } z\le \min\{y, 1-y\} \mbox{ we have } z^2 \le y(1-y) \mbox{ or } y(y-1)+z^2\le 0. \nonumber 
\end{eqnarray}

\noindent
{\it Proof of (\ref{groasymugl1}) for $0< z < \min \{y, 1-y \}  ,0 \le y \le 1,  s \ge  5 $ .} 
We enlarge the left-hand-side of   (\ref{groasymugl1}) first:
\begin{eqnarray}
 K(A(y+z), s)\,-\, K(A(y-z), s) \,= \, 
\,\frac{1}{\exp(s)-1}\left( \exp( A(y+z)\cdot s)\,- \, \exp( A(y-z)\cdot s)\right) \nonumber \\
= \, \,   \frac{ \exp(A(y)s) }{  \exp(s)-1 } \left[\exp\left(\frac{7}{10Q}sz \right) \, - \, 
\exp\left(-\, \frac{7}{10Q}sz\right) \right]      \nonumber 
\end{eqnarray}

\begin{eqnarray}
\mbox{OPT}_1(z)\,= \, 1\,+ \,L(A(y+z), s) \,+ \, L(A(y-z), s) \, \, \,= \, 1 +  \frac{1}{\exp(s)-s-1}\cdot \nonumber \\
\cdot  [\exp(A(y+z)s) \,- \, A(y+z)s\,-1 \,+  \,
\exp(A(y-z)s)  \,- \,    A(y-z)s\,-1 ] \nonumber \\
\ge \,   \mbox{( As } A(y+z), A(y-z)\le 1 \mbox{.)}  \nonumber \\
 1 \,\, + \,\, \frac{1}{\exp(s)-1} [\exp(A(y+z)s) \,+ \, \,
\exp(A(y-z)s)  \,- \, 2s\,-2 ] \, \,  = \, \frac{1}{\exp(s)-1} \cdot  \nonumber \\
\cdot \left[\exp(s)\,-2s\,-3\, + \, \exp(A(y)s)\left( \exp\left(\frac{7}{10Q}sz \right)  \,+ \, 
\exp\left(-\frac{7}{10Q}sz \right) \right) \right] \nonumber \\
\ge \,  \mbox{ (As } A(y)s \le s \mbox{ and } s\ge 2 \mbox{ so that } \exp(s)-2s-3>0\mbox{.)}\nonumber\\
\frac{\exp(A(y)s)}{\exp(s)-1} \left[\frac{\exp(s)\,-2s\,-3}{\exp(s)}\, + \,  \exp\left(\frac{7}{10Q}sz\right)   \,+ \, 
\exp\left(-\frac{7}{10Q}sz \right)  \right] \nonumber \\
\ge \, \frac{\exp(A(y)s)}{\exp(s)-1} \left[ 0.9 \, + \,  \exp\left(\frac{7}{10Q}sz\right)   \,+ \, 
\exp\left(-\frac{7}{10Q}sz \right)  \right],  \nonumber 
\end{eqnarray}
as $(\exp(s)-2s-3)/\exp(s) \ge 0.9$  for $s \ge 5.$ 
The denominator of the right-hand-side of (\ref{groasymugl1}) is  enlarged by $1+2y\le 3.$ We set 
\begin{eqnarray}
u= \exp\left(\frac{7}{10Q}sz\right)>1 \mbox{ and show (simple algebra from (\ref{groasymugl1})) } 
\frac{u-\frac{1}{u}}{0.9+ u+ \frac{1}{u}} \, < \, \frac{11}{3\cdot 7}z \nonumber \\
\mbox{ We have } z= (\ln u)\frac{10}{7}\frac{Q}{s}\, >\, (\ln u)\frac{10}{7} \mbox{ ( by } Q>s.) \nonumber \\
\mbox{ Therefore it is enough  to show  } \frac{u -\frac{1}{u}}{0.9+ u+ \frac{1}{u}}\,<\, (\ln u)\frac{10}{7}\frac{11}{21} \nonumber \\
\mbox{ Elementary means show that this is true for }  u>1. \nonumber 
\end{eqnarray}

\noindent
{\it Proof of (\ref{groasymugl2} ) for $ s\ge 15, 1 \ge y \ge 2/10 \, , \, 0< z \le \min \{y, 1-y \} .$} 
Inequality (\ref{groasymugl2}) follows from 
\begin{eqnarray}
\frac{6z  \cdot ((1-y)^2+3z^2)^{Q/2-1}}{(1+y)^Q}
\,< \, 
\frac{\frac{17}{10 Q}z}{1+2y} \nonumber \\
\Longleftrightarrow 
60Q(1+2y)((1-y)^2+3z^2)^{Q/2-1} \, < \, 17(1+y)^Q \label{groasymugl21} 
\end{eqnarray}

For $ y \le 1/2 $ we have $ z \le y  $  and (\ref{groasymugl21}) follows from 
\begin{eqnarray}
60Q(1+2y)(1-2y+4y^2)^{Q/2-1} \, < \,17 (1+y)^Q  \nonumber 
\end{eqnarray}
The preceding inequality holds for $ Q \ge s \ge 15 $  and $ 1/2\ge  y \ge 2/10$ (proof omitted.)  

For $y \ge 1/2$ we have $z \le 1-y$ and (\ref{groasymugl21}) follows from
\begin{eqnarray}
60Q(1+2y)(4(1-y)^2)^{Q/2-1} \, < \, 17 (1+y)^Q  \nonumber 
\end{eqnarray}
This inequality holds for $Q\ge s \ge 10$ and $y\ge 1/2$   (details omitted.)

\newpage

\section{Proof of Lemma \ref{LOCMAX} and 
Theorem \ref{LAPLA}} \label{PROLAPLA}


We consider   $\Psi( \bar{\omega}\, ,\, \bar{\lambda} )\, =\, \Psi( \bar{\omega}\, ,\, \bar{\lambda}\, ,\, \bar{a} \,,\, \bar{c}  \, )$
as function of $w_i, \lambda_i,$ $i=1, 2$ in a neighborhood of $(\omega_1, \omega_2)=(\lambda_1, \lambda_2)=(1/3, 1/3).$  
The parameters  $a_i, c_i$  are given by 
$ Q(a_i) = \lambda_i k\gamma /\omega_i,$ and $ c_0=1, \, \,\, \, R(c_1, c_2)= (\lambda_1k, \lambda_2k).$ 
Subsection \ref{LOLICO} shows that this is well defined and $a_i, c_i$ is  differentiable in $\lambda_i, \omega_i.$
For $\lambda_i=1/3, \omega_i=1/3$ we have $a_i=s, c_i=1 \, (Q(s)=k\gamma$ defining $s.$) 
We show that the partial derivatives of $\ln \Psi ( \bar{\omega}\, ,\, \bar{\lambda} )$ are $0$ for $\omega_i= \lambda_i=1/3$
and the Hessian matrix is negative definite. This implies Lemma \ref{LOCMAX}.

For $i=1, 2$ the  first derivatives are, with
$a_i', c_i' $ denoting the right derivatives of $a_i, c_i$  resp. and recalling that 
$Q(x)= \frac{xq'(x)}{q(x)}, q(x)=\exp(x)-x-1, R(x_1,x_2)= $ \\ 
$\left( \frac{x_1r_{x_1}(1, x_1, x_2)}{r(1, x_1, x_2)}\,, \, \frac{x_2r_{x_2}(1, x_1, x_2)}{r(1, x_1, x_2)} \right) $ 
\begin{eqnarray}
\frac{d \ln \Psi(\bar\omega, \bar\lambda)}{d \omega_i} &=&
    - \ln q(a_0) + \omega_0 \frac{a_0' q'(a_0)}{q(a_0)} + \ln \omega_0 + 1 + \nonumber \\
&& \quad    + \ln q(a_i)  + \omega_i \frac{a_i' q'(a_i)}{q(a_i)} - \ln \omega_i -1 - \nonumber \\
&& \quad    - k \gamma \lambda_0 \frac{a_0'}{a_0} - k \gamma \lambda_i \frac{a_i'}{a_i} \nonumber \\
&=& \ln \omega_0 - \ln \omega_i + \ln q(a_i) - \ln q(a_0) \quad  (\textrm{ using  } Q(a_i) = k\gamma\lambda_i/\omega_i)  \label{eqn:diffWi}. 
\end{eqnarray}

\begin{eqnarray}
\frac{d \ln \Psi(\bar\omega, \bar\lambda )}{d \lambda_i} &=&
    \omega_0 \frac{a_0' q'(a_0)}{q(a_0)} + \omega_i \frac{a_i' q'(a_i)}{q(a_i)} + \nonumber \\
&& \quad k \gamma \Bigg( - \ln \lambda_0 - 1 + \ln a_0 - \lambda_0 \frac{a_0'}{a_0} + \nonumber \\
&& \quad \quad          + \ln \lambda_i + 1 - \ln a_i - \lambda_i \frac{a_i'}{a_i} - \nonumber \\
&& \quad \quad          - \ln c_i - \lambda_1 \frac{c_1'}{c_1} - \lambda_2 \frac{c_2'}{c_2}
                \Bigg) + \nonumber \\
&& \quad \gamma \frac{c_1' r_{c_1}(1, c_1,c_2) + c_2' r_{c_2}(1, c_1, c_2)}{r(1, c_1, c_2)} \nonumber \\
&=& k \gamma ( \ln \lambda_i - \ln \lambda_0 + \ln a_0 - \ln a_i - \ln c_i )  \label{eqn:diffLi} \\
 & & \quad  ( \mbox{ using  } \, \, \, R(c_1,c_2) = (k\lambda_1, k\lambda_2), Q(a_i) = k\gamma\lambda_i/\omega_i) ) \nonumber 
\end{eqnarray}

For $\bar\lambda = \bar\omega = (1/3, 1/3)$ the terms in 
(\ref{eqn:diffWi}) and (\ref{eqn:diffLi}) yield $0$.

The second derivatives (with $i,j \in \{1,2\}, i \not= j)$ are (observe that some of the subsequent terms are
equal as the derivative does not depend on the ordering of the variables)
\begin{eqnarray}
\frac{d^2 \ln \Psi(\bar\omega, \bar\lambda)}{d \lambda_i,\lambda_i} &=&
    k \gamma \left( \frac{1}{\lambda_i} + \frac{1}{\lambda_0} +
        \frac{a_0'}{a_0} - \frac{a_i'}{a_i} - \frac{c_i'}{c_i} \right)
            \label{eqn:diffLiLi} \\
\frac{d^2}{d \lambda_i,\lambda_j} &=&
    k \gamma \left( \frac{1}{\lambda_0} + \frac{a_0'}{a_0} - \frac{c_i'}{c_i} \right)
            \label{eqn:diffLiLj} \\
\nonumber \\
\frac{d^2\ln \Psi(\bar\omega, \bar\lambda)}{d \omega_i,\omega_i} &=&
    - \frac{1}{\omega_0} - \frac{1}{\omega_i} + \frac{a_i' q'(a_i)}{q(a_i)}
    - \frac{a_0' q'(a_0)}{q(a_0)}
            \label{eqn:diffWiWi} \\
\frac{d^2\ln \Psi(\bar\omega, \bar\lambda)}{d \omega_i,\omega_j} &=&
    - \frac{1}{\omega_0} - \frac{a_0' q'(a_0)}{q(a_0)}
            \label{eqn:diffWiWj}
\end{eqnarray}

\begin{eqnarray}
\frac{d^2\ln \Psi(\bar\omega, \bar\lambda)}{d \omega_i, \lambda_i} &=&
    \frac{a_i' q'(a_i)}{q(a_i)} - \frac{a_0' q'(a_0)}{q(a_0)}   \label{eqn:diffWiLi} \\
\frac{d^2}{d \omega_i, \lambda_j} &=&
    - \frac{a_0' q'(a_0)}{q(a_0)}   \label{eqn:diffWiLj} \\
\nonumber \\
\frac{d^2\ln \Psi(\bar\omega, \bar\lambda)}{d \lambda_i, \omega_i} &=&
    k \gamma \left( \frac{a_0'}{a_0} - \frac{a_i'}{a_i} \right) \label{eqn:diffLiWi} \\
\frac{d^2 \ln \Psi(\bar\omega, \bar\lambda)}{d \lambda_i, \omega_j} &=&
    k \gamma \frac{a_0'}{a_0}   \label{eqn:diffLiWj}
\end{eqnarray}
In (\ref{eqn:diffLiLi}) - (\ref{eqn:diffLiWj}) we need  several 
$a_i'$ and $c_i'.$  We  get these
from the defining equations  $Q(a_i)$ and $R(c_1,c_2).$ 
\paragraph{Derivative of $a_0.$}
By $Q(a_i) = k \gamma \lambda_i/\omega_i$ we have
\begin{equation}
\frac{a_0 q'(a_0)}{q(a_0)} = \frac{k \gamma \lambda_0}{\omega_0} \quad \Leftrightarrow \quad
\frac{a_0}{k \gamma \lambda_0} = \frac{q(a_0)}{\omega_0 q'(a_0)}  \nonumber 
\end{equation}
Taking the derivative of both sides wrt.  $\omega_i$ yields
\begin{eqnarray}
\frac{a_0'}{k \gamma \lambda_0} =
    \frac{a_0' q'(a_0) \omega_0 q'(a_0) -
        q(a_0) \left( -q'(a_0) + \omega_0 a_0' q''(a_0) \right)}{\omega_0^2 q'(a_0)^2} \nonumber \\
=\, \frac{a_0'}{\omega_0}\, +\, \frac{q(a_0)}{\omega_0^2 q'(a_0)} \,- \, 
\frac{ a_0' q''(a_0)q(a_0)}{\omega_0 q'(a_0)^2 } \nonumber \\
\Longleftrightarrow \frac{a_0' q'(a_0)}{q(a_0)} =
    \frac{1}{\omega_0 \left( \frac{\omega_0}{k \gamma \lambda_0} + 
                \frac{q''(a_0) q(a_0)}{q'(a_0)^2} - 1 \right)} \nonumber
\end{eqnarray}
The last step is obtained by collecting all terms with $a_0'$ on the left,
multiplying with $q'(a_0)/q(a_0)$ and dividing through the term in brackets.  
We define
\begin{equation}
C(x) := \left( \frac{q(x)}{x q'(x)} + \frac{q''(x) q(x)}{q'(x)^2} - 1 \right). \nonumber 
\end{equation}
Using $Q(a_i) = k \gamma \lambda_i / \omega_i$ the  preceding equation becomes 
\begin{eqnarray}
\frac{{a_0'}_{,\omega_1} q'(a_0)}{q(a_0)} = \frac{{a_0'}_{,\omega_2} q'(a_0)}{q(a_0)} =
    \frac{1}{\omega_0 \left( \frac{q(a_0)}{a_0 q'(a_0)} + 
                \frac{q''(a_0) q(a_0)}{q'(a_0)^2} - 1 \right)} =
    \frac{1}{\omega_0 C(a_0)}. \nonumber 
\end{eqnarray}
We use  equation $Q(a_i) = k \gamma \lambda_i / \omega_i$ again to get 
\begin{eqnarray}
k \gamma \frac{{a_0'}_{\omega_1}}{a_0} = k \gamma \frac{{a_0'}_{\omega_2}}{a_0} =
    \frac{1}{\lambda_0 \left( \frac{q(a_0)}{a_0 q'(a_0)} +
                \frac{q''(a_0) q(a_0)}{q'(a_0)^2} - 1 \right)} =
    \frac{1}{\lambda_0 C(a_0)}  \nonumber 
\end{eqnarray}

\paragraph{Derivative of $a_1.$}
As for $a_0$ we get 
\begin{eqnarray}
k \gamma \frac{{a_1'}_{\omega_1}}{a_1} =
    - \frac{1}{\lambda_1 \left( \frac{q(a_1)}{a_1 q'(a_1)} +
                \frac{q''(a_1) q(a_1)}{q'(a_1)^2} - 1 \right)} =
    - \frac{1}{\lambda_1 C(a_1)} \nonumber \\
\frac{{a_1'}_{\omega_1} q'(a_1)}{q(a_1)} =
    - \frac{1}{\omega_1 \left( \frac{q(a_1)}{a_1 q'(a_1)} + 
                \frac{q''(a_1) q(a_1)}{q'(a_1)^2} - 1 \right)} =
    - \frac{1}{\omega_1 C(a_1)}. \nonumber 
\end{eqnarray}
The remaining  $a_i-$derivatives can be calculated in a similar way. For
$\omega_i = \lambda_i = \frac{1}{3}$ (then  $ a_i = s, c_i=1)$
we get 
\begin{equation}
k\gamma\frac{a_i'}{a_i} \quad \textrm{and} \quad
\frac{a_i' q'(a_i)}{q(a_i)} \quad \textrm{is} \quad \frac{3}{C(s)} \mbox{ for } i=0 \mbox{ and }  -\frac{3}{C(s)} \mbox{ for } i=1, 2 \label{MAT1}. 
\end{equation}

\paragraph{Derivatives of $c_i$}
By $R(c_1,c_2) = (k\lambda_1, k\lambda_2)$ we have
\begin{equation}
\frac{c_1 r_{c_1}(1, c_1,c_2)}{r(1, c_1,c_2)} = k \lambda_1 \quad \Longleftrightarrow \quad
\frac{c_1}{k} = \frac{\lambda_1 r(1, c_1,c_2)}{r_{c_1}(1, c_1, c_2)} \nonumber 
\end{equation}
Taking the derivative  wrt. $\lambda_1$ leads to  (omitting the argument  $1$) 
\begin{eqnarray} 
{c_1'}_{\lambda_1} \left( \frac{1}{k} - \lambda_1 +
        \lambda_1 \frac{r(c_1,c_2) r_{c_1,c_1}(c_1,c_2)}{r_{c_1}(c_1,c_2)^2} \right) = \nonumber \\
\quad  = \lambda_1 {c_2'}_{\lambda_1} \left( \frac{r_{c_2}(c_1,c_2)}{r_{c_1}(c_1,c_2)} -
                    \frac{r(c_1,c_2) r_{c_1,c_2}(c_1,c_2)}{r_{c_1}(c_1,c_2)^2} \right)
                    + \frac{r(c_1,c_2)}{r_{c_1}(c_1, c_2)} \label{eqn:diffC1L1}
\end{eqnarray}
Also by $R(c_1,c_2) = (k\lambda_1, k\lambda_2)$ we have
\begin{equation}
\frac{c_2 r_{c_2}(c_1,c_2)}{r(c_1,c_2)} = k \lambda_2 \quad \Longleftrightarrow \quad
\frac{c_2}{k \lambda_2} = \frac{r(c_1,c_2)}{r_{c_1}} \nonumber 
\end{equation}
Taking the derivative wrt.  $\lambda_1$ again leads to
\begin{eqnarray}
{c_2'}_{\lambda_1} \left( \frac{1}{k \lambda_2} - 1 +
        \frac{r(c_1,c_2) r_{c_2,c_2}(c_1,c_2)}{r_{c_2}(c_1,c_2)^2} \right) = \nonumber \\
\quad  = \lambda_1 {c_1'}_{\lambda_1} \left( \frac{r_{c_1}(c_1,c_2)}{r_{c_2}(c_1,c_2)} -
                    \frac{r(c_1,c_2) r_{c_2,c_1}(c_1,c_2)}{r_{c_2}(c_1,c_2)^2} \right)
                      \label{eqn:diffC2L1}
\end{eqnarray}
Again we consider the  point $\lambda_1 = \lambda_2 = \frac{1}{3}$ then $ c_1 = c_2 = 1$  and
equations
(\ref{eqn:diffC1L1}) and (\ref{eqn:diffC2L1}) yield 
\begin{equation}
2 {c_1'}_{\lambda_1} = {c_2'}_{\lambda_1} + 9 \quad \textrm{and} \quad
2 {c_2'}_{\lambda_1} = {c_1'}_{\lambda_1}. \nonumber
\end{equation}
Therefore we have $\frac{{c_1'}_{\lambda_1}}{c_1} = 6$ and
$\frac{{c_2'}_{\lambda_1}}{c_2} = 3$.
Analogously for the derivatives wrt. $\lambda_2$ we get
$\frac{{c_1'}_{\lambda_2}}{c_1} = 3$ and
$\frac{{c_2'}_{\lambda_2}}{c_2} = 6$.

Putting the derivatives together we get from (\ref{eqn:diffLiLi}) - (\ref{eqn:diffLiWj}) the following Hessian-Matrix of $\ln \Psi(\bar\omega,\bar\lambda)$
at the point  $\omega_i=\lambda_i=1/3$ ,  abbreviating $D=3/C(s), $ 
\begin{equation}
H = \left(
\begin{array}{cccc}
-2(\frac{1}{3}+D)   & -(\frac{1}{3}+D)  & 2D    & D     \\
-(\frac{1}{3}+D)    & -2(\frac{1}{3}+D) & D     & 2D    \\
2D                  & D                 & -2(\frac{8}{3}k\gamma+D)  & -(\frac{8}{3}k\gamma+D)   \\
D                   & 2D                & -(\frac{8}{3}k\gamma+D)   & -2(\frac{8}{3}k\gamma+D)
\end{array}
\right) \nonumber 
\end{equation}
$H $ is negative definite iff $-H$ is positive definite.

\begin{lem}[Jacobi]
A matrix $A = A^T = (a_{ij}) \in \mathbb{R}^{n \times n}$ is positive definite iff
the determinants of ist $n$ main-sub-matrices $S_i$ are positive.
\begin{equation*}
S_1 = a_{11}, \quad
S_2 = \left( \begin{array}{cc} a_{11} & a_{12} \\ a_{21} & a_{22}\end{array} \right),
..., \quad
S_k = \left( \begin{array}{ccc} a_{11} & \hdots & a_{1k} \\
                                \vdots && \vdots \\
                                a_{k1} & \hdots & a_{kk}
            \end{array} \right),
..., \quad
S_n = A
\end{equation*}
\label{lem:jacobi}
\end{lem}

\noindent By Lemma \ref{lem:jacobi} $-H$ is positive definite,  as $D>0$ as $C(x)>0$ for $x>0, $ and
\begin{eqnarray*}
\det S_1 &=& 2\left( \frac{1}{3} + D \right) > 0 \\
\det S_2 &=& 3\left( \frac{1}{3} + 2D + 3D^2 \right) > 0 \\
\det S_3 &=& \frac{16}{9} k \gamma + \frac{2}{3} D +
            \frac{32}{3} k \gamma D + 2D^2 + 16 k \gamma D^2 > 0 \\
\det S_4 = \det (-H)&=& \frac{64}{9} k^2 \gamma^2 +
                D^2 + 64 k^2 \gamma^2 D^2 + 16 k \gamma D^2 + \\
&& \quad        + \frac{128}{3} k^2 \gamma^2 D + \frac{16}{3} k \gamma D > 0.
\end{eqnarray*}

\noindent
{\bf Theorem \ref{LAPLA} \quad (repeated) } 
Let $U={\cal U}_\varepsilon(1/3, 1/3).$ There is an $\varepsilon>0$ such that \\
\hspace*{4cm} $\sum_{\bar{\lambda}  , \bar{\omega} \in U} N(\bar{w}, \bar{l})/N_0 \,< \, C \cdot 3^{(1-\gamma)n}.$ \\  
\begin{proof} 
For $ \bar{\lambda}  , \bar{\omega} \in U $ and $a_i$ given by $Q(a_i)= \lambda_ik\gamma/\omega_i$ and 
$c_0=1$ and $R(c_1, c_2)= ( \lambda_1k,\, \lambda_2k)$ we have 
$ \frac{N(\bar{w}, \bar{l})}{N_0} \, \, \le  \, 
\,O \left(\frac{1}{n^2}\right) \Psi( \bar{\omega}\, ,\, \bar{\lambda}\, ,\, \bar{a} \, ,\, \bar{c}  \, )^n $
by Corollary  \ref{PSIINU}.  
Let $\bar{x}=(x_1, \dots, x_4)$ and $h(\bar{x})\,= \,  \ln \Psi( \bar{\omega}\, ,\, \bar{\lambda}\, ,\, \bar{a} \, ,\, \bar{c}  \, )$ 
with $\omega_1=x_1, \omega_2=x_2, \lambda_1=x_3,  \lambda_2=x_4 $ and $a_i, c_i$ as before  for  $\bar{\omega}, \bar{\lambda} \in U.$ 
Let $\overline{1/3}= (1/3, 1/3, 1/3, 1/3)$ then $h(\overline{1/3})=\ln 3^{1-\gamma},\, h_{x_i}(\overline{1/3})\,= \, 0 \, $
and $-\mbox{Hess}(h)(\overline{1/3})$ $,  \mbox{Hess}(h)$ the Hessian matrix of $h, $ is positive definite (
proved above, note  $\mbox{Hess}(h)(\overline{1/3})=H.)$ We abbreviate 
$h_{i, j}\,= \, h_{x_i, x_j}( \overline{1/3})$  and by Taylor's Theorem we have for $\sum_i \,x_i^2 \,\rightarrow \,\,0$ 
\begin{eqnarray} 
h(\overline{1/3}+ \bar{x}) \,= \, h(\overline{1/3}) \,- \, \frac{1}{2} \, \, \sum_i\, \sum _j -h_{i,j} x_i x_j\, + \, o(\sum_i x_i^2) \nonumber \\
\le \,h(\overline{1/3}) \,- \, \frac{1}{2} \, \, \left( \sum_i \,- (h_{i,i} + \delta)  x_i^2\, + \, \sum_i\, 
\sum _{j \neq i} -h_{i,j} x_i x_j \right)\label{TAYLOR1}
\end{eqnarray}
with $\delta $ arbitrarily small for  $\sum_i \,x_i^2 $ small enough. We 
pick $\delta$ such 
that 
$-(\mbox{Hess}(h)(\overline{1/3})+ \delta I )$
is still  positive definite.  

We consider  (\ref{TAYLOR1}) with 
$ x_1=w_1/n-1/3,\,\, \, x_2=w_2/n-1/3$ 
and $x_3=l_1/(k\gamma n)-1/3\, \,\, x_4= l_2/(k \gamma n)-1/3.$
Then 
\begin{eqnarray}
\sum_{\bar{\omega}, \bar{\lambda} \in U} \,   \Psi ( \bar{\omega}\, ,\, \bar{\lambda}\, ,\, \bar{a} \, ,\, \bar{c}  \, )^n\,
= \,\sum_{\bar{\omega}, \bar{\lambda} \in U} \exp ( h(x_1, x_2, x_3, x_4) n) \nonumber \\
\, \le \, 3^{(1-\gamma)n} \cdot \sum_{\bar{\omega}, \bar{\lambda} \in U} \exp
\left[\,- \, \frac{1}{2}\left( \, \sum_i \, - (h_{i,i} + \delta)  x_i^2\, + \, \sum_i\, \sum _{j \neq i}-h_{i,j} x_i x_j \right)n \right] \label{TAYLOR2}
\end{eqnarray} 
Note that $\omega_i=w_i/n,\, \lambda_i=l_i/(k\gamma n), \, w_i, l_i $ integer. 
We distribute the factor $n$ into the $x_i$ multiplying 
each $x_i$ with $\sqrt{n}:$ 
$x_1 \sqrt{n}\,= \, w_1/\sqrt{n}-\sqrt{n}/3$ 
and $x_3\sqrt{n}\,= \, l_1/(k \gamma\sqrt{n})-\sqrt{n}/3.$ 
As $w_i, l_i$ are integers, the  sum  in (\ref{TAYLOR2}) multiplied  with $1/(\sqrt{n}^4(k\gamma)^2)$ 
is a Riemannian sum of the integral 
$\int \int \int \int \exp[-(1/2) ( - (h_{i,i} + \delta)  x_i^2\, + \, \sum_i\, \sum _{j \neq i}-h_{i,j} x_i x_j)]dx_1 dx_2 dx_3 dx_4$ 
with  bounds $-\infty, \, \, \infty$  for each $x_i.$  Following \cite{dBRU}, page 71, the integral evaluates to 
$(2\pi)^2/\sqrt{D}$ where $D>0$ is  the determinant of $(-h_{i, j})+ \delta I.$ Thus for $n$ large the 
sum in (\ref{TAYLOR2})is  $(2\pi)^2/\sqrt{D}(1+o(1)) \sqrt{n}^4(k\gamma)^2\,= \, O(n^2).$ 
The claim follows.
\end{proof} 

\quad \quad

\section{Remaining proofs}  \label{REPRO}

\subsection{Local limit consideration} \label{LOLICO}

{\bf Lemma  \ref{MLOC} (repeated)}
Let $Cn \ge m \ge (2+\varepsilon)n, \, \, C, \varepsilon>0$ constants. Then 
\begin{eqnarray}
M(m, n) \,= \, \Theta(1)\cdot \left( \frac{ m}{a e}\right)^m\cdot q(a)^n \mbox{ with } a \mbox{ defined by }  Q(a)=\frac{m}{n} \nonumber
\end{eqnarray}

\begin{proof}  
As $Q(x)$  is increasing the  assumptions for $m/n$ imply  that   $a$ is bounded away from $0$ and $\infty.$ 
Let $ X \,= \, X(x) $ be a random variable with Prob$[ X=j ]\,= \, (x^j/j!)/ q(x), \,$ for  $  j \ge 2, $ and
let  $X_1, \dots ,X_n $ be independent copies of $X.$ Then  
\begin{eqnarray}
\sum_{ l_i \ge 2} {m \choose l_1, \dots , l_n} \,= \, \mbox{Prob}[ X_1 + \dots + X_n\,= \, m ] \cdot \frac{q(x)^n}{x^m}\cdot m!.  \nonumber 
\end{eqnarray}
We have   E$[X]= xq'(x)/q(x)=Q(x).$ We pick $x=a$ then  E$[X]=m/n, $  E$[X_1+ \dots + X_n]\,= \, m.$ 
The bounds on $a$ imply that 
$C > \mbox{VAR}[X] > \epsilon >0$ (constants $\epsilon, C $ not the same as above.) 
Therefore   the Local  Limit Theorem for lattice type random variables , cf. \cite{DUR} , Theorem 5. 2, page 112,
implies that  Prob$[ X_1 + \dots + X_n\,= \, m  ]=  \Theta\left(\frac{1}{\sqrt{m }}\right) .$ 
Applying Stirling's formula in the form  $ m! = \Theta (\sqrt{m}) \left(\frac{m}{e}\right)^m $
yields the claim. 
\end{proof}

We come to Lemma \ref{KLOC}. First we show that $R(c_1, c_2)\, = \, (R_1(c_1, c_2), R_2(c_1, c_2))\,=\,(k \lambda_1, k\lambda_2)$ 
with $R_i(x_1, x_2)= \frac{x_i r_{x_i}(1, x_1, x_2)}{r(1, x_2, x_2)}$ defines $c_i=c_i(\lambda_1, \lambda_2)$ 
and that $c_i$  is  differentiable with respect to $\lambda_i$ for $(\lambda_1, \lambda_2)\in {\cal U}_\varepsilon(1/3, 1/3).$ 
By the theory of implicit function of several variables we need to show that the
Jacobian Determinant of $R(x_1, x_2)$ is $\neq 0$ for $x_1=x_2=1.$ The Jacobian Matrix  of $R(x_1, x_2)$ is , omitting the arguments $x_i,$ recalling that  $r=r(1, x_1, x_2)$ is our polynomial, 
\begin{equation}
J = \frac{1}{r^2}\left(
\begin{array}{cc}
  (r_{x_1}+x_1r_{x_1, x_1})r-x_1 r_{x_1}^2 \quad & \quad x_1r_{x_1, x_2}r-x_1 r_{x_1}r_{x_2} \\
   & \\
x_2r_{x_1, x_2}r-x_2 r_{x_1}r_{x_2} \quad & \quad  (r_{x_2}+x_2r_{x_2, x_2})r-x_2 r_{x_2}^2
\end{array}
\right). \nonumber 
\end{equation}

For $x_1=x_2=1$ we get the following values:
$r=r(1, 1, 1)= 3^{k-1}, \, r_{x_1}= r_{x_2}=k3^{k-2}, \, r_{x_1, x_1}=r_{x_2, x_2}= r_{x_1, x_2}  = k(k-1)3^{k-3}.$
>From this we get that the determinant of $J$ for $x_1=x_2=1$ is  $.... \neq 0.$ \\

\noindent
{\bf Lemma  \ref{KLOC} (repeated)}
There is an   $\varepsilon>0$ such that for $(\lambda_1, \lambda_2) \in  {\cal U}_\varepsilon(1/3, 1/3)$  
\begin{eqnarray*}
K(\bar{l}) \,= \, O\left( \frac{1}{n}\right) \cdot \frac{r(1, c_1, c_2)}{c_1^{l_1}c_2^{l_2}} \mbox{ with } R(c_1, c_2)=(k\lambda_1, k\lambda_2) \mbox{ defining } c_1, c_2.  \nonumber 
\end{eqnarray*}
\begin{proof}
The previous consideration shows that $(c_1, c_2)$ is close to $(1, 1)$ and well-defined. Let
$(X, Y)=(X(x_1, x_2), Y(x_1, x_2))$ be the random vector with 
\begin{equation*}
\mbox{Prob}[(X, Y)\,=\,(k_1, k_2)]\,= \, \frac{ { k \choose k-k_1-k_2, \,k_1, \,k_2} x_1^{k_1}x_2^{k_2}}{r(1, x_1, x_2)} \mbox{ if } k_1=k_2 \mod 3
\end{equation*}
and $0$ otherwise. Then $E(X, Y)\,= \,( R_1(x_1, x_2),\, R_2(x_1, x_2)).$ We consider $m$ independent  copies 
$(X_i\, , Y_i)$ of $(X,\, Y)$ with $(x_1,\,x_2)= (c_1,\, c_2).$ Then $E\left[\sum_i\, (X_i, Y_i)\right]= (k \lambda_1 m, \,k\lambda_2 m)=(l_1, l_2).$ 
Let $DCo$ be the determinant of the covariance matrix of $(X, Y).$ We show below that for $(c_1,c_2)$ close to $(1, 1)$ we have  that
$DCo >0$  for constants.  The Local Limit Theorem  for lattice random vectors \cite{INDER}, Theorem 22.1, Corollary 22.2 with $k=2$
shows that Prob$\left[\sum_i\, (X_i, Y_i) \,= \, (k\lambda_1 m, \,k\lambda_2 m)\right]\,= \, \Theta(1/m).$  This implies the claim.

The covariance matrix of $(X, Y)$ is defined as
\begin{equation}
Co \,= \, \left(
\begin{array}{cc}
  EX^2-(EX)^2\quad & \quad  E[XY]-\, E[X]E[Y]  \\
   & \\
   E[XY]-\, E[X]E[Y]              \quad & \quad   EY^2-(EY)^2
\end{array}
\right). \nonumber 
\end{equation}
For $(X, Y)\,=\,(X(x_1, x_2), Y(x_1, x_2))$ we get 
\begin{eqnarray*}
EX^2\,= \, \frac{x_1 (x_1 r_{x_1, x_1}(1, x_1, x_2)\,+ \, r_{x_1}(1, x_1, x_2)}{r(1, x_1,x_2))}, \, \\
 EY^2\,= \, \frac{x_2 (x_2r_{x_2, x_2}(1, x_1, x_2)\,+ \, r_{x_2}(1, x_1, x_2)}{r(1, x_1,x_2))}, \, \nonumber \\
E[XY]\,= \, \frac{x_1x_2r_{x_1x_2}(1, x_1, x_2)}{r(1, x_1, x_2)}. 
\end{eqnarray*}
This leads to  a matrix similar to the Jacobian Matrix above: For  $x_1=x_2=1$ its determinant is positive. 
\end{proof}

\subsection{The sharp threshold} 
To prove the sharp threshold we apply a general theorem.
Let $A \subseteq \{0,\, 1\}^N$ and let $a_m$ be the number of elements of 
$A$ with exactly $m$ $1'$s. We let $\mu_p(A) \,= \, \sum_{m=0}^N \, a_m \cdot p^m \cdot (1-p)^{N-m}$
be the probability of $A,$ note $a_m \le {N \choose m}.$  If $A$ is a non-trivial,
monotone set we have that  $\mu_p(A)$ is a strictly increasing, continuous, differentiable  function in $0 \le p \le 1.$
In this case for $0 \le \tau \le 1$ we have that $p_\tau$ is well defined by $ \mu_{p_\tau}(A)=\tau.$ 
Not let $A=(A_n)_{n\ge 1}$ and let be $A_n$  be monotone. We say that $A$ has a coarse threshold iff 
there exist constants $0 < \rho < \tau < 1$ such that $(p_{\tau}- p_{\rho})/p_{\rho} \, \ge \varepsilon$
for a constant $\varepsilon$ (and infinitely many $n.$) 
We can assume that $p_{\tau}=O(p_\rho)$ otherwise the threshold is
clearly coarse. Moreover, we assume  that $p_{1-o(1)}\,= \,o(1).$

\begin{theor}[ Bourgain, \cite{FRI} , Theorem 2.2 ]
There exist functions $\delta=\delta(C, \tau)>0$ and  $K=K(C, \tau)$ 
such that the following holds: Let   $A=A_n$ with $A \subseteq \{0,\, 1\}^N$ be 
a monotone set with $\tau \le \mu_p(A) \le 1-\tau$ for constant $1/2> \tau > 0$ 
and assume that $p \cdot \frac{d\mu_p(A)}{dp} < C.$ Then at least one of the following two possibilities
holds: \\
1.\begin{eqnarray*} \mbox{Prob}_p[a \in A\, \, ; \, \, \exists b \in A \, \, , |b| \le K \, \, , b \subseteq a] \,> \, \delta
\end{eqnarray*}
2. There exists  $b\, \in \{0, 1\}^N   
, \, b \notin \, A\, , |b| \le K$  such that the conditional probability 
\begin{eqnarray*}
\mbox{Prob}_p[a \in A \,| \, b \subseteq a ]\,> \, \mbox{Prob}_p[A]\,+ \, \delta .
\end{eqnarray*}
\end{theor}

\begin{cor} \label{Thresh}
$A=(A_n)$ has a sharp threshold if 
$ p_{1-o(1)}=O(p_\tau)$ for all $\tau>0,$   and 
for each $1/2>\tau >0, \, \delta >0, \varepsilon >0 , K, $ $p_\tau < p < p_{1-\tau}$ and all sufficiently large $n$ the following two statements hold: \\
1. \begin{eqnarray*} \mbox{Prob}_p[a \in A\, \, ; \, \, \exists b \in A \, \, , |b| \le K \, \, , b \subseteq a]\, \,<\, \delta. \end{eqnarray*}
2. If  $b\, \in \{0, 1\}^N   
, \, b \notin \, A \, , |b| \le K$ with the conditional probability $\mbox{Prob}_p[a \in A \,| \, b \subseteq a ]\,> \, \mbox{Prob}_p[A]\,+ \, \delta $
then $\mbox{Prob}_{p(1+\varepsilon)}[A] \,> \, 1-\tau$
\end{cor}
\begin{proof}
Assume, that $A$ has a coarse threshold. Let $1>\alpha> \beta>0$  be such that 
$(p_{\alpha}- p_{\beta})/p_{\beta} \, \ge \varepsilon.$
We abbreviate  $q=(p_\alpha+ p_\beta)/2.$ By strict monotonicity of $\mu_p(A)$ we have 
$\mu_q(A)\,= \, \gamma $ for a $\alpha> \gamma > \beta.$ We have that 
$\frac{\gamma-\beta}{q-p_\beta} \,= \, \frac{d\mu_p(A)}{dp}|p=p^*$ for a $p _\beta< p^*<q$ 
(by the Mean Value Theorem.) We have that $(q-p_\beta)/p^* \, \ge \varepsilon'$ as $p^*=O(p_\beta).$ 
Therefore  $\frac{\gamma -\beta}{q-p_\beta} \cdot p^*\,= \, \left(\frac{d\mu_p(A)}{dp}|p=p^*\right) \cdot p^* \le C$ 
for a constant $C.$ The preceding theorem applies to $p^*.$ 
Our assumption implies that the first item of the theorem does not hold.

Therefore the second item of the preceding  theorem must hold for $p=p^*.$
We have that $p^*+\frac{p_\alpha-p_\beta}{ 2}\, < p_\alpha.$ 
Therefore   $p^*\left( 1+ \frac{p_\alpha-p_\beta}{p^* \cdot  2}\right) < p_\alpha.$ 
Moreover $\frac{p_\alpha-p_\beta}{p^* \cdot  2}\,> \, \varepsilon''$ as $p^* =O(p_\beta).$ 
Our second assumption shows that the 
preceding statement cannot hold. Therefore the second item of the preceding theorem does not hold, too.  
Therefore $A$ 
cannot have a coarse threshold.  
\end{proof} 

Let $F(n, p)$ be the random formula 
of equations $y_1+ \dots + y_k= a \mod 3,$  $0 \le a \le 2$ over $n$ variables 
where each equation is picked  with probability $c/n^{k-1}$ independently.  

\begin{lem} \label{THRESH}
Unsatisfiability of  $ F(n, p) $ has a sharp threshold. 
\end{lem} 
\begin{proof}
We apply Corollary \ref{Thresh}. Let $p=c/n^{k-1}.$ Observe that $F(n, p)$ is unsatisfiable whp. for $c>1$ by expectation calculation. 
Concerning the first item of the corollary  we show 
that $F(n, p)$ does not contain a subformula over a bounded number of variables such that  each 
variable occurs at least twice. The expected number of such subformulas over $1\le l\le B,$ $B$ constant  variables 
is bounded above by $ {n \choose l }\cdot\left(c/n^{k-1}\right)^{2l/k}\le O(1) \cdot n^{(2/k-1) l}.$ As $k\ge 3$  
and $l\ge 1$ the geometric series  shows that the expectation of the number of such subformulas with 
$\le B$ variables is $o(1).$ 
As each unsatisfiable formula contains a subformula where each variable occurs at least twice
we have no unsatisfiable subformula of bounded size whp. 
The first item of the corollary holds.  

Concerning the second item, let $B$ be a fixed satisfiable formula and let $p < 1/n^{k-1}.$  We assume that 
$\mbox{Prob}[\mbox{UNSAT} (B \cup F(n, p)) ] >   \mbox{Prob}[\mbox{UNSAT}(F(n, p)) ]+ \delta. $
UNSAT$(F)$ is the event that $F$ is unsatisfiable. 
With high probability  $F(n, p)$ contains only equations with $1$ or none variables from  $B$ 
(as $p< 1/n^{k-1}$ and  the number of variables of $B$ is constant. )

Consider a fixed satisfiable formula  $F$ over the variables not in $B$
We pick each equation with 
exactly one variable in $B$ with probability $p=c/n^{k-1}$ independently.
We assume that the resulting random formula is unsatisfiable with
probability $\delta>0.$ We show that this implies that 
the random instance obtained from $F$ by  adding {\it  each } equation with probability 
$ \varepsilon / n^{k-1},$ independently,  $\varepsilon>0$ constant. 
is unsatisfiable with high probability. This directly implies that 
the second item of Corollary \ref{Thresh} holds. 

Consider a fixed variable $x$  of $F.$ We  throw in the 
equations containing  $x$ with   $ \varepsilon / n^{k-1},$ 
We show below that 
the resulting random formula is unsatisfiable with 
probability $\delta'>0,$ $\delta'$  constant. 
Throwing {\it each} equation with probability 
$\varepsilon/n^{k-1},$  the  expected number of variables $x$ such that
the equations containing $x$ lead to unsatisfiability  of $F$ 
is  $\delta' n.$ 
For $x \neq x'$ the equations with $x$ or $ x'$ are
nearly independent. Tschebycheff's inequality shows 
that we  even  have a linear number of variables $x$ whose equations 
yield unsatisfiability  whp. 

We show the statement above concerning the fixed variable $x.$ 
When throwing in the equations with one variable in $B$ with 
$p=c/n^{k-1}$ we get with probability $\delta$ a set
$U$ such that $F\cup B \cup U$ is unsatisfiable.  
With probability slightly lower, but still constant $>0$ we can
assume that $U$ is of bounded size. 
Now consider a satisfying assignment $a$ of $B.$ We replace the 
variable from $B$ in each equation by its value under $a$ 
and get a set of equations  with $k-1$ variables each.
When we add these equations to $F$ the resulting formula is
unsatisfiable. 

Now consider our  variable $x$ from 
$F$ and throw in   each equation containing 
$x$ with probability $\varepsilon/n^{k-1}.$
With constant probability $>0$ we get 
the a set $U'$  obtained from 
a set $U$ as above by replacing the variable from 
$B$ by $x.$ With the same probability we 
get $U_0$ instead of $U'$ where $U_0$ is obtained as follows:
Let $E$ be an equation of $U$ such that the  variable from $B$ has the value 
$j$ in the satisfying assignment $a$ from 
$B.$ The variable from $B$ is replaced with $x$ in 
$E$ and we subtract $j$ from the right hand side. 
The resulting formula is unsatisfiable for all
assignments which have  $x=0.$
$U_1$ is defined by adding $1-j$ to the right hand-side.
The resulting formula is unsatisfiable for $x=1.$
$U_2$ is defined by adding $2-j$ and the resulting 
formula is unsatisfiable for $x=2.$ 
With constant probability $>0$ we get one such set 
$U_j.$

To get unsatisfiability for all $3$ values of $x$ 
we observe that with probability roughly 
$\delta^3$ we get three sets $U, V, W$ 
with one variable in $B$ which are disjoint and
each of them causes unsatisfiability. 
This implies that with constant probability $>0$ we
get three sets $U_0, V_1, W_2$ of equations with $x.$
The resulting formula is unsatisfiable for any 
value of $x.$
\end{proof}

\newpage

\noindent
{\bf \Large II.  Uniquely extendible constraints }

\setcounter{section}{0}

\section {Outline}

A uniquely extendible constraint $C$  on a given domain  $D$ is 
a function from $D^k$ to true, false with the following restriction:
For any   argument list with a gap at an arbitrary position, like  
$(d_1, \dots d_{i-1}, -,  d_{i+1}, \dots , d_k)$ 
there is a unique $d \in D$ such that \\ $C(d_1,\dots d , \dots , d_k)$
evaluates to true. Note that 
$C(d_1 ,\dots , d, \dots, d_k) = $ true implies that $C(d_1, \dots ,d', \dots ,d_k)=$ false
for  $d \neq d'.$ 
The random constraint is a uniform random member from the
set of all uniquely extendible constraints over $D.$  
Let $\Gamma$ be the set of all 
such constraints.  Typical examples of such constraints
are linear  equations with $k$ variables,  modulo  $|D|.$  
A threshold result analogous to Lemma \ref{THRESH} can be proved by  similar arguments
based on symmetry properties of uniquely extendible constraints.

Given a set of $n$ variables a clause is an ordered $k$-tuple of variables
equipped with a uniquely extendible constraint. 
The number of all formulas with $m$ clauses is $M(km, n) \cdot |\Gamma|^m, $  we denote $N_0=M(km, n)$
(notation cf. (\ref{DEFN0}).) A random formula is a uniform random element of the set of all 
formulas. The random variable $X$ gives the number of solutions
of a formula and E$[X] = (1/d)^{(1-\gamma)n}, m = \gamma n.$ 
This follows from symmetry considerations.
For two assignments $a, b$ we study E$[X_aX_b]$
where $X_a$ is $=1$ iff the formula is true under $a.$ 
It turns out that E$[X_a X_b]$ depends only on the number of variables which have different
values under $a, b.$  Let 
DIFF$(a, b)=$ the set of variables with different values under $a$ and $b.$

Given a $k$-tuple $a$ of values from  $D$ and another $k$-tuple $b$ differing 
from $a$ in exactly $i,\,\  0 \le \, i\, \le k,$ slots, 
we let $p_i$ be the probability that the random constraint is true under 
$b$ conditional on the event that it is true under $a.$  
The following very simple 
generating polynomial for the ${k \choose i}  \cdot p_i$ is the  observation making our proof possible.   
\begin{lem} 
(a) (From \cite{COMO}) $p_0=1, \, \, p_{i+1}\,= \,  \frac{1}{d-1} \left(1 - p_i \right).$ \\
(b)
\begin{eqnarray}
\mbox{ Let } p(z) \,\,= \,\,\frac{1}{d}\left( (1+z)^k \,+ \, (d-1) \left( 1-\frac{z}{d-1} \right)^k \right) \, 
\mbox{ then } 
p(z) \,\,= \,\,  \sum_i {k \choose i} p_i \cdot z^i \nonumber 
\end{eqnarray}
\end{lem}
\begin{proof}
(b) We need to show that $p_i\,= \,\frac{1}{d}\left( 1+ (-1)^i \left( \frac{1}{d-1}\right)^{i-1}\right).$
This holds for $i=0, i=1.$ 
For $\,i >1 \,$   we get by induction: 
\begin{eqnarray}
p_i\,= \frac{1}{d-1}(1-p_{i-1}) \, \,\,= \,\, \,\frac{1}{d-1}\left( 1\,- \, \frac{1}{d}\left( 1+(-1)^{i-1}\left(\frac{1}{d-1}\right)^{i-2}\right)\right) \,=\, \nonumber \\
\,= \, \frac{1}{d-1} \, \,\, - \,\, \,  \frac{1}{d(d-1)}-  \frac{1}{d}(-1)^{i-1}\left(\frac{1}{d-1}\right)^{i-1}\,\,= \, \,
\frac{1}{d}\left( 1+ (-1)^i\left(\frac{1}{d-1}\right)^{i-1}\right). \nonumber 
\end{eqnarray}
\end{proof}

\begin{eqnarray} 
\mbox{ We let } C_j = \frac{|\Gamma|}{d} \cdot {k \choose j}  \cdot p_j  \mbox{ for } 0 \le j \le k   \,  \, ,  \, \, \, 
K(l)\,= \, \sum_{j_1+ \dots + j_m=l}\,C_{j_1}\cdots C_{j_m} . \nonumber \\
\mbox{ Then } \hat{N}(w, l) \,= \, M(l, w)M(km-l, n-w) K(l)  \nonumber 
\end{eqnarray}
is the number of formulas $F$ true under two assignments $a, b$ with $|\mbox{DIFF}(a, b)|=w$
and the variables with different values occupy exactly $l$ slots of $F.$
The factors ${k \choose j}$ of $C_{j}$  count  how to distribute the $l$ slots. 
The factor $ M(l, w)M(km-l, n-w)$ counts how to place the variables into these slots.
The factors  $\frac{|\Gamma|}{d} \cdot p_{j}$ count the number of constraints 
such that the formula becomes true under $a, b.$ Given an assignment $a$ the number of assignment
formula pairs $(b, F)$ with  $|\mbox{DIFF}(a, b)|=w,$  $F$ is true under $a, b,$ and the variables  from 
$\mbox{DIFF}(a, b)$ occupy exactly $l$ slots is 
\begin{eqnarray}
N(w, l) = {n \choose w}(d-1)^w\cdot \hat{N}(w, l).\, \,  \mbox{ And E}[X^2]= d^n\sum_{w, l}\, N(w, l)\cdot \frac{1}{N_0\cdot |\Gamma|^m}
\nonumber 
\end{eqnarray}
The next theorem is analogous to Theorem \ref{EX2EI}. 
\begin{theor}  \label{UNMAIN}
 $\sum_{w, l} \, N(w, l)/(N_0 |\Gamma|^m) \, \, \le C d^{(1-2\gamma)n}\, \, , \, \, $  $k\ge 8, \, \, m=(1-\gamma)n.$
\end{theor}
We let $\lambda=l/km$ and $\omega=w/n$ with $w, l$ always having the meaning above. The proof of Theorem \ref{UNMAIN} 
follows the pattern of Theorem \ref{EX2EI}. We omit all steps referring to the summation, they are quite analogous. 
The details to bound the summands are however different.  We have 
\begin{eqnarray}
K(l) =  \mbox{Coeff} [z^l, \, p(z)^m ] 
\cdot \left( \frac{|\Gamma|}{d}\right)^m \le \left( \frac{p(c)|\Gamma|}{d}\right)^m\cdot \frac{1}{c^l} \mbox{ for } c>0. \nonumber 
\end{eqnarray}

We define  $\Psi(\omega, \, \lambda, \, x, \, y, \,z\,):= \,$ 
\begin{eqnarray}
 \left(\frac{(d-1)q(x)}{q(s)\omega}\right)^\omega 
\left( \frac{q(y)}{q(s)( 1-\omega )} \right)^{1- \omega} \cdot  \left( \frac{\lambda s}{xz} \right)^{\lambda k \gamma}
\left( \frac{ (1- \lambda ) s }{y}\right)^{ ( 1-\lambda ) k \gamma}\cdot \left (\frac{p(z)}{d} \right)^\gamma. \nonumber  
\end{eqnarray}
We have $\Psi(1-1/d,\, 1-1/d, \, s, \, s,\, d-1 )\,= \, d^{1-2\gamma}\quad, $  $s$ is given by $Q(s)=k\gamma,$ 
cf. discussion around Lemma \ref{MLOC}. As Lemma \ref{LEMBA} we have the next Lemma; the subsequent Theorem is as Theorem \ref{OPT}.
\begin{lem}
$ N(w, l)/(N_0 |\Gamma|^m) \, \le \,  \Psi(\omega, \lambda, a, b, c)\cdot O(n) \mbox{ for }  a, b, c >0.$ 
\end{lem} 
Observe that for  $Q(s)=k\gamma \ge 8$  we have $s\ge 7.$
\begin{theor} \label{UNOPT}
Let $ d=4 $ and $ s \ge 7.$ For any $\lambda>0$ there exist $a, b, c>0$ such that:\\
(1) $\Psi(\omega, \lambda, a, b, c) \, \le \, d^{1-2\gamma}.$ \\
(2) For any $\varepsilon >0$,  $\lambda $ not $\varepsilon-$close to $1-1/d,$   $\Psi(\omega, \lambda, a, b, c)\, \le \, d^{1-2\gamma}- \delta.$
\end{theor} 
Two reals $a, b$ are $\varepsilon-$close iff $|a-b |< \varepsilon.$ 
To treat  $\lambda$ close to $(d-1)/d$ we  consider the function $P(z)= zp'(z)/p(z)$
(cf. discussion after Corollary \ref{SUMOPT}.)
We have $P(d-1)= k(1-1/d)$ and the derivative  $P'(d-1)>0.$  Thus we can define 
$c=c(\lambda)$ for $\lambda$  $\varepsilon-$close to $1-1/d$ by $P(c)=k \lambda.$ 
And $c(\lambda)$ is differentiable. 
As Lemma \ref{KLOC}, Corollary \ref{PSIINU}, and Lemma \ref{LOCMAX}
we get the next $3$ items. To prove Lemma \ref{UNMAX} the
Hessian matrix of  $\Psi(\omega, \lambda, a, b ,c)$ is considered (calculation analogously to \cite{MITZ}.)
\begin{lem}
There is an $\varepsilon>0$ such that for $\omega, \lambda\, \, $   $\, \, \varepsilon-$close to $1-1/d$ 
we have $K(l) = O(1/\sqrt{n}) \cdot \left( p(c)|\Gamma|/d\right)^m\cdot 1/c^l$ with $P(c)=k\lambda.$ 
\end{lem} 

\begin{cor}
There is an $\varepsilon \, > \,0$ such that for $\omega, \lambda $ being  $\varepsilon-$close to $1-1/d$
$N(w, l)/(N_0 |\Gamma|^m) \le O(1/n) \cdot  \Psi(\omega, \, \lambda,\, a,\, b, \,  c)$ with 
$Q(a)=l/w, Q(b)=(km-l)/(n-w), P(c)= \lambda k.$ 
\end{cor}

\begin{lem} \label{UNMAX}
The function  $\Psi(\omega, \lambda, a, b ,c)$ with  $a, b,c$ given by 
$Q(a)=l/w, Q(b)=(km-l)/(n-w), P(c)= \lambda k$ has a local maximum with value 
$d^{1-2\gamma} $ for $\lambda= \omega=1-1/d.$ In this case we have 
$a=b=s$ and $c=d-1.$
\end{lem}
\begin{eqnarray}
\mbox{ We define }\, \,  \mbox{OPT}_1(x, y,  s)\,= \, (d-1) \cdot \frac{ q(sx)}{q(s)}  \,+ \, \frac{q(sy)}{q(s)} , \quad \quad \nonumber \\
\mbox{OPT}_2(x, y, z, s)\,\,=  \,\,
\, \left( \frac{1}{y+xz} \right)^{Q}  , \, y+xz>0\nonumber \\
\mbox{OPT}_3( z, s)\,= \, (1+z)^Q \,+ \, (d-1)\cdot \left| 1\,- \frac{z}{d-1}\right|^Q \, \, , \, Q=Q(s)   \nonumber \\
\mbox{OPT}(x, y, z, s) \,=\, \mbox{OPT}_1(x,y, s) \cdot \mbox{OPT}_2(x, y, z, s) \cdot 
\mbox{OPT}_3(z, s) .  \nonumber 
\end{eqnarray}
As Lemma \ref{LEMOPT} we have the next Lemma. We prove Theorem \ref{UNOPT} based on this lemma. 
We cannot proceed analogously to
the proof of Theorem \ref{OPT}  because the polynomial $p(z)$ is not as symmetric as $r(x_0, x_1, x_2).$ 
The two cases $\lambda$ small (in Section  \ref{SMALL}) and $\lambda$ 
large (in Section \ref{LARGE}) are treated separately. 
\begin{lem} \label{lagrkl}
\begin{eqnarray}
\mbox{ Let } \,  a, b, c\,>0  \, \, \,\mbox{be such that } \frac{\lambda}{1-\lambda}\,= \,\frac{a c} {b} . \, \, 
 \mbox{ Then } \, \, 
\Psi(\omega, \lambda, as, bs, c ) \, \, \le \, \, \frac{1}{d^{2\gamma}} \mbox{OPT}(a, b, c, s).  \,\, \nonumber 
\end{eqnarray}
\end{lem}

\section{Proof of Theorem \ref{UNOPT} for $d=4,\, s\ge 7 , \,  \lambda \le 1-1/d$} \label{SMALL}

We restrict attention to $d=4$  fix $b=1$ and consider   $c, a $  with $0 \le c \le 3$ and $ 0 \le a \le 1.$ 
With these values OPT$(a, b, c, s)$ leads to the following notation used in this Section.
\begin{eqnarray}
\mbox{OPT}_1(a,  s)\,= \, 3 \cdot \frac{q(sa)}{q(s)}\,+ \,1   \, \,, \, \, \mbox{OPT}_2(a, c, s)\,= \, \left( \frac{1}{1+ac} \right)^{Q}  \nonumber \\
\mbox{OPT}_3( c, s)\,= \, (1+c)^Q \,+ \, 3\cdot \left( 1\,- \frac{c}{3}\right)^Q \, \,\nonumber \\
 \mbox{OPT}(a, c, s) \,=\, \mbox{OPT}_1(a,  s) \cdot \mbox{OPT}_2(a, c, s) \cdot 
\mbox{OPT}_3( c, s) .  \nonumber
\end{eqnarray}
The values of OPT$(a, c, s)$ at the corners of the rectangle for $0\le c\le 3,\, \, 0 \le  a \le 1 $ are: 
\begin{eqnarray}
\mbox{OPT}(0, 0, s)\,= \, 4\, \, \,\,\,,\, \mbox{OPT}(0, \, 3 , \, s)\,=\, 4^Q  \nonumber \\
\mbox{OPT}(1, 0, s)\,= \, 4^2 \, \, \,\,,\, \, \mbox{OPT}(1, \, 3, \, s)\,=\, 4  \label{ecken}
\end{eqnarray}

\begin{figure}[th]
\centering
\includegraphics[width=0.49\textwidth]{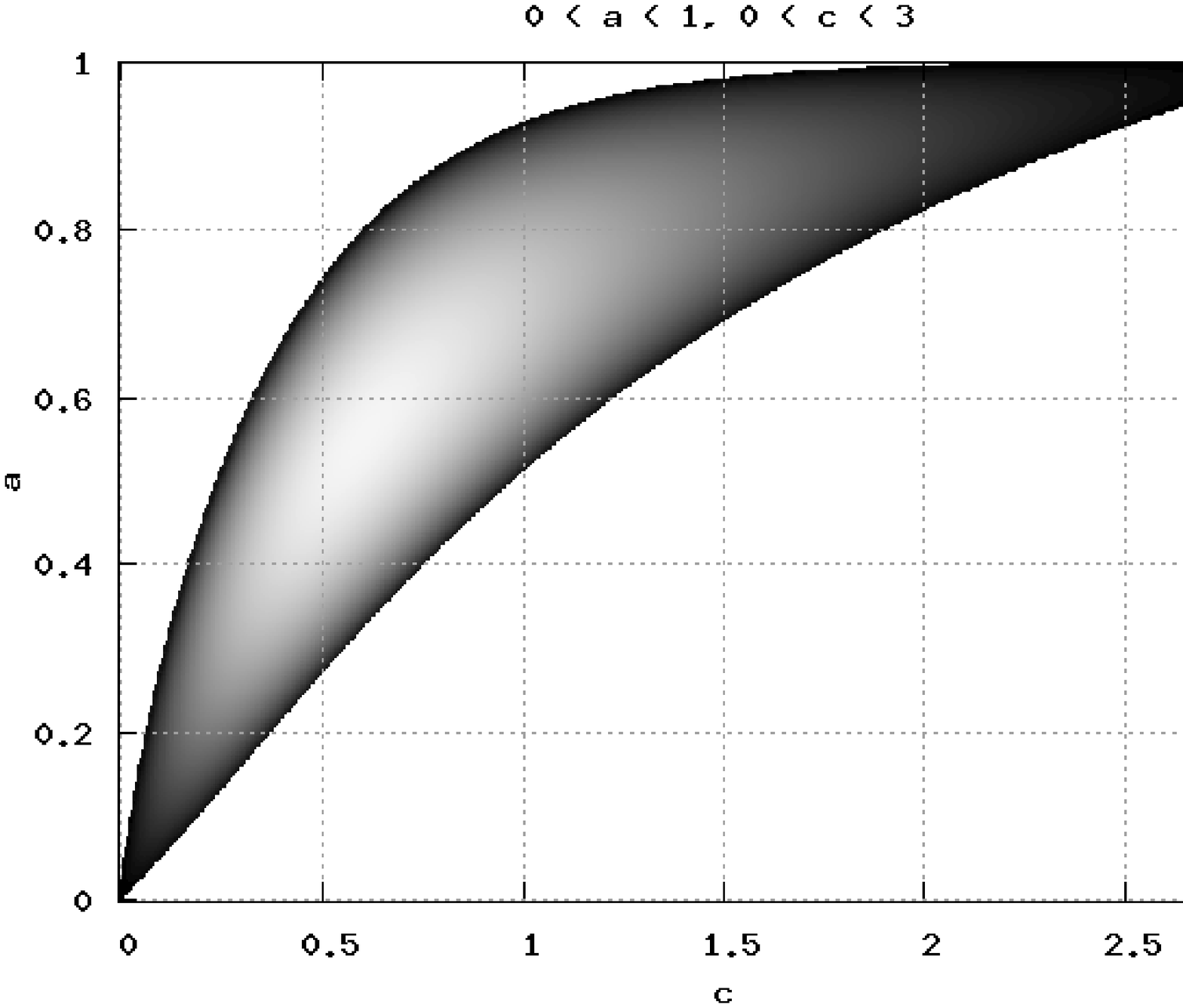} \hfill
\includegraphics[width=0.49\textwidth]{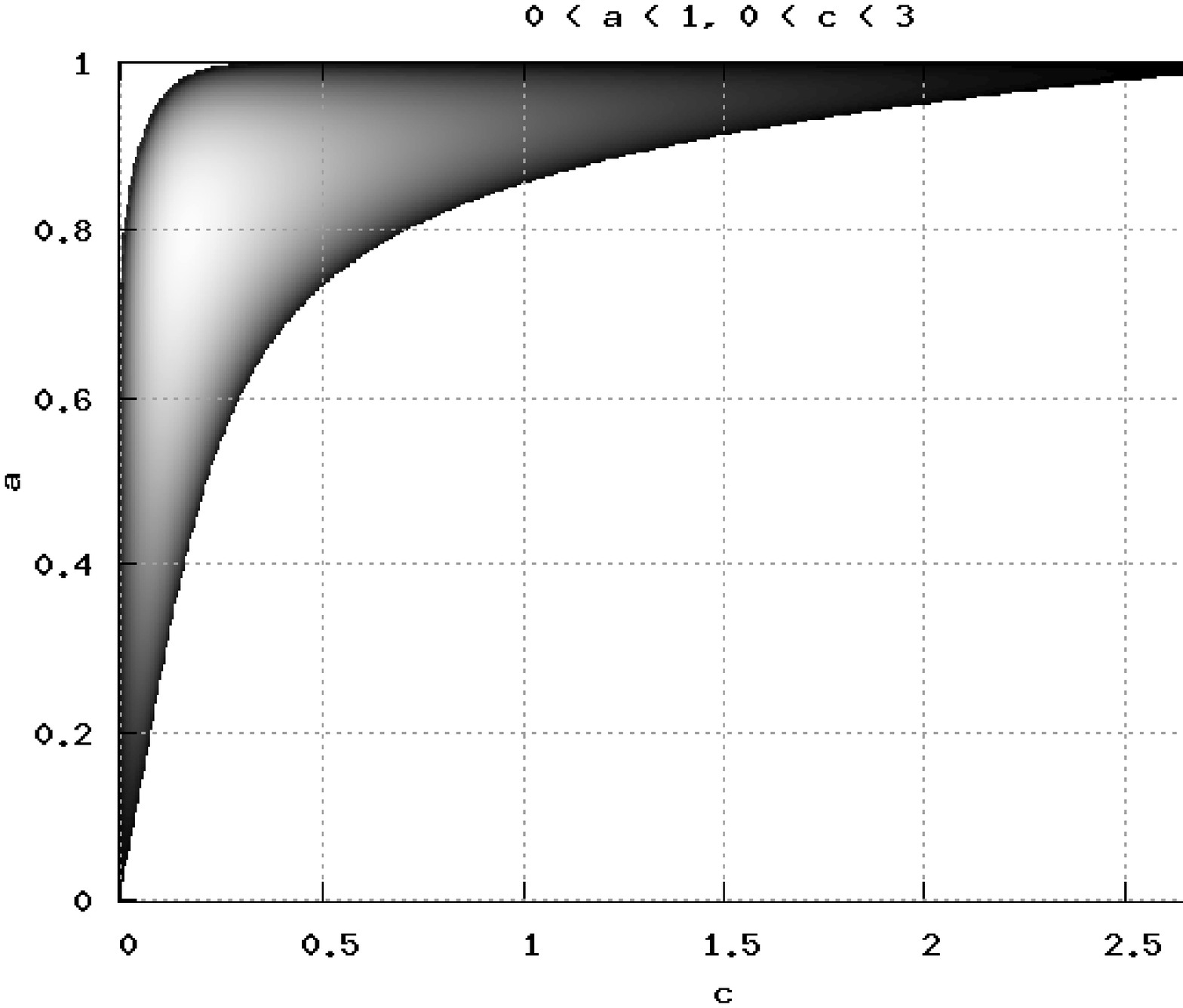}
\caption{OPT$(a,c,s)$ over the rectangle $0\le a \le 1, 0 \le c \le 3$ for $s=3$ and $s=14$.}
\end{figure}

We prove four lemmas. 
Observe that $A(c, s)$ in Lemma \ref{flagekl} is a  flat linear function in $c\ge 0$ from $A(0, s)=1-\frac{7}{10Q}$  to $A(3, s)=1.$  
\begin{lem} \label{flagekl}
\begin{eqnarray}
\mbox{ Let } s \ge 7 \mbox{ and let }\,A(c)\,= \, A(c, s) \,\,=\,\, \frac{7}{Q\cdot 10\cdot 3}\cdot c \,\,+1 \, \, - \frac{7}{10 \cdot Q}. \nonumber \\
\mbox{Then } \, 
 \mbox{OPT}(A(c), \, c, \, s\,) \,\mbox{ is strictly increasing in  } \, \,  1 \le c <3. \nonumber 
\end{eqnarray}
\end{lem}
$A(c, s)$ in the subsequent Lemma is a steep linear function starting at  $A(0, s)=0.$ 
\begin{lem} \label{stgekl}
\begin{eqnarray}
\mbox{ Let } s\ge 6 \mbox{ and let }\, A (c)\,\,=\,\, A(c, s) \,\,=\,\, \frac{Q}{2}\cdot c\, \,. \nonumber \\ 
\mbox{ Then }\, \mbox{OPT}(A(c), \, c, \, s) \, \mbox{  is strictly decreasing for }  0\, <c\, \le \, \frac{1}{Q} \nonumber 
\end{eqnarray}
\end{lem}

\begin{lem} \label{einmi}
(a) For each constant $0 \le a \le 1$   OPT$(a, c, s)$ as a function in $c$ with $ 0 \le c \le 3$ has a unique local minimum. \\ 
(b) For each constant $0 \le c \le 3$  OPT$(a, c, s)$ as a function in $a$ with $ 0 \le a \le 1$ has a unique local minimum.
\end{lem}

\begin{lem} \label{pukl}
 Let $ s\, \ge 6\,$   then OPT$(a, \, \, c, \, \, s\,) \,< \, 4-\delta $   for $   (a, \,  c)\,= \, $
\begin{eqnarray}
 =\,\left(\frac{1}{2}, \, \, \frac{1}{Q} \, \right)\, ,\, \left(\frac{1}{2}, \,  \frac{2}{Q}\,\right),\, 
\,\left(\frac{2}{3}, \, \frac{2}{Q} \,\right) , \, \left(\frac{2}{3}, \, \frac{3}{Q} \,\right), \, 
\left( 1\, \, - \, \,  \frac{7}{15Q}\, , \,  \frac{3}{Q} \right),  \, \, 
 \, \left(1\, - \, \frac{7}{15Q}\, \, ,  \, 1 \,\right)  \nonumber 
\end{eqnarray}
\end{lem}

\noindent
{\it Proof of Theorem \ref{UNOPT} for $\lambda\le 1-1/d.$ (cf. proof of Theorem \ref{OPT} after Lemma \ref{lem4}.)}
We have $ \lambda \le 1-1/d \Longleftrightarrow  \lambda /(1-\lambda)\le d-1.$
Using Lemma \ref{lagrkl}  we  need to show that for each $P\le d-1$ we have a decomposition 
$P=ac$ such that OPT$(a, c, s)\le 4$ of $4- \delta.$  
Lemma \ref{flagekl}  treats $ 1-7/(15Q)  \le P \le d-1.$ 
Lemma \ref{pukl} together with Lemma \ref{einmi} treat $ 1-7/(15Q)\ge P \ge 1/(2Q).$ 
Finally Lemma \ref{stgekl}  treats $1/(2Q) \ge P > 0.$ Observe that OPT$(0, 0, s)=4$ and
we need to look into the proof of Lemma \ref{lagrkl} to get the  $-\delta$ required
for small $P.$ 
 \qed

\subsection{Proof  of Lemma \ref{flagekl}} 
\noindent
{\bf Lemma \ref{flagekl} \quad (repeated)} 
{\it \begin{eqnarray}
\mbox{ Let } s \ge 7 \mbox{ and let }\,A(c)\,= \, A(c, s) \,\,=\,\, \frac{7}{Q\cdot 10\cdot 3}\cdot c \,\,+1 \, \, - \frac{7}{10 \cdot Q}. \nonumber \\
\mbox{Then } \, 
 \mbox{OPT}(A(c), \, c, \, s\,) \,\mbox{ is strictly increasing in  } \, \,  1 \le c <3. \nonumber 
\end{eqnarray} }

\begin{proof} 
\begin{eqnarray}
\mbox{ Some  notation:} \, \mbox{PPLUS3}(x, y)\,= \, (1\,+ \, x)^y \,+ \, 3\left(1-\frac{x}{3}\right)^y, \,\, x \le 3 \nonumber \\
\mbox{PMINUS}(x, y) \,= \,(1\,+ \, x)^y \,-\, \left(1-\frac{x}{3}\right)^y, \,\, x \le 3 \label{notPPPM}    \\
\frac{d}{dc}\ln \mbox{PPLUS3}(c, Q)\,= \, Q \cdot \frac{\mbox{PMINUS}(c, Q-1)}{\mbox{PPLUS3}(c, Q)}, \, \nonumber \\
A\, \,=\,\, A(c), \, \,   A' \,  =\, \frac{d}{dc}\,A\,= \frac{7}{10\cdot 3 \cdot Q}  \, \,\, , \,  \, \, 
\frac{d}{dc}\ln \mbox{OPT}(A, c, s)\, = \, \nonumber \\
\,=\,  
\frac{3\frac{\exp(As)-1)}{\exp(s)-s-1} \cdot s A'}{\mbox{OPT}_1(A,  s)}\,- \,
Q\cdot \frac{A'c\,+ \,A}{1+Ac}\,+ \, Q \cdot  \frac{\mbox{PMINUS}(c, Q-1)}{\mbox{PPLUS3}(c, Q)} \, >=<0 \nonumber \\
\Longleftrightarrow \, \frac{3\frac{\exp(As)-1}{\exp(s)-1}\cdot A'}{\mbox{OPT}_1(A, s)}\, - \, \frac{A+A'c}{1+Ac}\,+ \, 
 \frac{\mbox{PMINUS}(c, Q-1)}{\mbox{PPLUS3}(c, Q)} \, >=<\,0 \,\,\,\,\,\;\;\, \, \,  \label{flagekl1} \\
\mbox{(Division with } Q\,= \, \frac{s(\exp(s)-1)}{\exp(s)-s-1} \mbox{.)} \nonumber \\
\mbox{ For } c\,= \,3 \mbox{ the derivative is } =0, \,\, A(3, s)=1\, \,  \mbox{( OPT}(1, 3, s)   =4.) \nonumber 
\end{eqnarray}

We split the right-hand-side of (\ref{flagekl1}) into two additive terms. Inequalities (\ref{flagekl2}) and 
(\ref{flagekl3}) imply that the $\frac{d}{dc} $ OPT$(A, c, s)>0.$  For $c=3$ both left-hand-sides are $=0.$ 
\begin{eqnarray}
\frac{3\frac{\exp(As)-1)}{\exp(s)-1} \cdot A'}{\mbox{OPT}_1(A,  s)}\,- \,
 \frac{A'c}{1+Ac}\, > \, 0 \label{flagekl2} \\
- \,   \frac{A}{1+Ac}\,+ \,    \frac{\mbox{PMINUS}(c, Q-1)}{\mbox{PPLUS3}(c, Q)} \,>\,0  \label{flagekl3}
\end{eqnarray}

\noindent
{\it Proof of (\ref{flagekl2}) for $0 \le c <3 \, $ and $ s \ge 7 .$ } 
\begin{eqnarray}
K\,:= \, \frac{\exp(sA)-1}{\exp(s)-1}, \,\, \,  L:= \frac{\exp(sA)-sA -1}{\exp(s)-s-1} \nonumber \\
\mbox{ We need to show } \frac{3KA'}{3L+1}\, \, > \, \frac{A'c}{1+Ac} \, \, 
\Longleftrightarrow \, \, \, \frac{3K}{3L+1}\, \, > \, \frac{c}{1+Ac}  \nonumber \\
\Longleftrightarrow \, \, 3\left( K+ KAc -Lc \right) \,> \,c \, \, \, \, 
\mbox{ (For } c= 3 \mbox{ both sides are } \,=3 .) \label{flagekl21} 
\end{eqnarray}
By (\ref{BAKL}) we have $ L \, \le \,  K$ and   (\ref{flagekl21}) is implied by 
\begin{eqnarray}
3K\left( 1+ Ac -c \right) \,> \,c \, \, \quad \mbox{ (For } c=3 \mbox{ both sides are  } =3.\mbox{)} \label{flagekl22}
\end{eqnarray}
$K \ge 0 $ is increasing and convex,   $1+ Ac -c $ is $>0,$ and increasing for $c>3/2,$  and convex. Therefore the 
left-hand-side of  (\ref{flagekl22}) is convex for $c\ge 3/2.$ Therefore, for $ c\ge 3/2,$  it follows from 
\begin{eqnarray}
\left(\frac{d}{dc} 3K\left( 1+ Ac -c \right) \right)_{|c=3}\, < \,\left( \frac{d}{dc}c \right)_{|c=3}\,= \,1. \label{flagekl23}\\
\frac{d}{dc} 3K\left( 1+ Ac -c \right)\,\, \, \, \, \, 
\,= \, \frac{3\exp(sA)\cdot s \cdot \frac{7}{30Q}}{\exp(s)-1}\cdot(1+Ac-c) \,+ \, \nonumber \\
\frac{3 ( \exp (sA)-1)}{\exp(s)-1}\cdot  \left( \frac{7}{30Q}c + \frac{7}{30Q}c +1- \frac{7}{10Q} -1\right). \nonumber \\
\mbox{ Therefore } \, \frac{d}{dc} 3K\left( 1+ Ac -c \right)_{|c=3}\,= \, \nonumber \\
\frac{7\exp(s)(\exp(s)-s-1)}{10(\exp(s)-1)^2}\,\,+ \, \, \frac{21(\exp(s)-s-1)}{10s(\exp(s)-1)}\, \, \, \, \label{flagekl24}
\end{eqnarray}
For $s=7$ we get that (\ref{flagekl24}) is $<0.995.$ 
Moreover  it is decreasing in $s$ (proof omitted) and (\ref{flagekl23}) holds for all $s \ge 7$ and $c \ge 3/2.$  
For $c\le 3/2 $ we argue as in the proof of Lemma \ref{lem4}(a) cf. the argument following (\ref{grosymugl12}). \\

\noindent
{\it Proof of (\ref{flagekl3}) for $1\, \le c \,<  \,3 \,$ and $ s \ge 5 $ .}
\begin{eqnarray} 
\mbox{ We need to show } \frac{A}{1+Ac}\,\,< \, \frac{\mbox{PMINUS}(c, Q-1)}{\mbox{PPLUS3}(c, Q)} \, \nonumber \\
\Longleftrightarrow  A \, \cdot \, \mbox{PPLUS3}(c, Q)\,< \, (1+Ac) \mbox{PMINUS}(c, Q-1) \nonumber \\
\Longleftrightarrow  A\cdot\left(\mbox{PPLUS3}(c, Q) \,- \, c \cdot \mbox{PMINUS}(c, Q-1)\right) 
\,= \, A\cdot \mbox{PPLUS3}(c, Q-1)\,\nonumber \\ 
< \, \mbox{PMINUS}(c, Q-1) \nonumber \\
\Longleftrightarrow A \,< \, \frac{\mbox{PMINUS}(c, Q-1) }{\mbox{PPLUS3}(c, Q-1)}\,=\, \frac{(1+c)^{Q-1}- (1-\frac{c}{3})^{Q-1}}{(1+c)^{Q-1}+ 3 (1-\frac{c}{3})^{Q-1}}\;\;\;\;\label{flagekl31} \\
\mbox{(For } c=3 \mbox{ both sides of (\ref{flagekl31}) are  }= 1.) \nonumber 
\end{eqnarray}

 For $ c=1 $  inequality ( \ref{flagekl31}) becomes  
\begin{eqnarray}
1-\frac{7}{15Q} \,< \, 
\frac{2^{Q-1}- \left(\frac{2}{3} \right)^{Q-1}}{2^{Q-1}+ 3 \left( \frac{2}{3} \right)^{Q-1}} \,= \, 1- \frac{4\left( \frac{1}{3} \right)^{Q-1}}{1+3\left(\frac{1}{3}\right)^{Q-1}} \nonumber \\
\mbox{As }\, 4\left( \frac{1}{3} \right)^{Q-1}\,< \, \frac{7}{15Q} \mbox{ for } \, \,  Q \ge 5 ,\, \, 
\mbox{ (\ref{flagekl31}) holds for }
 c=1 \mbox{ and } s \ge 5 \, \mbox{ as } Q \ge s. \nonumber  
\end{eqnarray}

To show that (\ref{flagekl31}) holds 
for all $3\,> \,c \ge\, \, 1$ we show 
that the right-hand-side is concave for $c>1.$   
\begin{eqnarray}
\mbox{ Numerator  of } \frac{d}{dc} \frac{\mbox{PMINUS}(c, Q-1) }{\mbox{PPLUS3}(c, Q-1)} \,= \, \nonumber \\
(Q-1)\cdot \left((1+c)^{Q-2} \,+ \, \frac{1}{3} \left( 1- \frac{c}{3}\right)^{Q-2}\right)\cdot 
\left((1+c)^{Q-1} \,+ \, 3 \left( 1- \frac{c}{3}\right)^{Q-1}\right) \,- \, \nonumber \\
\,- \, (Q-1)\cdot \left((1+c)^{Q-1} \,-\,  \left( 1- \frac{c}{3}\right)^{Q-1}\right)\cdot 
\left((1+c)^{Q-2} \,- \,  \left( 1- \frac{c}{3}\right)^{Q-2}\right) \ \,\nonumber \\
\,= \, (Q-1) \cdot  (1+c)^{Q-2}\cdot \left( 1 - \frac{c}{3} \right)^{Q-2}\cdot \left(\left( \frac{1}{3}+1\right)(1+c)
\,+ \, \left(1- \frac{c}{3}\right)\cdot (1+3) \right) \, \, \nonumber \\
\,= \, (Q-1)\cdot(1+c)^{Q-2}\cdot \left( 1 - \frac{c}{3} \right)^{Q-2} \cdot \left( \frac{1}{3}+3+2\right) \;\;\;\;\;\label{flagekl32} \\
\mbox{ We have that  }(1+c)\cdot \left( 1 - \frac{c}{3}\right) 
\mbox{ is decreasing  for }  c > 1, 
\mbox{ and PPLUS3}(c, Q-1) \mbox{ is increasing } . \nonumber  \\
\mbox{Therefore the right-hand-side of
(\ref{flagekl31}) is concave.  } \nonumber  
\end{eqnarray}

\end{proof}

\newpage  
\subsection{ Proof of Lemma \ref{stgekl}}
\quad \\
\quad \\
\noindent
{\bf Lemma \ref{stgekl} \quad (repeated)}
\begin{eqnarray}
\mbox{ Let } s\ge 6 \mbox{ and let }\, A (c)\,\,=\,\, A(c, s) \,\,=\,\, \frac{Q}{2}\cdot c\, \,. \nonumber \\ 
\mbox{ Then }\, \mbox{OPT}(A(c), \, c, \, s) \, \mbox{  is strictly decreasing for }  0\, <c\, \le \, \frac{1}{Q} \nonumber 
\end{eqnarray}

\begin{proof} 
Analogously to (\ref{flagekl2}) and (\ref{flagekl3}) this follows from  (\ref{stegekl1}) and (\ref{stegekl2}.) 
(Notation cf. ( \ref{notPPPM}.)
\begin{eqnarray}
 \mbox{ with } A=A(c),\, \, A' \,= \,\, Q/2 \, \,\frac{3\frac{\exp(As)-1)}{\exp(s)-1}}{\mbox{OPT}_1(A,\, s)}A'\,- \, \frac{A'c}{1+Ac}\,\,< \, \,0, \, 
  \label{stegekl1} \\
- \,   \frac{A}{1+Ac}\,+ \,    \frac{\mbox{PMINUS}(c, Q-1)}{\mbox{PPLUS3}(c, Q)} \,<\,0  \label{stegekl2}.
\end{eqnarray}

\noindent
{\it Proof of (\ref{stegekl1}) for $ s \ge 5.55 \, \,$ and $  0 \,<\, c \, \le \,1/Q$ }
\begin{eqnarray}
K\,:= \, \frac{\exp(sA)-1}{\exp(s)-1}, \,\, \,  L:= \frac{\exp(sA)-sA -1}{\exp(s)-s-1} \nonumber \\
\mbox{ We need to show } \frac{3KA'}{3L+1}\, \, < \, \frac{A'c}{1+Ac} \, \, 
\Longleftrightarrow \, \, \frac{3K}{3L+1}\, \, < \, \frac{c}{1+Ac}  \nonumber \\
\Longleftrightarrow \, \, 3\left( K+ KAc -Lc \right) \,< \,c 
\mbox{ For } c=0 \mbox{ both sides  are } \, =0 .\label{stegekl11} 
\end{eqnarray} 
As $AK\le L$ by (\ref{BAKL}) we get that  (\ref{stegekl11}) is implied by $3 \cdot K < c.$ 
For $c=0$ both sides of  $3 \cdot K < c$  are $0.$ The left-hand-side is convex. 
It is sufficient to show  $3 \cdot K < c.$
for $c=1/Q.$ Plugging in the definition of $1/Q$  for $c$ and $ A(1/Q \, , s)\,=\, 1/2 $ into $K$ we need to show
\begin{eqnarray}
\frac{3\exp(s/2)-1}{\exp(s)-1}< \frac{\exp(s)-s-1}{ s (\exp(s)-1)} \, 
\Longleftrightarrow \, 3 s\exp(s/2) \, < \,\exp(s)-1  \nonumber 
\end{eqnarray}
For $s\ge 6$ the preceding inequality holds by simple consideration. \\

\noindent
{\it Proof of (\ref{stegekl2}) for $s\ge 2 \, \, $ and $c \le 1/Q$}
Analogously to the proof of (\ref{flagekl31}) we need to show 
\begin{eqnarray}
A\,= \, \frac{Q}{2}c \,> \, \frac{\mbox{PMINUS}(c, Q-1) }{\mbox{PPLUS3}(c, Q-1)}\,=\, 
\frac{(1+c)^{Q-1}- (1-\frac{c}{3})^{Q-1}}{(1+c)^{Q-1}+ 3 (1-\frac{c}{3})^{Q-1}}\;\;\;\;\label{stgekl21} \\
\mbox{ For } c=0 \mbox{ both sides of the preceding inequality are } = 0 \nonumber 
\end{eqnarray}
We show,  that $A'\,>$ the derivative wrt. $c$ of the right-hand-side of (\ref{stgekl21}). 
Using (\ref{flagekl32}) we need to show

\begin{eqnarray}
\frac{Q}{2} \cdot \left( (1+c)^{Q-1} + 3\left(1-\frac{c}{3}\right)^{Q-1}\right)^2\, \, > \, \, 
(Q-1)\cdot (1+c)^{Q-2}\cdot \left( 1- \frac{c}{3} \right)^{Q-2}\cdot \frac{16}{3}. \nonumber \\
\mbox{Note }\, Q\cdot (1+c)^{Q-1}\cdot \left( 1- \frac{c}{3} \right)^{Q-1}\cdot \frac{16}{3}
\ge (Q-1)\cdot (1+c)^{Q-2}\cdot \left( 1- \frac{c}{3} \right)^{Q-2}\cdot \frac{16}{3} \nonumber \\
\mbox{ as } \, \, \,  (1+c) \cdot \left(1-\frac{c}{3}\right) \ge 1 \mbox{ for } 0 \le c \le 1/Q < 2 .\nonumber \\
\mbox{Enlarging the right-hand-side  it is sufficient to show } \nonumber \\
3\left( (1+c)^{Q-1} + 3\left(1-\frac{c}{3}\right)^{Q-1}\right)^2\,>\, 
 (1+c)^{Q-1}\cdot \left( 1- \frac{c}{3} \right)^{Q-1}\cdot 32 . \nonumber \\
\Longleftrightarrow 3\left(1+ 3\left(\frac{1-\frac{c}{3}}{1+c}\right)^{Q-1}\right)^2\,> \, 32\cdot \left(\frac{1-\frac{c}{3}}{1+c}\right)^{Q-1} \nonumber 
\end{eqnarray}
Setting $x= \left(\frac{1- \frac{c}{3}}{1+c}\right)^{Q-1}$ it is easy to see 
that the preceding inequality holds for $x \ge 0,$
and therefore clearly for $ c \le 1/Q < 3  .$
\end{proof}

\subsection{Proof of Lemma \ref{pukl} and  Lemma \ref{einmi} }
Lemma \ref{einmi} follows by elementary consideration, see the analogous situation in the 
proof of Lemma \ref{lem2}(a) and Lemma \ref{lem3} (a). 

\noindent
{\bf Lemma \ref{pukl} \quad (repeated)} 
 Let $ s\, \ge 6\,$   then OPT$(a, \, \, c, \, \, s\,) \,< \, 4-\delta $   for $   (a, \,  c)\,= \, $
\begin{eqnarray}
 =\,\left(\frac{1}{2}, \, \, \frac{1}{Q} \, \right)\, ,\, \left(\frac{1}{2}, \,  \frac{2}{Q}\,\right),\, 
\,\left(\frac{2}{3}, \, \frac{2}{Q} \,\right) , \, \left(\frac{2}{3}, \, \frac{3}{Q} \,\right), \, 
\left( 1\, \, - \, \,  \frac{7}{15Q}\, , \,  \frac{3}{Q} \right),  \, \, 
 \, \left(1\, - \, \frac{7}{15Q}\, \, ,  \, 1 \,\right)  \nonumber 
\end{eqnarray}
\begin{proof}
The claim for  $a = \frac{1}{2} , \, \, c\,= \, \frac{1}{Q}$  is included in  Lemma \ref{stgekl}. 

\begin{eqnarray}
\mbox{FIRSUM}(a, c, s)\,= \, (1+c)^Q\mbox{OPT}_2(a, c, s)\,\,= \, \, \left(\frac{1+c}{1+ac}\right)^Q \nonumber \\
\mbox{SECSUM}(a, c, s)\,= \, 3\left(1- \frac{c}{3}\right)^Q \mbox{OPT}_2(a, c, s)\,= \, 
3\left(\frac{1-\frac{c}{3}}{1+ac}\right)^Q
\mbox{ then }\nonumber \\
\mbox{OPT}(a, c, s) \,= \, \mbox{OPT}_1(a, s)\cdot \left[ \mbox{FIRSUM}(a, c, s)\,+ \, \,\mbox{SECSUM}(a, c, s)\right].  \nonumber \\
\mbox{For } x , y\ge 0 \, \mbox{  we have } \mbox{FIRSUM}\left(x, \frac{y}{Q}, s \right) \, =\,\,
\left(\frac{1+ \frac{y}{Q}}{1+x\cdot \frac{y}{Q}}\right)^Q \,= \, \, \, \nonumber \\
=\, \, \left( 1 \, + \, \frac{\frac{y}{Q}(1-x)}{1+ x\frac{y}{Q}}\right)^Q \, \,
\le \, \exp\left( \frac{y(1-x)}{1+x\frac{y}{Q}}\right) \le \,\exp(y(1-x)) \label{bFIR}  \\
\mbox{ We have that  OPT}_1(a, s) \mbox{is decreasing in } s \mbox{ for constant } a < 1. \label{OPT1dec}
\end{eqnarray}

\begin{eqnarray}
\mbox{ Let } \, a=\frac{1}{2}, \, \, c\,= \, \frac{2}{Q}. \nonumber \\
\mbox{ We have by (\ref{bFIR}) } 
\mbox{FIRSUM}(a, c, s) \, < \,\exp(1)  \nonumber \\
\mbox{SECSUM} (a, c, s) \mbox{ is decreasing in } s\ge 0. \nonumber \\
\mbox{(As can be shown by elementary means.)} \, \, \nonumber \\
\mbox{OPT}_1(a, s) \left( \mbox{SECSUM}(a, c, s)+ \exp(1) \right) \,< \,3.913 \,\mbox{ for } s=5 \nonumber \\
\mbox{ and decreasing in } s \mbox{ with  (\ref{OPT1dec})} \nonumber 
\end{eqnarray}

\begin{eqnarray}
\mbox{ Let } \, a=\frac{2}{3}, \, \, c\,= \, \frac{2}{Q}. \nonumber \\ 
\mbox{ We have }  \, \mbox{FIRSUM}(a, c, s) \, < \,\exp(2/3)  \nonumber \\
\mbox{SECSUM} (a, c, s) \mbox{ is decreasing in } s\ge 0.  \, \, \nonumber \\
\mbox{OPT}_1(a, s) \left( \mbox{SECSUM}(a, c, s)+ \exp(2/3) \right) \,< \,3.962 \,\mbox{ for } s=4 \nonumber \\
\mbox{ and decreasing in } s \mbox{ with  (\ref{OPT1dec})} \nonumber 
\end{eqnarray}

\begin{eqnarray}
\mbox{ Let } \, a=\frac{2}{3}, \, \, c\,= \, \frac{3}{Q}. \mbox{ We have } 
\mbox{FIRSUM}(a, c, s) \, < \,\exp(1)  \nonumber \\
\mbox{SECSUM} (a, c, s) \mbox{ is decreasing in } s\ge 2.  \, \, \nonumber \\
\mbox{OPT}_1(a, s) \left( \mbox{SECSUM}(a, c, s)+ \exp(1) \right) \,< \,3.985 \,\mbox{ for } s=6 \nonumber \\
\mbox{ and decreasing in } s \mbox{ with  (\ref{OPT1dec})} \nonumber 
\end{eqnarray}

\begin{eqnarray}
 \mbox{ Let } a=1-7/(15Q), \, c=3/Q .\nonumber \\
\mbox{OPT}_1(a, s) \mbox{  is increasing in } s  \mbox{ to } 3\exp(-7/15)+1. \nonumber \\
\mbox{FIRSUM }(a, c, s), \, \mbox{SECSUM}(a, c, s) \mbox{are both  decreasing in } s. \nonumber \\
(3\exp(-7/15)+1)\left( \mbox{SECSUM}(a, c, s)+ \mbox{FIRSUM}(a, c, s)  \right)
 \,< \,3.9 \mbox{ for } s=4 \nonumber 
\end{eqnarray}

The case $a=1\, - \, \frac{7}{15Q}$ and $c=1\,$ is included in Lemma \ref{flagekl}.  

\end{proof}

\section{Proof of Theorem \ref{UNOPT}   for $d=4,\, \lambda \ge 1-1/d, s \ge 5.$ } \label{LARGE}

We fix $a=1.$  Observe that $B(1/c)$ in the subsequent lemma goes from $1$ to $1-1/(2Q)$ for $c\ge 3.$ 
\begin{lem} \label{lagr}
Let $ B(x) \, = \,  B(x, s) \, = \, 1+3/(2Q)x-1/(2Q). $ Then OPT$(1, B(1/c), c, s) $ is strictly  
decreasing in $c \ge 3.$ 
\end{lem}

\noindent
{\it Proof of Theorem \ref{UNOPT} for $\lambda\ge 1-1/d.$}
We have $\lambda/(1-\lambda) \ge d-1.$ For each $P \ge d-1$ we have 
$c$ such that $P \, =\, \frac{c}{B(1/c)}.$ As OPT$(1, 1, 3, s)=4$ the Theorem follows. \qed

\noindent
{\it Proof of Lemma \ref{lagr}.}
We  rewrite OPT$(1, B(1/c), c, s)$ first.  
We multiply OPT$_2$  with $c^Q$ and OPT$_3$ with $1/c^Q$ and get (using  $c\ge 3$  to get rid of the absolute value)    
OPT$(1, B(1/c, s), c, s) \,= \,$
\begin{eqnarray}
=\,  \left( 3+ \frac{q(B(1/c, s) s)}{q(s)}\right)\left(\frac{1}{\frac{B(1/c, s)}{c}+1}\right)^Q \left(
 \left(\frac{1}{c}+1\right)^Q +
3\left( \frac{1}{3}- \frac{1}{c}\right)^Q  \right) .\nonumber \\
\mbox{ We substitute } c \mbox{ for } 1/c \mbox{ in the preceding equation.  The claim   follows from  } 
\nonumber \\
\left( 3+ \frac{q(B(c))s)}{q(s)}\right)\left(\frac{1}{B(c)c+1}\right)^Q \left(
 \left(c+1\right)^Q +
3\left(\frac{1}{3}-c\right)^Q  \right) \nonumber \\
 \mbox{increases in } 0\, <\, c \,<\, 1/3. \label{lagro}
\end{eqnarray}

We use the following notation in the sequel:
\begin{eqnarray}
    \OPT_1(b,s)  =  3 + \frac{q(sb)}{q(s)}\, , \, \OPT_2(b,c,s) \, = \,\left( \frac{1}{bc+1} \right)^Q ,          \nonumber           \\
    \OPT_3(c,s)  =  (c+1)^Q+\,3\left(\frac{1}{3}-c\right)^Q , \, c \le \frac{1}{3} \nonumber \\
\OPT(b, c, s)= \OPT_1(b, s)\OPT_2(b, c, s)\OPT_3(c, s). \nonumber \\
\mbox{ For } b=1, c= \frac{1}{3} \mbox{ we have }  \OPT(b, c, s)=4.  \mbox{ We abbreviate  } \nonumber \\
\mbox{PM}(x, y)= (x+1)^y- 3\left(\frac{1}{3}-x\right)^y , \, x \le \frac{1}{3} \nonumber \\
\mbox{PP}(x, y)= (x+1)^y+ 3\left(\frac{1}{3}-x\right)^y, \, x \le \frac{1}{3} \nonumber \\
B\,= B(c, s), \quad  B' = \diff{c}B = \frac{3}{2Q}, \quad q'(x) = \exp(x)-1  \nonumber  \\
q(x)=\exp(x)-x-1, \, \, \, \, \, \diff{c}\ln(\OPT(B,c,s) > = <0  \nonumber \\
\Longleftrightarrow \frac{\frac{B'sq'(sB)}{q(s)}}{3+\frac{q(sB)}{q(s)}}\,-Q\frac{B'c+B}{1+Bc}
 + Q\frac{\mbox{PM}(c, Q-1)}{\mbox{PP}(c, Q)} > 0  \nonumber \\
\Longleftrightarrow \frac{\frac{B'q'(sB)}{q'(s)}}{3+\frac{q(sB)}{q(s)}} - \frac{B'c+B}{1+Bc}
   + \frac{\mbox{PM}(c, Q-1)}{\mbox{PP}(c, Q)}  >=< 0. \mbox{( Division by } Q.) 
\label{gro_ableit}
\end{eqnarray}

\begin{figure}[th]
\centering
\includegraphics[width=0.49\textwidth]{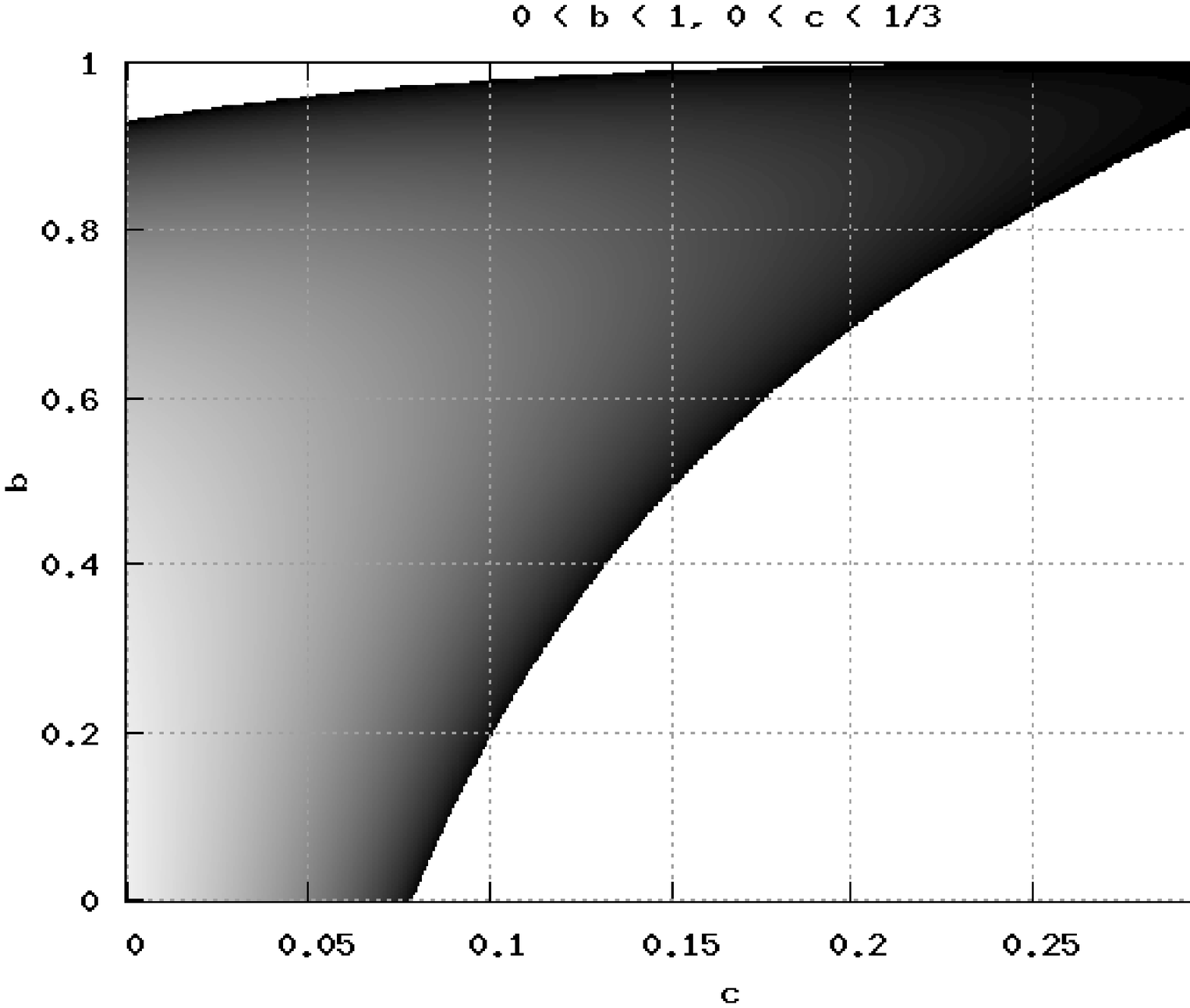} \hfill
\includegraphics[width=0.49\textwidth]{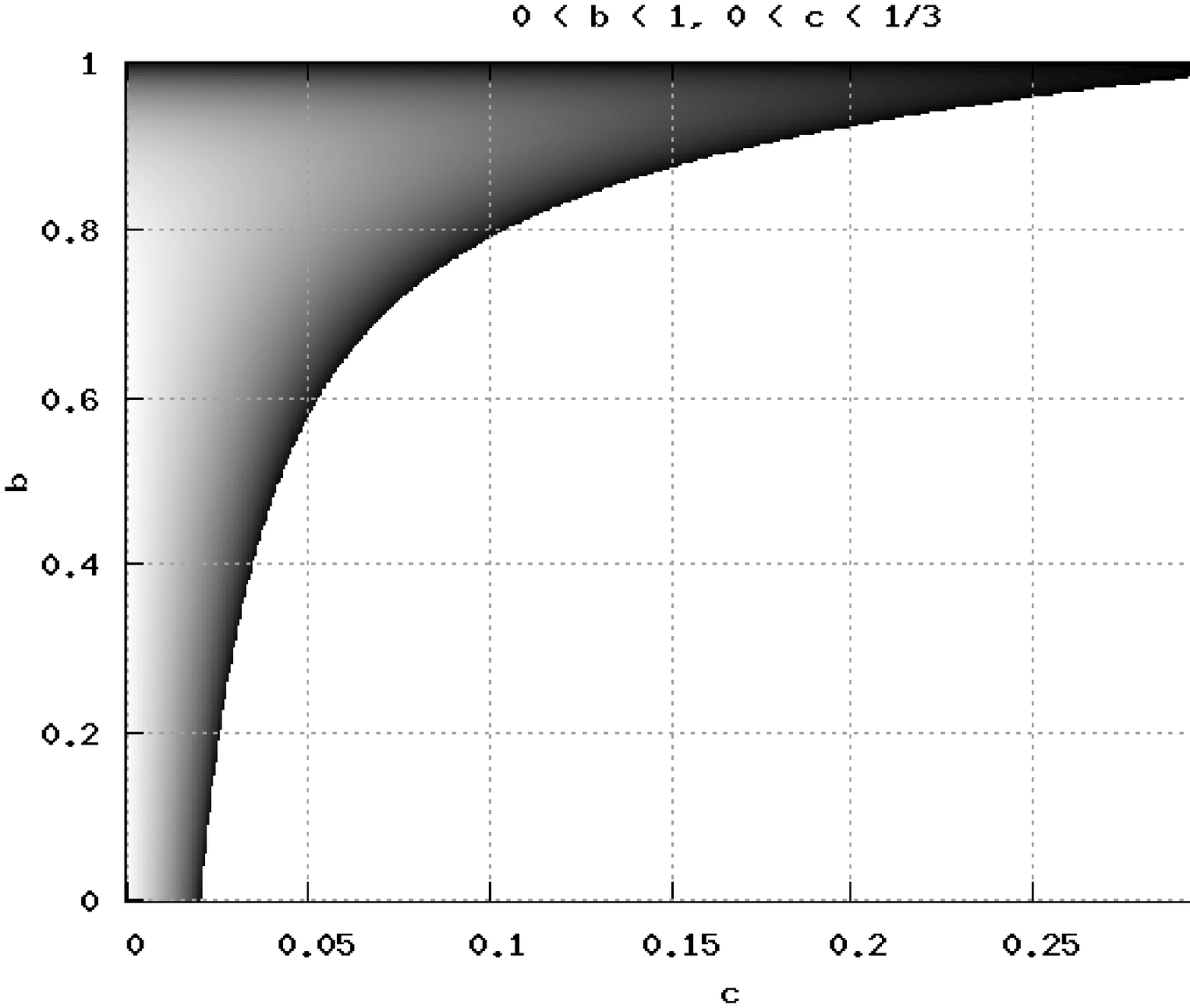}
\caption{OPT$(b,c,s)$ over the rectangle $0\le b \le 1, 0 \le c \le 1/3$ for $s=3$ and $s=14$.}
\end{figure}

For $c=\frac{1}{3}$  we have $\OPT(B, 1/3, s)=4$, and the derivative is $0$.
We split (\ref{gro_ableit}) into two additive terms. 
The following two inequalities
directly imply (\ref{lagro}.) 

\begin{eqnarray}
\frac{\frac{B'q'(sB)}{q'(s)}}{3+\frac{q(sB)}{q(s)}} - \frac{B'c}{1+Bc} & > & 0 \label{gro_teil1}\\
 \frac{\mbox{PM}(c, Q-1)}{\mbox{PP}(c, Q)}- \frac{B}{1+Bc} & > & 0 \label{gro_teil2}
\end{eqnarray}

\noindent
{\it Proof of (\ref{gro_teil1}) for $\, s>2 \, .$}
Let $K = \frac{q'(sB)}{q'(s)}$ and $L= \frac{q(sB)}{q(s)}.$ By (\ref{BAKL}) we have $L\le K,$  
and as $B' > 0$  it is sufficient to show
\begin{equation}
\frac{K}{3+K} > \frac{c}{1+Bc} \Leftrightarrow K(1+Bc-c) > 3c. \nonumber 
\end{equation}

For $c=\frac{1}{3}$ both sides of the preceding inequality are $=1$. 
It is easy to observe that   $K(1+Bc-c)$ is convex in $c$  for $c>1/(3\cdot 2).$ 
(Cf. proof of Lemma \ref{lem4}) and
$3c$ is a linear function. If at $c=\frac{1}{3}$ the derivative of $3c$ is greater than the
derivate of $K(1+Bc-c)$ the second intersection of both sides (if any) lies at some point
$c>\frac{1}{3}$ and the claim holds for $1/(3\cdot 2<c<\frac{1}{3}$. For $c<1/(3\cdot 2) $ we argue as in the
proofs of the  Lemmas mentioned above. 
Therefore it is sufficient to show that at $c=\frac{1}{3}$
\begin{equation}
\diff{c}K(1+Bc-c) < \diff{c}3c. \nonumber 
\end{equation}
We have 
\[ K' = \frac{B's\exp(sB)}{\exp(s)-1}. \] 
and at $c=1/3$
\begin{eqnarray}
K'(\underbrace{1+Bc-c}_{=1}) + \underbrace{K}_{=1}(\underbrace{B'c+B-1}_{=1/2Q}) & < & 3   \nonumber \\
\Leftrightarrow \frac{3(\exp(s)-s-1)\exp(s)}{2(\exp(s)-1)^2}+\frac{\exp(s)-s-1}{2s(\exp(s-1)} & < & 3 \label{last}
\end{eqnarray}
We omit the proof that inequality (\ref{last}) holds for $s\ge 2.$ \\

\noindent
{\it Proof of (\ref{gro_teil2}) for $s\ge 5$.}
As in (\ref{flagekl31}) inequality (\ref{gro_teil2}) is equivalent to 
\begin{equation}
B < \frac{\mbox{PM}(c,Q-1)}{\mbox{PP}(c,Q-1}     \label{gro_teil2_umform}
\end{equation}

The left hand side is a linear function in c and the right 
hand side a strictly increasing, concave
function in $c.$
For $c=\frac{1}{3}$ both sides of (\ref{gro_teil2_umform}) are $1$. So we must show that (\ref{gro_teil2_umform}) holds for $c=0$.
Setting $c=0$ leads to
\begin{equation}
1-\frac{1}{2Q} < \frac{1-3\left( \frac{1}{3} \right)^{Q-1}}{1+3\left( \frac{1}{3} \right)^{Q-1}}
    = \frac{1- \left(\frac{1}{3}\right)^{Q-2}}{1+ \left(\frac{1}{3}\right)^{Q-2}}. \nonumber 
\end{equation}
For $Q=5$ we get $\frac{9}{10} < \frac{13}{14}$. We omit the argument that 
the last inequality holds for all $Q\ge 5$ and therefore as $Q>s$ for all $s \ge5.$

\end{document}